    \documentclass[12pt,a4paper]{article}
    \usepackage[cp1251]{inputenc}
    \usepackage[english]{babel}
\usepackage{algorithm}
\usepackage{multirow}
\usepackage[noend]{algpseudocode}    
    \usepackage{amsmath,amssymb,amsthm}
\usepackage{tikz}
\usepackage{xypic}
\usetikzlibrary{shapes}
     \usepackage{graphicx}
     \usepackage{wrapfig}
\renewcommand*\circled[1]{\tikz[baseline=(char.base)]{
            \node[shape=circle,draw,inner sep=2pt] (char) {#1};}}

  \DeclareMathOperator\CSP{CSP}

\DeclareMathOperator{\proj}{pr}
  \DeclareMathOperator\Inv{Inv}
  \DeclareMathOperator\Pol{Pol}
  \DeclareMathOperator\NP{NP}
  \DeclareMathOperator\coNP{co-NP}
  
  \DeclareMathOperator\DP{DP}
  \DeclareMathOperator\Clo{Clo}
  \DeclareMathOperator\BoolClo{BoolClo}
  
  \DeclareMathOperator\QCSP{QCSP}

  \DeclareMathOperator{\Var}{Var}

\renewcommand{\le}{\leqslant}
\renewcommand{\ge}{\geqslant}

\theoremstyle{definition}
\theoremstyle{plain}
\newtheorem{thm}{Theorem}
\newtheorem{conj}{Conjecture}
\newtheorem{lem}[thm]{Lemma}

\newtheorem{cor}[thm]{Corollary}

\newtheorem{definition}{Definition}

\newtheorem*{THMThreeElementClassification}{Theorem~\ref{ThreeElementClassification}}

\newtheorem*{THMconservative}{Theorem~\ref{thm:conservative}}

\begin{document}

\title{QCSP monsters and the demise of the Chen Conjecture}
\author{Dmitriy Zhuk\thanks{The author has received funding from the European Research Council (ERC) under the European Unions Horizon 2020 research and innovation programme (grant
agreement No 771005) and from Russian Foundation for Basic Research (grant 19-01-00200).} \and Barnaby Martin}

\maketitle
\begin{abstract}
We give a surprising classification for the computational complexity of the Quantified Constraint Satisfaction Problem
over a constraint language $\Gamma$, QCSP$(\Gamma)$, where $\Gamma$ is a finite language over $3$ elements which contains all constants. In particular, such problems are either in P, NP-complete, co-NP-complete or PSpace-complete. Our classification refutes the hitherto widely-believed Chen Conjecture.

Additionally, we show that already on a 4-element domain 
there exists a constraint language $\Gamma$ such that 
$\QCSP(\Gamma)$ is DP-complete (from Boolean Hierarchy), and on a 10-element domain 
there exists a constraint language giving the complexity class
$\Theta_{2}^{P}$.

Meanwhile, we prove the Chen Conjecture for finite conservative languages $\Gamma$. If the polymorphism clone of such $\Gamma$ has the polynomially generated powers (PGP) property then QCSP$(\Gamma)$ is in NP. Otherwise, the polymorphism clone of $\Gamma$ has the exponentially generated powers (EGP) property and QCSP$(\Gamma)$ is PSpace-complete. 
\end{abstract}

\section{Introduction}

The \emph{Quantified Constraint Satisfaction Problem} $\QCSP(\Gamma)$ is the generalization of the \emph{Constraint Satisfaction Problem} $\CSP(\Gamma)$ which, given the latter in its logical form, augments its native existential quantification with universal quantification. That is, $\QCSP(\Gamma)$ is the problem to evaluate a sentence of the form $\forall x_1 \exists y_1 \ldots \forall x_n \exists y_n \ \Phi$, where $\Phi$ is a conjunction of relations from the \emph{constraint language} $\Gamma$, all over the same finite domain $D$.
Since the resolution of the Feder-Vardi ``Dichotomy'' Conjecture, classifying the complexity of $\CSP(\Gamma)$, for all finite  $\Gamma$, between P and NP-complete \cite{BulatovFVConjecture,ZhukFVConjecture}, a desire has been building for a classification for $\QCSP(\Gamma)$. Indeed, since the  classification of the \emph{Valued CSPs} was reduced to that for CSPs \cite{KolmogorovKR17}, the QCSP remains the last of the older variants of the CSP to have been systematically studied but not classified. More recently, other interesting open classification questions have appeared such as that for \emph{Promise CSPs} \cite{BrakensiekG18} and finitely-bounded, homogeneous infinite-domain CSPs \cite{BartoP16}. The exact dates of these classification questions are not exactly cast in stone but the question for QCSPs dates at least to the announcement of the classification for the 2-element case in \cite{Schaefer}.

While $\CSP(\Gamma)$ remains in NP for any finite $\Gamma$, $\QCSP(\Gamma)$ can be PSpace-complete, as witnessed by \emph{Quantified 3-Satisfiability} or \emph{Quantified Graph 3-Colouring} (see \cite{BBCJK}). It is well-known that the complexity classification for QCSPs embeds the classification for CSPs: if $\Gamma+1$ is $\Gamma$ with the addition of a new isolated element not appearing in any relations, then $\CSP(\Gamma)$ and $\QCSP(\Gamma+1)$ are polynomially equivalent. Thus, and similarly to the Valued CSPs, the CSP classification will play a part in the QCSP classification. It is now clear that $\QCSP(\Gamma)$ can achieve each of the complexities P, NP-complete and PSpace-complete. It has thus far been believed these were the only possibilities (see \cite{BBCJK,hubie-sicomp,HubieSIGACT,Meditations,QC2017} and indeed all previous papers on the topic).

A key role in classifying many CSP variants has been played by Universal Algebra. 
We say that a $k$-ary operation $f$ \emph{preserves} 
an $m$-ary relation $R$, 
whenever $(x^1_1,\ldots,x^m_1),\ldots,(x^1_k,\ldots,x^m_k)$ in $R$, then also $(f(x^1_1,\ldots,x^1_k),\ldots,f(x^m_1,\ldots,x^m_k))$ in $R$.
The relation $R$ is called \emph{an invariant} of $f$, 
and the operation $f$ is called \emph{a polymorphism} of $R$.
An operation $f$ is \emph{a polymorphism} of $\Gamma$ if it preserves every relation 
from $\Gamma$. The \emph{polymorphism clone} $\Pol(\Gamma)$ is the set of all polymorphisms of $\Gamma$. 
Similarly, a relation $R$ is \emph{an invariant} of 
a set of functions $F$ if it is preserved by every operation from $F$.
By $\Inv(F)$ we denote the set of all invariants of $F$.
We call an operation $f$ \emph{idempotent} if $f(x,\ldots,x)=x$, for all $x$. An idempotent operation $f$ is a \emph{weak near-unanimity} (WNU) operation if $f(y,x,x,\ldots,x)=f(x,y,x,\ldots,x)=
\cdots=f(x,x,\ldots,x,y)$. We recall the following form of the (now proved) Feder-Vardi Conjecture.
\begin{thm}[CSP Dichotomy \cite{BulatovFVConjecture,ZhukFVConjecture}]
Let $\Gamma$ be a finite constraint language with all constants. If $\Gamma$ admits some WNU polymorphism, then $\CSP(\Gamma)$ is in P.  Otherwise, $\CSP(\Gamma)$ is NP-complete.
\label{thm:csp}
\end{thm}
\noindent Indeed, this theorem holds with the same criterion even without constants \cite{Barto2018}. However, for the general CSP one may assume without loss of generality that $\Gamma$ contains all constants (one can imagine these appearing in various forms, one possibility being all unary relations $x=c$, for $c \in D$). This is equivalent to the assumption that all operations $f$ of $\Pol(\Gamma)$ are idempotent. We can achieve this by moving to an equivalent constraint language known as the \emph{core}. The situation is more complicated for the QCSP and it is not known that a similar trick may be accomplished (see \cite{ChenMM13}). However, all prior conjectures for the QCSP have been made in this safer environment where we may assume idempotency and almost all classifications apply only to this situation. A rare exception to this is the paper \cite{QCSPmonoids} where the non-idempotent case is described as the \emph{terra incognita}.  We will henceforth assume $\Gamma$ contains all constants.

For the purpose of pedagogy it is useful to look at the $\Pi_2$ restriction of $\QCSP(\Gamma)$, denoted $\QCSP^2(\Gamma)$, in which the input is of the form $\forall x_1 \dots \forall x_n \exists y_1 \dots \exists y_m \;\Phi$. In order to solve this restriction of the problem it suffices to look at (the conjunction of) $|D|^n$ instances of $\CSP(\Gamma)$. It is not hard to show (see \cite{AU-Chen-PGP}) that, if $D^n$ can be generated under $\Pol(\Gamma)$ from some subset $\Sigma \subseteq D^n$, then one need only consult (the conjunction over) of $|\Sigma|$ instances of $\CSP(\Gamma)$. Suppose there is a polynomial $p$ such that for each $n$ there is a subset $\Sigma \subseteq D^n$ of size at most $p(n)$ so that  $D^n$ can be generated under $\Pol(\Gamma)$ from $\Sigma$, then we say  $\Pol(\Gamma)$ has the \emph{polynomially generated powers} (PGP) property. Under the additional assumption that there is a polynomial algorithm that computes these $\Sigma$, we would have a reduction to $\CSP(\Gamma)$. It turns out that if the nature of the PGP property is sufficiently benign a similar reduction can be made for the full $\QCSP(\Gamma)$ to the CSP with constants \cite{AU-Chen-PGP,LICS2015}. Another behaviour that might arise with $\Pol(\Gamma)$ is that there is an exponential function $f$ so that the smallest generating sets under  $\Pol(\Gamma)$ for $\Sigma \subseteq D^n$ require size at least $f(n)$. We describe this as the \emph{exponentially generated powers} (EGP) property. 
The outstanding conjecture in the area of QCSPs is the merger of Conjectures 6 and 7 in \cite{Meditations} which we have dubbed in \cite{MFCS2017} the \emph{Chen Conjecture}.
\begin{conj}[Chen Conjecture]
Let $\Gamma$ be a finite constraint language with all constants. If $\Pol(\Gamma)$ has PGP, then $\QCSP(\Gamma)$ is in NP; otherwise $\QCSP(\Gamma)$ is PSpace-complete.
\end{conj}
In \cite{Meditations}, Conjecture 6 gives the NP membership and Conjecture 7 the PSpace-completeness. In light of the proofs of the Feder-Vardi Conjecture, the Chen Conjecture implies the trichotomy of idempotent QCSP among P, NP-complete and PSpace-complete. Chen does not state that the PSpace-complete cases arise only from EGP, but this would surely have been on his mind (and he knew there was a dichotomy between PGP and EGP already for 3-element idempotent algebras \cite{AU-Chen-PGP}). Since \cite{ZhukGap2015}, it has been known for any finite domain that only the cases PGP and EGP arise (even in the non-idempotent case), and that PGP is always witnessed in the form of \emph{switchability}. It follows we know that the PGP cases are in NP \cite{AU-Chen-PGP,LICS2015}. 
\begin{thm}[\cite{MFCS2017}]
Let $\Gamma$ be a finite constraint language with all constants such that $\mathrm{Pol}(\Gamma)$ has PGP. Then $\QCSP(\Gamma)$ reduces to a polynomial number of instances of $\CSP(\Gamma)$ and is in NP.
\label{thm:PGP-in-NP}
\end{thm}
Using the CSP classification we can then separate the PGP cases into those that are in P and those that are NP-complete. 

A tantalizing characterization of idempotent $\Pol(\Gamma)$ that are EGP is given in \cite{ZhukGap2015}, where it is shown that $\Gamma$ must allow the \emph{primitive positive} (pp) definition (of the form $\exists x_1 \dots \exists x_n \; \Phi$) of relations
$\tau_n$ with the following special form.
\begin{definition}
Let the domain $D$ be so that $\alpha \cup \beta = D$ yet neither of $\alpha$ nor $\beta$ equals $D$. 
Let $S = \alpha^{3}\cup\beta^{3}$ and $\tau_n$ be the $3n$-ary relation given by  $\bigvee_{i \in [n]} S(x_i,y_i,z_i)$. 
\label{def:tau}
\end{definition}
The complement to $S$ represents the Not-All-Equal relation and 
the relations $\tau_n$ allow for the encoding of the complement of \emph{Not-All-Equal 3-Satisfiability} (where $\alpha \setminus \beta$ is $0$ and $\beta \setminus \alpha$ is $1$). Thus, if one has polynomially computable (in $n$) pp-definitions of $\tau_n$, then it is clear that $\QCSP(\Gamma)$ is co-NP-hard \cite{MFCS2017}.  In light of this observation,  it seemed that only a small step remained to proving the actual Chen Conjecture, at least with co-NP-hard in place of PSpace-complete. 


In this paper we refute the Chen Conjecture in a strong way while giving a long-desired classification for $\QCSP(\Gamma)$ where $\Gamma$ is a finite $3$-element constraint language with constants. Not only do we find $\Gamma$ so that $\QCSP(\Gamma)$ is co-NP-complete, but also we find $\Gamma$ so that $\Pol(\Gamma)$ has EGP yet $\QCSP(\Gamma)$ is in P. In these latter cases we can further prove that all pp-definitions of $\tau_n$ in $\Gamma$ are of size exponential in $n$. 
Additionally, we show that on a 4-element domain 
there exists a constraint language $\Gamma$ such that 
$\QCSP(\Gamma)$ is DP-complete (from Boolean Hierarchy), and on a 10-element domain 
there exists a constraint language giving the complexity class
$\Theta_{2}^{P}$.
Our main result for QCSP can be given as follows.
\begin{thm}
Let $\Gamma$ be a finite constraint language on $3$ elements which includes all constants. Then $\QCSP(\Gamma)$ is either in P, NP-complete, co-NP-complete or PSpace-complete.
\end{thm}

Meanwhile, we prove the Chen Conjecture is true for the class of finite conservative languages (these are those that have available all unary relations). One might see this as among the larger natural classes on which the Chen Conjecture holds. Another form of ``conservative QCSP'', in which relativization of the universal quantifier is permitted, has been given by Bodirsky and Chen \cite{BodirskyC09}. They uncovered a dichotomy between P and PSpace-complete, whereas the QCSP for finite conservative languages bequeaths the following trichotomy.

\begin{thm}[Conservative QCSP]
Let $\Gamma$ be a finite constraint language with all unary relations. If $\Pol(\Gamma)$ has PGP, then $\QCSP(\Gamma)$ is in NP. If $\Gamma$ further admits a WNU polymorphism, then $\QCSP(\Gamma)$ is in P, else it is NP-complete. Otherwise, $\Pol(\Gamma)$ has EGP and $\QCSP(\Gamma)$ is PSpace-complete.
\label{thm:conservative}
\end{thm}

It is hard to exaggerate how surprising our discovery of multitudinous complexities above P for the QCSP is. In Table~\ref{tab:model-cheking}  (reproduced but updated from \cite{MadelaineM18}), all syntactic fragments of first-order logic built from subsets of $\{\forall,\exists,\wedge,\vee,\neg,=\}$ are considered. It can be seen that they all give model-checking problems with simple, structured complexity-theoretical classifications (the classifications are simple but not necessarily the proofs), except the QCSP ($\{\forall,\exists,\wedge\}$, with or without $=$), and its dual  ($\{\forall,\exists,\vee\}$, with or without $\neq$), whose complexity classification is in any case a mirror of that for the QCSP. This holds for complexity classes of P and above (the classification of CSP complexities within P is quite rich).
\begin{table}[]
    \centering
     \begin{tabular}{|l|l|p{9cm}|} 
    \hline
    Fragment & Dual & Classification \\
    \hline
    $\{ \exists, \vee\}$            & 
    $\{ \forall, \wedge\}$          &
    {\multirow{3}{*}{$\mathrm{LogSpace}$}}    \\
    $\{ \exists, \vee, =\}$        & 
    $\{ \forall, \wedge, \neq \}$   &
    \\
    $\{ \exists, \vee, \neq \}$        & 
    $\{ \forall, \wedge, = \}$   &
    \\
    \hline
    $\{ \exists, \wedge, \vee \}$ & 
    $\{ \forall, \wedge, \vee \}$ &
    $\mathrm{LogSpace}$\ if there is some element $a$ so that all relations contain 
    \\
    $\{ \exists, \wedge, \vee, = \}$    & 
    $\{ \forall, \wedge, \vee, \neq \}$ & 
    the tuple $(a,\ldots,a)$, and $\mathrm{NP}$-complete otherwise
    \\
    $\{ \exists, \wedge, \vee, \neq \}$    & 
    $\{ \forall, \wedge, \vee, = \}$ & 
    \\
    \hline
    $\{ \exists, \wedge\}$     & 
    $\{ \forall, \vee\}$       &
   CSP dichotomy: $\mathrm{P}$ if language admits some WNU polymorphism,\\
   $\{ \exists, \wedge, = \}$ & 
    $\{ \forall, \vee, \neq \}$& and $\mathrm{NP}$-complete otherwise.
    \\
    \hline
    $\{ \exists, \wedge, \neq \}$ & 
    $\{ \forall, \vee, = \}$      &
    $\mathrm{NP}$-complete for domain at least $3$, reduces to Schaefer classes
    otherwise.
    \\
    \hline
    $\{ \exists, \forall, \wedge\}$     & 
    $\{ \exists, \forall, \vee\}$       &
    \multirow{2}{*}{QCSP polychotomy?} \\
   $\{ \exists, \forall, \wedge, = \}$ & 
    $\{ \exists, \forall, \vee,  \neq\}$    &
    \\
    \hline
    $\{ \exists, \forall, \wedge, \neq \}$ & 
    $\{ \exists, \forall, \vee,  =\}$    &
    $\mathrm{PSpace}$-complete for domain at least $3$, reduces to Schaefer classes
    for Quantified Sat otherwise.\\
    \hline
        \multicolumn{2}{|c|}{$\{ \forall, \exists, \wedge, \vee\}$} &  Tetrachotomy: $\mathrm{P}$, $\mathrm{NP}$-complete, $\mathrm{co-NP}$-complete or  $\mathrm{PSpace}$-complete\\
    \hline
    \multicolumn{2}{|l|}{
      $\{ \forall, \exists, \wedge, \vee, =\}$  
      $\{ \forall, \exists, \wedge, \vee, \neq \}$
    }%
    &
    \multirow{2}{*}{$\mathrm{LogSpace}$ when domain at most $1$, $\mathrm{PSpace}$-complete otherwise
    }
    \\
    \multicolumn{2}{|c|}{$\{ \neg,\exists, \forall, \wedge, \vee, = \}$} &  \\
    \hline
    \multicolumn{2}{|c|}{$\{ \neg,\exists, \forall, \wedge, \vee   \}$} &  
    $\mathrm{LogSpace}$ when language contains only empty or full relations,
    $\mathrm{PSpace}$-complete otherwise\\ 
    \hline
  \end{tabular} 
    \caption{Complexity classifications for model-checking syntactic fragments of first-order logic. The ``Schaefer classes'' are the domain $2$ CSPs that have a polynomial algorithm, as given in the seminal work \cite{Schaefer}.}
    \label{tab:model-cheking}
\end{table}

\subsection{Related Work}


In \cite{MFCS2017}, we have proved a variant of the Chen Conjecture using infinite relational languages encoded in quantifier-free logic with constants and equality.
An algebra consists of a finite domain and a set of operations on that domain. A polymorphism clone is an excellent example of an algebra which additionally satisfies certain properties of closure.
\begin{thm}[Revised Chen Conjecture \cite{MFCS2017}]
\label{thm:revised}
Let $\mathbb{A}$ be an idempotent algebra on a finite domain $A$ where we encode relations in $\mathrm{Inv}(\mathbb{A})$ in quantifier-free logic with constants and equality. If $\mathbb{A}$ satisfies PGP, then $\QCSP(\mathrm{Inv}(\mathbb{A}))$ is in NP. Otherwise, $\QCSP(\mathrm{Inv}(\mathbb{A}))$ is co-NP-hard.
\end{thm}
In this theorem it was known that co-NP-hardness could not be improved to PSpace-completeness, because $\QCSP(\mathrm{Inv}(\mathbb{A}))$ is co-NP-complete when, e.g., $\mathbb{A}=\mathrm{Pol}(\{0,1,2\};0,1,2,$ $\tau_1,\tau_2,\ldots\})$ where $\alpha=\{0,2\}$ and $\beta=\{1,2\}$. However, $\mathrm{Inv}(\mathbb{A})$  is not finitely related, that is, generated from a finite set of relations under pp-closure. It was not thought possible that there could be finite $\Gamma$ such that $\QCSP(\Gamma)$ is co-NP-complete. If we take the tuple-listing encoding of relations instead of quantifier-free logic with constants and equality, Theorem \ref{thm:revised} is known to fail \cite{MFCS2017}.

The systematic complexity-theoretic study of QCSPs dates to the early versions of \cite{BBCJK} (the earliest is a technical report from 2002). By the time of the journal version \cite{BBCJK}, the significance of the semilattice-without-unit $s=s_c$ (definition at opening of Section~\ref{HardnessSection}) had become apparent in a series of papers of Chen \cite{Chen04,hubie-sicomp,AU-Chen-PGP}. Although CSP$(\mathrm{Inv}(\{s\}))$ is in P it is proved in \cite{BBCJK} that QCSP$(\mathrm{Inv}(\{s\}))$ is PSpace-complete (even for some finite sublanguage of $\mathrm{Inv}(\{s\})$). We were unable to use the proof from that paper to expand the PSpace-complete classification in the $3$-element case, but we have expanded it nonetheless.

Finally, the study of which sequences of relations $R_i$, of arity $i$, have polynomial-sized (in $i$) pp-definitions in a finite constraint language $\Gamma$, has been addressed in \cite{LagerkvistW17}. Of course, this question for our relations $\tau_i$ plays a central role in this paper.

\subsection{Structure of the paper}

The paper is organized as follows (see Figure~\ref{fig:sec-dep}).
In Section 2 we formulate the main results of the paper.
We start with the classification of the complexity of $\QCSP(\Gamma)$ 
for constraint languages $\Gamma$ on a 3-element domain containing all constants.
Then we show how we can combine 
two constraint languages in one constraint language and explain 
how this idea gives exotic complexity classes such as $\DP = \NP\wedge \coNP$.

In Sections 3-5 we give necessary definitions,
prove Chen's conjecture for the conservative case,
and prove how to combine two constraint languages in one constraint 
language to generate new complexity classes.
In Section 6
we prove hardness results we need for the 3-element case and 
show how $\QCSP(\Gamma)$ 
can be reduced to 
$\QCSP^{2}(\Gamma)$ (the $\Pi_2$ restriction of $\QCSP(\Gamma)$).
In Sections 7 and 8 
we define two families of constraint languages $\Gamma$ such that 
$\Pol(\Gamma)$ has the EGP property 
but $\QCSP(\Gamma)$ can be solved in polynomial time. 
Additionally, we show that 
$\tau_{n}$ can be pp-defined from $\Gamma$ only by 
a formula of exponential size, which explains why 
the proof of co-NP-hardness does not work.

In Sections 9-10 we prove the classification of the complexity 
of $\QCSP(\Gamma)$ for constraint languages on a 3-element domain 
containing all constants.
In Section 9 we derive the existence of 
necessary relations in $\Gamma$ from the fact that $\Pol(\Gamma)$ 
has the EGP property. 
In Section 10 we complete the proof of the 3-element classification.

\begin{figure}
\[
 \xymatrix{
 & \fbox{\circled{4} Conservative Case} \ar[rrd] & & \\
        \fbox{\circled{3} Preliminaries} \ar[r]    & \fbox{\circled{5} QCSP Monsters} \ar[rr]  & &  \fbox{\circled{2} Main Results}  \\
 &  \fbox{\parbox{3.2cm}{\circled{6} Reductions \&  \\ Hardness \mbox{3-elements}}} \ar[rd] & & \\
 & \fbox{\circled{7} Strange Structure 1} \ar[r] & \fbox{\parbox{1.9cm}{\circled{9} EGP \\ and WNU \\ \mbox{3-elements}}} \ar[r] &\fbox{\parbox{1.9cm}{\circled{10} Main \\ result \\ \mbox{3-elements}}} \ar[uu] \\
  & \fbox{\circled{8} Strange Structure 2} \ar[ru] & &
    }
\]
\caption{Section dependency diagram for this article. 
}
\label{fig:sec-dep}
\end{figure}
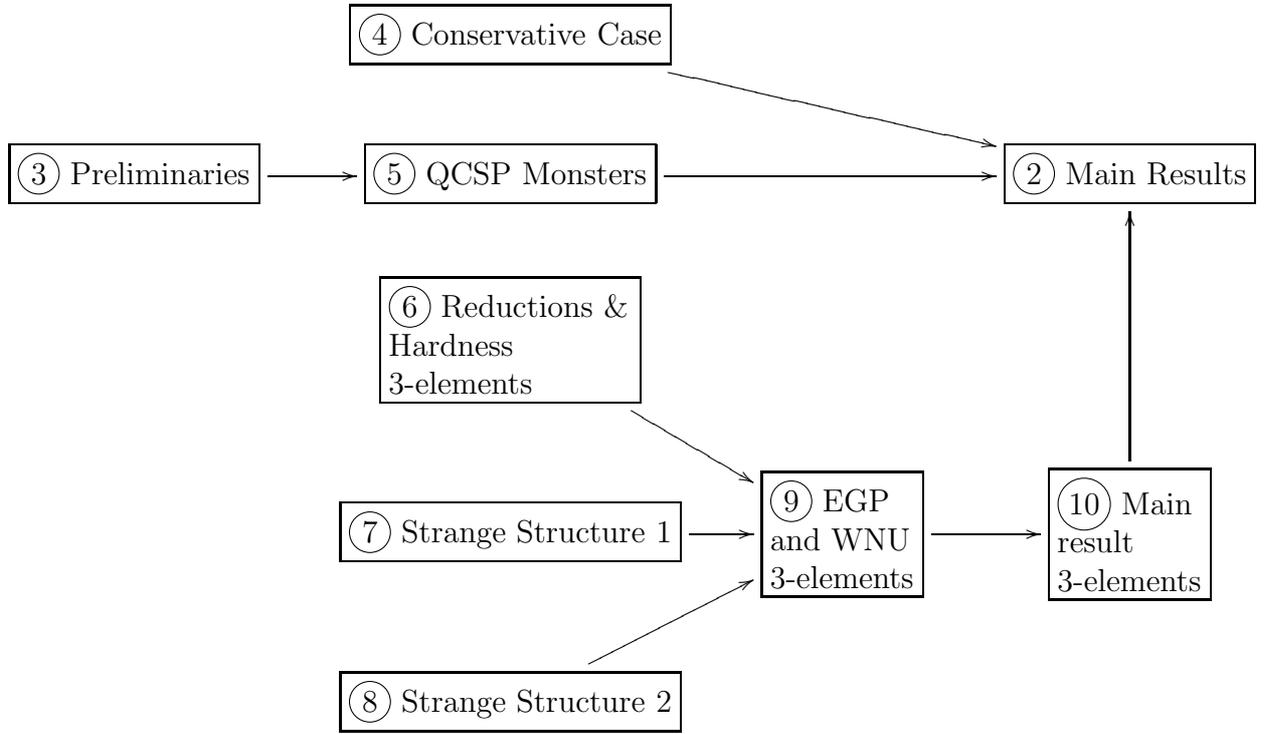

\section{Main Results}\label{MainResults}
In this section we formulate two main results of the paper:
classification of the complexity of $\QCSP(\Gamma)$ 
for all constraint languages $\Gamma$ on a 3-element domain containing 
all constants, 
and a theorem showing how we can combine constraint languages to 
obtain exotic complexity classes.

\subsection{QCSP on a 3-element domain}

Let $a$ and $c$ be constants of our domain $\{0,1,2\}$.

$f_{a,c}(x,y,z) = 
\begin{cases}
x,& \text{if $x=y$ or $y=z=a$}\\
c, &\text{otherwise.}
\end{cases}$

$s_{a,c}(x,y) = \begin{cases}
x, & \text{if $x=y$ or $y=a$}\\
c, & \text{otherwise.}
\end{cases}$

$g_{a,c}(x,y) = \begin{cases}
x, & \text{if $x=a$ or $y\neq c$}\\
c, & \text{otherwise.}
\end{cases}$

$s_{c}(x,y) = \begin{cases}
x, & \text{if $x=y$}\\
c, & \text{otherwise.}
\end{cases}$

\noindent If $\{a,b,c\}=\{0,1,2\}$ then call any idempotent operation $f$, such that 
$f(x,a,b)=x$ and $f(x,c,c) = c$, \emph{$ab$-stable}. Note that $f(x,y,z) = s_{a,c}(x,y)$
is an $ab$-stable operation.
We get the following characterization of the complexity of $\QCSP(\Gamma)$ on a 3-element domain.

\begin{thm}\label{ThreeElementClassification}
Suppose 
$\Gamma$ is a finite constraint language on $\{0,1,2\}$ with constants. Then $\QCSP(\Gamma)$ is 
\begin{enumerate}
\item in P, if $\Pol(\Gamma)$ has the PGP property and has a WNU operation.
\item NP-complete, if $\Pol(\Gamma)$ has the PGP property and has no WNU operation.
\item PSpace-complete, if $\Pol(\Gamma)$ has the EGP property and
has no WNU operation.
\item PSpace-complete, if $\Pol(\Gamma)$ has the EGP property and $\Pol(\Gamma)$ does not contain an $ab$-stable operation.
\item in P, if $\Pol(\Gamma)$ contains $s_{a,c}$ and $g_{a,c}$ for some $a,c\in\{0,1,2\}$, $a\neq c$.
\item in P, if $\Pol(\Gamma)$ contains $f_{a,c}$ for some $a,c\in\{0,1,2\}$, $a\neq c$.
\item co-NP-complete otherwise (in which case $\Pol(\Gamma)$ additionally has an ab-stable operation).
\end{enumerate}
\end{thm}

Note that the semilattice $s_c$ can be derived 
from each of the operations 
$f_{a,c}$, $s_{a,c}$.
As we know from \cite{BBCJK}, 
the problem 
$\QCSP(\Inv(s_2))$ is PSpace-complete.
Figure \ref{fig:EGP-WNU}
demonstrates how 
adding new operations 
makes the constraint language weaker and the 
quantified constraint satisfaction problem easier.
Note that all the languages in Figure \ref{fig:EGP-WNU} have the EGP property.
\begin{figure}
\begin{center}
\begin{tikzpicture}[scale=0.88]
 \node[shape=ellipse, rounded corners,
    draw, align=center,
    top color=white, bottom color=blue!20]  (s) at (-2,6) {$\{s_2\}$};
 \node[shape=ellipse, rounded corners,
    draw, align=center,
    top color=white, bottom color=blue!20]  (s02) at (-8,4) {$\{s_{0,2}\}$};
 \node[shape=ellipse, rounded corners,
    draw, align=center,
    top color=white, bottom color=blue!20]  (g02s) at (-4,4) {$\{g_{0,2},s_2\}$};
 \node[shape=ellipse, rounded corners,
    draw, align=center,
    top color=white, bottom color=blue!20]  (f02) at (-8,2) {$\{f_{0,2}\}$};
  \node[shape=ellipse, rounded corners,
    draw, align=center,
    top color=white, bottom color=blue!20]  (g02s02) at (-4,2) {$\{g_{0,2},s_{0,2}\}$};
 \node[shape=ellipse, rounded corners,
    draw, align=center,
    top color=white, bottom color=blue!20]  (g12s) at (0,4) {$\{g_{1,2},s_2\}$};
 \node[shape=ellipse, rounded corners,
    draw, align=center,
    top color=white, bottom color=blue!20]  (s12) at (4,4) {$\{s_{1,2}\}$};
 \node[shape=ellipse, rounded corners,
    draw, align=center,
    top color=white, bottom color=blue!20]  (g12s12) at (0,2) {$\{g_{1,2},s_{1,2}\}$};
  \node[shape=ellipse, rounded corners,
    draw, align=center,
    top color=white, bottom color=blue!20]  (f12) at (4,2) {$\{f_{1,2}\}$};

\node[scale=1.5] at (-2,7) {PSpace};
\node[scale=1.5] at (-8,6) {co-NP};
\node[scale=1.5] at (-8,1) {P};
\node[scale=1.5] at (-4,1) {P};
\node[scale=1.5] at (0,1) {P};
\node[scale=1.5] at (4,1) {P};

\path[->] (s) edge  (s02);
\path[->] (s) edge  (g02s);
\path[->] (s02) edge  (f02);
\path[->] (s02) edge  (g02s02);
\path[->] (g02s) edge  (g02s02);
\path[->] (s) edge  (s12);
\path[->] (s) edge  (g12s);
\path[->] (s12) edge  (f12);
\path[->] (s12) edge  (g12s12);
\path[->] (g12s) edge  (g12s12);

\draw (-7.1,7) .. controls (-7,1.2) and (3,1.2)  .. (3.1,7);
\draw (-6.2,1) .. controls (-6,3.4) and (-2,3.4)  .. (-1.8,1);
\draw (-9.6,1) .. controls (-9.8,3.4) and (-6.2,3.4)  .. (-6.4,1);
\draw (-2.2,1) .. controls (-2,3.4) and (2,3.4)  .. (2.2,1);
\draw (2.4,1) .. controls (2.2,3.4) and (5.8,3.4)  .. (5.6,1);

\end{tikzpicture}
\end{center}
\caption{Constraint languages defined as
invariants of sets of operations and their complexity.}
\label{fig:EGP-WNU}
\end{figure}
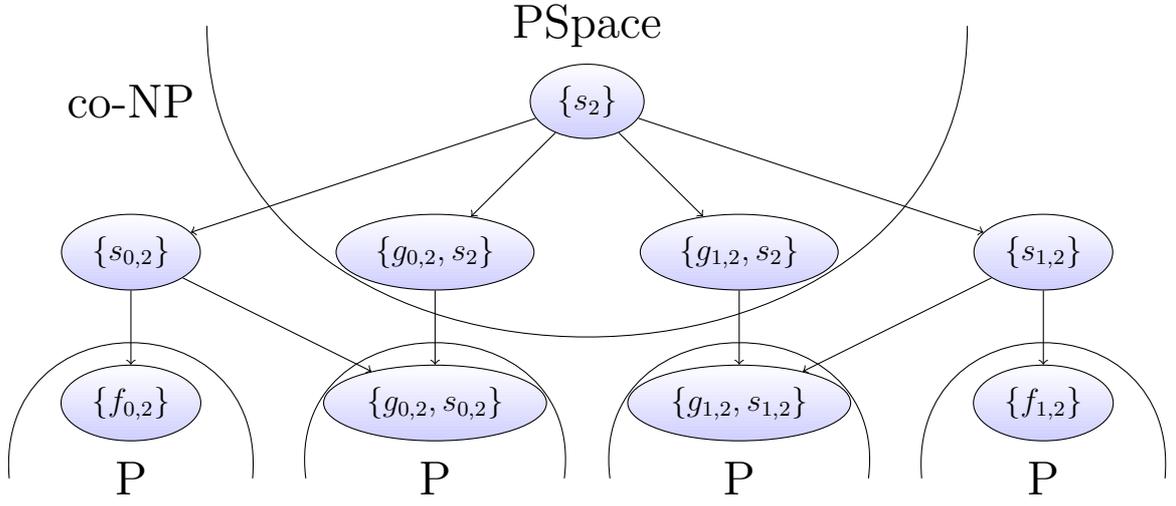

\subsection{Examples of Constraint Languages}

In this section we will show an example of a constraint language 
on $\{0,1,2\}$ for every case of Theorem~\ref{ThreeElementClassification} and explain  informally why it gives the respective complexity class.

Case (1). $\Pol(\Gamma)$ has the PGP property and has a WNU operation, 
and $\QCSP(\Gamma)$ is in P. 
We can build a constraint language $\Gamma$ with a single ternary relation $x-y+z=1 \bmod 3$. That $\Pol(\Gamma)$ has the PGP property and has a WNU operation as witnessed by the operation $x-y+z=1 \bmod 3$. We can solve $\QCSP(\Gamma)$ by reducing it to a polynomial number of instances of $\CSP(\Gamma)$, as detailed in Theorem~\ref{thm:PGP-in-NP}, then solving these by Gaussian elimination.

Case (2). 
$\Pol(\Gamma)$ has the PGP property but has no WNU operation, 
and $\QCSP(\Gamma)$ is NP-complete.
We can  take a single ternary relation $\{(1,0,0),(0,1,0),(0,0,1)\}$ that does not involve $2$. As a result universal quantifiers do not play any role and can be omitted, therefore 
the problem is equivalent to the CSP over the relation $\{(1,0,0),(0,1,0),(0,0,1)\}$ on the domain $\{0,1\}$ which is known to be NP-complete \cite{Schaefer}.

Case (3). 
$\Pol(\Gamma)$ has the EGP property and
has no WNU operation, 
and $\QCSP(\Gamma)$ is PSpace-complete.
We can take the closely related single ternary relation $$\{(x,0,0),(0,x,0),(0,0,x) : x \in \{1,2\} \}.$$ The QCSP over this language is equivalent to the QCSP over 
the relation $\{(1,0,0),(0,1,0),(0,0,1)\}$ on the domain $\{0,1\}$ which is known to be PSpace-complete \cite{BBCJK}.

Case (4). $\Pol(\Gamma)$ has the EGP property and $\Pol(\Gamma)$ does not contain an $ab$-stable operation, 
and $\QCSP(\Gamma)$ is PSpace-complete.
Unlike case (3), $\Gamma$ may have a WNU polymorphism.
Let $\tau$ be the ternary relation on $\{0,1,2\}$
consisting of all tuples 
$(a,b,c)$ such that $\{a,b,c\}\neq \{0,1\}$.
Then the complement to $\tau$ is 
equal to $\mathrm{NAE}_{3}$, where 
$\mathrm{NAE}_{3} = \{0,1\}^{3}\setminus \{(0,0,0),(1,1,1)\}$.
Put $$\sigma(x_{1},x_{2},x_{3},y_{1},y_{2}) =
(y_{1},y_{2}\in\{0,2\})\wedge (\tau(x_{1},x_{2},x_{3})\vee (y_{1}=y_{2})).$$
We claim that $\Gamma = \{\sigma,x=0,x=1,x=2\}$ satisfies the case (4) and $\QCSP(\Gamma)$ is PSpace-complete
(see Section~\ref{HardnessSection} for the detailed explanation).
The reduction will be from the complement of (monotone) \emph{Quantified Not-All-Equal 3-Satisfiability} (co-QNAE3SAT) which is co-PSpace-hard (see \cite{papa}) and consequently also PSpace-hard (as PSpace is closed under complement). 
First, we define a predicate 
$p(x_{1},\dots,x_{n},y,y')$ 
such that 
for 
$y\neq y'$, $y,y'\in\{0,2\}$,  and all 
$x_{1},\dots,x_{n}\in\{0,1\}$
it
equals 
the quantifier-free part of the instance of co-QNAE3SAT, 
and $p$ always holds if $x_{i}=2$ for some $i$ or 
$y=y'\in\{0,2\}$.
Put 
$p(x_{1},\dots,x_{n},y_0,y_{k})=\exists y_{1}\dots\exists y_{k-1}\bigwedge_{i=1}^{k}\sigma(z_{i}^1,z_{i}^2,z_{i}^3,y_{i-1},y_{i}),$
which expresses 
 the quantifier-free part of the instance of co-QNAE3SAT
 $$\Phi:= \neg (\mathrm{NAE}_3(z^1_1,z^2_1,z^3_1) \wedge \ldots \wedge \mathrm{NAE}_3(z^1_k,z^2_k,z^3_k)).$$
It remains to show how to add quantifiers to $\Phi$:
\begin{itemize}
    \item 
    The formula $\forall x_{1}\; p(x_{1},\dots,x_{n},y,y')$
    obviously expresses $\forall x_{1} \;\Phi$.
    \item It is more complicated with the existential quantifier because we need $x_{1}$ to be from $\{0,1\}$
    and we cannot just write $\exists x_{1}\; p(x_{1},\dots,x_{n},y,y')$ to express $\exists x_{1}\;\Phi$. 
    Nevertheless, the following trick suggested by Miroslav Ol\u{s}\'ak solves the problem.
    To express $\exists x_{1}\;\Phi$ we define
    $$ 
    p'(x_{2},\dots,x_{n},y,y') = \exists z \forall x_{1} \exists z' 
    \left(p(x_{1},\dots,x_{n},z,z') \wedge \sigma(x_{1},0,0,y,z') \wedge \sigma(x_{1},1,1,y',z')\right).$$
    Thus, instead of adding the existential quantifier to $x_{1}$ 
    we add the universal quantifier to $x_{1}$,
    two existential quantifiers, and two additional constraints. 
In fact, 
if $z$ were chosen to be equal to $y$ then
to satisfy the formula for $x_{1}=0$ 
we send $z'$ to $y'$, which, assuming $y\neq y'$, implies that 
$\Phi$ holds on $x_{1}=0$.
Otherwise, if $z\neq y$ then
to satisfy the formula for $x_{1}=1$ 
we send $z'$ to $y$, which, assuming $y\neq y'$, implies that 
$\Phi$ holds on $x_{1}=1$.
Thus, we simulated adding the existential quantifier to $\Phi$.
\end{itemize}

Cases (5)-(6).
Two concrete constraint languages are presented in Sections~\ref{StrangeOneSection} and 
\ref{StrangeTwoSection}.
These constraint languages 
have the EGP property but the $\QCSP$ over these languages can be solved in polynomial time. 
The polynomial algorithm is based on the following
two ideas.
\begin{itemize}
    \item The existence of the polymorphism 
    $s_{a,c}$ guarantees that 
    it is sufficient (see Lemma \ref{ReductionToQCSP2}) to find a polynomial algorithm for sentences of the 
    form 
    $\forall x_{1}\dots\forall x_{n}
    \exists y_{1} \dots \exists y_{m}\; \Phi$.
    \item 
Since $\Pol(\Gamma)$ has the EGP property, the usual reduction from 
$\QCSP$ to $\CSP$ (without looking at the formula) cannot help because 
theoretically we need to check exponentially many tuples 
$(x_{1},\dots,x_{n})$.
The key idea here is to calculate the list of tuples
$(x_{1},\dots,x_{n})$ we need to check to be sure that 
the formula holds for all the tuples. 
To calculate them we look at the formula
$\Phi$, substitute necessary values to some variables, 
and solve polynomially many CSP instances solvable in polynomial time (we have the semilattice polymorphism).
See the algorithms in Sections~\ref{StrangeOneSection} and \ref{StrangeTwoSection}
for more details.
\end{itemize}

Case (7). The first and probably easiest example of
a constraint language $\Gamma$ such that 
$\QCSP(\Gamma)$ is co-NP-complete is as follows. 
Let $R_{AND}=\begin{pmatrix}
0&0&1&1&2&\cdot\\
0&1&0&1&\cdot&2\\
0&0&0&1&\cdot&\cdot
\end{pmatrix}$ 
and $R_{OR}=\begin{pmatrix}
0&0&1&1&2&\cdot\\
0&1&0&1&\cdot&2\\
0&1&1&1&\cdot&\cdot
\end{pmatrix}$ where columns are tuples of the relations and $\cdot$ means any element of $\{0,1,2\}$, that is $R_{AND}$
and $R_{OR}$ are row-wise the truth table of AND
and OR respectively. 
Put $\Gamma = \{R_{AND}, R_{OR}\}$. The proof 
of co-NP-completeness can be divided into two steps
(see Section~\ref{HardnessSection} for the details):
\begin{itemize}
\item Using existential quantifiers we can compose these relations 
in the same way as we compose the operations 
$AND$ and $OR$ on $\{0,1\}$.
Using this idea we efficiently  define 
(see Lemma~\ref{ANDORHardness}) 
the complement to Not-All-Equal 3-Satisfability.
Hence, adding universal quantifiers makes the 
problem co-NP-hard.
\item To show that the problem is in co-NP we need to explain a polynomial strategy for the Existential player in a (Hintikka) game corresponding to the truth of the instance of the QCSP. It turns out that to calculate the best move Existential
can assume that all the later choices of the Universal
player are 0. It is also important that the Existential player
should play 2 if several moves are possible. Since the CSP over this language is solvable in polynomial time, after fixing the choice of the Universal player this can be done in polynomial time.

\end{itemize}

\subsection{QCSP Monsters}

The following theorem shows how we can combine constraint languages to obtain QCSPs with different complexities.

\begin{thm}\label{QCSPMonsters}
Suppose $\Gamma_{1}$, $\Gamma_{2}$, and $\Gamma_{3}$ are finite constraint languages on sets $A_1$, $A_{2}$, and $A_{3}$, respectively, and $\Gamma_{1}$ contains a constant relation 
$(x=a)$. Then there exist
constraint languages 
$\Delta_{1}$, 
$\Delta_{2}$,
$\Delta_{3}$,
$\Delta_{4}$,
on the domains of size
$|A_{1}|+1$, 
$|A_{2}|\cdot|A_{3}|+|A_{2}|+|A_{3}|$, 
$2\cdot|A_{2}|+|A_{3}|+2$,
and 
$|A_{2}|\cdot|A_{3}|+|A_{2}|+|A_{3}|+2$,
respectively, such that 
$\QCSP(\Delta_{i})$ is polynomially equivalent to the 
following problem:
\begin{enumerate}
\item[i=1] Given an instance of $\QCSP(\Gamma_{1})$
and instance of an NP-complete problem;
decide whether both of them hold,
i.e. $\QCSP(\Gamma_{1})\wedge \mathit{NP}$.
\item[i=2] Given an instance of $\QCSP(\Gamma_{2})$ and 
an instance of $\QCSP(\Gamma_{3})$;
decide whether both of them hold,
i.e. $\QCSP(\Gamma_2)\wedge \QCSP(\Gamma_3)$.
\item[i=3] Given $n>0$, 
instances $I_{1},\ldots,I_{n}$ of $\QCSP(\Gamma_{2})$,
and
instances $J_{1},\ldots,J_{n}$ of $\CSP(\Gamma_{3})$;
decide whether
$(I_{1}\vee J_{1})\wedge \dots\wedge (I_{n}\vee J_{n})$
holds,
i.e.
$(\QCSP(\Gamma_{2})\vee \CSP(\Gamma_{3}))
\wedge
\dots
\wedge
(\QCSP(\Gamma_{2})\vee \CSP(\Gamma_{3}))$.

\item[i=4] Given $n>0$, instances  
$I_{1},\ldots,I_{n}$ of $\QCSP(\Gamma_{2})$,
and
instances $J_{1},\ldots,J_{n}$ of $\QCSP(\Gamma_{3})$;
decide whether
$(I_{1}\vee J_{1})\wedge \dots\wedge (I_{n}\vee J_{n})$
holds,
i.e.
$(\QCSP(\Gamma_{2})\vee \QCSP(\Gamma_{3}))
\wedge
\dots
\wedge
(\QCSP(\Gamma_{2})\vee \QCSP(\Gamma_{3}))$.
\end{enumerate}
\end{thm}

\begin{proof}
The proof for $i=1$, $i=2$, $i=3$, and $i=4$ 
follows from Lemmas~\ref{ConstraintConjunctionOne},
\ref{ConstraintConjunctionTwo}, \ref{ConstraintDisjunctionFirst}, and \ref{ConstraintDisjunction},
respectively.
\end{proof}

\begin{cor}\label{DPon4elements}
There exists a finite constraint language 
$\Gamma$ on a 4-element domain such that 
$\QCSP(\Gamma)$ is DP-complete (where 
$\mathit{DP} = \mathit{NP}
\wedge \mathit{co}\mbox{-}\mathit{NP}$ from the Boolean hierarchy).
\end{cor}

\begin{proof}
By Theorem~\ref{ThreeElementClassification},
there exists a constraint language $\Gamma_{1}$ on a 3-element domain with constants such that $\QCSP(\Gamma_{1})$ is co-NP-complete.
Applying Theorem~\ref{QCSPMonsters} with $i=1$
to $\Gamma_{1}$ we obtain a constraint language 
$\Gamma$ on a 4-element domain such that $\QCSP(\Gamma)$ is polynomially equivalent 
to DP.
\end{proof}

The complexity class $\Theta^{\mathrm{P}}_2$ (see \cite{LukasiewiczM17} and references therein) admits various definitions, one of which is that it allows a Turing machine polynomial time with a logarithmic number of calls to an NP\,oracle. A condition proved equivalent to this, through Theorems 4 and 7 of \cite{BussH91}, is as follows. 
In this theorem $i \le p(|x|)$ indicates $i$ is a  positive integer smaller than $p(|x|)$, where $x$ is a string of length $|x|$.
\begin{thm}[\cite{BussH91}]
Every predicate in $\Theta^{\mathrm{P}}_2$ can be defined by a formula of the form
$\exists i \le p(|x|) \ A(i,x) \wedge \neg B(i,x)$ as well as by a formula of the form $\forall i \le p'(|x|) \ A'(i,x) \vee \neg B'(i,x)$, where $A, B, A'$ and $B'$ are NP-predicates and $p$ and $p'$ are polynomials.
\label{thm:BH}
\end{thm}
\noindent The second (universal) characterization will play the key role in the following observation.
\begin{cor}
There exists a finite constraint language 
$\Gamma$ on a 10-element domain such that 
$\QCSP(\Gamma)$ is $\Theta^{\mathrm{P}}_2$-complete.
\label{cor:emil}
\end{cor}
\begin{proof}
By Theorem~\ref{ThreeElementClassification},
there exists a constraint language $\Gamma_{1}$ on a 3-element domain with constants such that $\QCSP(\Gamma_{1})$ is co-NP-complete.
Choose a constraint language $\Gamma_{2}$ on a 2-element domain such that 
$\CSP(\Gamma_{2})$ is NP-complete.
Using item 3 of Theorem~\ref{QCSPMonsters},
we construct a constraint language $\Gamma$ so that $\QCSP(\Gamma)$ is equivalent to the truth of $(I_{1}\vee J_{1})\wedge \dots\wedge (I_{n}\vee J_{n})$, where $I_1,\ldots,I_n$ are instances of $\QCSP(\Gamma_{1})$ and $J_1,\ldots,J_n$ are instances of $\CSP(\Gamma_{2})$. 

To prove membership of $\QCSP(\Gamma)$ in $\Theta^{\mathrm{P}}_2$, we use the second characterization of Theorem \ref{thm:BH} together with $A'(i,x)$ indicating that $J_i$ is a yes-instance of $\CSP(\Gamma_{2})$ and $\neg B'(i,x)$ indicating that $I_i$ is a yes-instance (or $B'(i,x)$ indicating $I_i$ is a no-instance) of $\QCSP(\Gamma_{1})$. Thus, we want $i$ to range over numbers from $1$ to $n$, so in the predicates $A'(i,x)$ and $\neg B'(i,x)$ we should in particular set these to be true if $i$ is not a number from $1$ to $n$.

To prove that $\QCSP(\Gamma)$ is $\Theta^{\mathrm{P}}_2$-complete, we use again the second formulation of characterization of Theorem \ref{thm:BH}, but this time break the universal quantification into a conjunction of length $p'(|x|)$.
\end{proof}



\section{Preliminaries}

Let $[n]=\{1,\ldots,n\}$. We identify a constraint language $\Gamma$ with a set of relations over a fixed finite domain $D$. We may also think of this as a first-order relational structure. If $\Phi$ is a first-order formula including $x_1,\ldots,x_n$ among its free variables and not containing $y_1,\ldots,y_n$ in any capacity, then $\Phi^{x_1,\ldots,x_n}_{y_1,\ldots,y_n}$ is the result of substituting the free occurrences of $x_1,\ldots,x_n$ by $y_1,\ldots,y_n$, respectively. If $I$ is an instance of $\QCSP(\Gamma)$, then $\mathrm{Var}(I)$ refers to the variables mentioned in $I$.

We always may assume that an instance of $\QCSP(\Gamma)$ is of the prenex form $$\forall x_1\exists y_1 \forall x_2 \exists y_2 \dots \forall x_{n} \exists y_{n} \Phi,$$ since if it is not it may readily be brought into such a form in polynomial time. Then a solution is a sequence of (Skolem) functions 
$f_{1},\ldots, f_{n}$ such that 
$$(x_1,f(x_1),x_2,f_2(x_1,x_2),\ldots,x_{n},f_n(x_1,\ldots,x_n))$$ is a solution of $\Phi$ for all 
$x_{1},\ldots,x_{n}$ (i.e. $y_{i} = f_{i}(x_{1},\ldots,x_{i})$). This belies a (Hintikka) game semantics for the truth of a QCSP instance in which a player called Universal plays the universal variables and a player called Existential plays the existential variables, one after another, from the outside in. The Skolem functions above give a strategy for Existential. In our proofs we may occasionally revert to a game-theoretical parlance.



An \emph{algebra} $\mathbb{A}$ consists of domain and a set of operations defined on that domain. The most important type of algebra in this paper is a clone. Let $\Clo(G)$ be the \emph{clone} generated by the set of operations $G$, that is the closure of $G$ under the addition of projections and composition, where the composition of a $k$-ary operation $f$ and $m$-ary operations $g_1,\dots,g_{k}$ is the $m$-ary operation defined by $f(g_1,\dots,g_{k})$.

In general with our operators, if the argument is a singleton set, we omit the curly brackets. A \emph{subalgebra} of $\mathbb{A}$ consists of a subset $D$ of the domain of $\mathbb{A}$, that is preserved by all the operations of $G$, together with all the operations of $\mathbb{A}$ restricted to $D$. A \emph{congruence} on an algebra $\mathbb{A}$ is an equivalence relation $\sim$ on its domain so that, for each $k$-ary operation $f$ in $\mathbb{A}$, $f(x_1,\ldots,x_k) \sim f(y_1,\ldots,y_k)$ whenever $x_1\sim y_1$, \ldots, $x_k \sim y_k$. We can quotient $\mathbb{A}$ by $\sim$ in the obvious way to obtain a new algebra that we describe as a \emph{homomorphic image} of $\mathbb{A}$. A \emph{factor} of $\mathbb{A}$ is a homomorphic image of a subalgebra of $\mathbb{A}$.

A formula of the form 
$\exists y_{1}\dots\exists y_{n} \Phi$,
where $\Phi$ is a  conjunction of relations from $\Gamma$ is called 
\emph{a positive primitive formula (pp-formula) over $\Gamma$}.
If $R(x_{1},\dots,x_{n}) = \exists y_{1}\dots\exists y_{n} \Phi$, 
then we say that $R$ is \emph{pp-defined} by 
$\exists y_{1}\dots\exists y_{n} \Phi$, 
and $\exists y_{1}\dots\exists y_{n} \Phi$ is called 
\emph{a pp-definition}.
Note that if a relation $R$ is pp-definable over $\Gamma$ then it is preserved by any operation $f\in\Pol(\Gamma)$ \cite{bond,geiger1968closed}.

In a pp-formula we allow always, except for Section 5, the use of constants from the domain. 
Note that using constants is equivalent to having all unary relations $x=c$ in our constraint language.
On the algebraic side, this corresponds to assuming all polymorphism operations are idempotent. 
For a conjunctive formula $\Phi$
by $\Phi(x_{1},\dots,x_{n})$ we denote the
$n$-ary relation defined by 
a pp-formula 
where all variables except $x_{1},\dots,x_{n}$ are existentially quantified.
We do not require the variables
$x_{1},\dots,x_{n}$ to be different; for instance, $\Phi(x_1,x_1)$ defines pairs of equal elements.  
Equivalently, 
$\Phi(x_{1},\dots,x_{n})$ is the set of all tuples 
$(a_{1},\dots,a_{n})$ such that 
$\Phi$ has a solution with 
$(x_{1},\dots,x_{n}) = (a_{1},\dots,a_{n})$.

For a $k$-ary relation $R$ and a set of coordinates $B \subset [k]$, define $\proj_B(R)$ to be the $|B|$-ary relation obtained from $R$ by projecting to $B$, or equivalently, existentially quantifying variables at positions $[k] \setminus B$.

For a tuple 
$\alpha$ by $\alpha(n)$ we denote the $n$-th element of $\alpha$.
We define relations by matrices where the columns list the tuples.

Let $\alpha$ and $\beta$ be strict subsets of $D$ so that $\alpha \cup \beta = D$. The most interesting cases arise when $\alpha \cap \beta \neq \varnothing$ but we will not insist on this at this point. An $n$-ary operation $f$ is \emph{$\alpha\beta$-projective} if there exists $i \in [n]$ so that $f(x_1,\ldots,x_n) \in \alpha$, if $x_i \in \alpha$, and $f(x_1,\ldots,x_n) \in \beta$, if $x_i \in \beta$. In this case, we may say that $f$ is $\alpha\beta$-projective to coordinate $i$. It is now known that an idempotent algebra $\mathbb{A}$ over domain $D$ has EGP iff there exists $\alpha$ and $\beta$, strict subsets of $D$, so that all operations of $\mathbb{A}$ are $\alpha\beta$-projective \cite{ZhukGap2015}.

\section{The conservative case}

In this section we prove Theorem 4 describing the complexity 
of $\QCSP(\Gamma)$ for conservative constraint languages $\Gamma$, i.e. 
languages containing all unary relations.
As it was mentioned in the introduction, 
if $\Pol(\Gamma)$ has the PGP property then we can reduce $\QCSP(\Gamma)$ to several copies 
of $\CSP$. 
Thus, the only open question was the complexity for the EGP case.  
Here we will use the following fact from \cite{MFCS2017}.

\begin{lem}[\cite{MFCS2017}]
\label{tauRelationsExistence}
Suppose $\Gamma$ is a constraint language on domain $D$ with constants, such that
$\Pol(\Gamma)$ has the EGP property.
Then there exist $\alpha,\beta\subsetneq D$ such that 
$\alpha\cup\beta = D$ and $\tau_{n}$ (as in Definition \ref{def:tau}) is pp-definable from $\Gamma$
for every $n\ge 1$.
\end{lem}
\noindent It turns out 
that if $\Gamma$ contains all unary relations then 
two copies of $\tau_{k}$ can be composed to define 
the relation $\tau_{2(k-1)}$ as follows.
Choose $0 \in \alpha \setminus \beta$ and $1 \in \beta \setminus \alpha$, 
then
\begin{align*} \tau_{2(k-1)}(x_1,y_1,z_1\ldots,x_{2(k-1)}, y_{2(k-1)}&,z_{2(k-1)})=
\exists w \; \tau_{k}(x_1,y_1,z_1\ldots,x_{k-1},y_{k-1},z_{k-1},0,0,w) \wedge \\
\tau_{k}&(x_{k},y_{k},z_{k},\ldots,x_{2(k-1)},y_{2(k-1)},z_{2(k-1)},1,1,w) \wedge w \in \{0,1\}.
\end{align*}
\noindent Identifying variables in $\tau_k$ we can derive
$\tau_{k-1}$, therefore $\tau_k$ is pp-definable from $\tau_j$ and unary relations whenever $k\ge j \ge 3$. 
\begin{lem}
There is a polynomially (in $k$) computable pp-definition of $\tau_k$ from $\tau_3$ and unary relations.
\label{lem:cons}
\end{lem}
\begin{proof}
As above we can define $\tau_{2(k-1)}$ in a recursive fashion using  two copies of $\tau_k$
plus a single new existential quantifier whose variable is restricted to being on domain $\{0,1\}$.  
Note that in the recursive pp-definition of $\tau_{k}$ over $\tau_3$ every 
variable that is not quantified appears just once, 
each quantified variable appears three times, and 
most variables are not quantified.
Therefore, our recursive scheme gives a polynomially computable pp-definition of $\tau_{k}$.
\end{proof}
\noindent 
We are now in a position to prove  Theorem \ref{thm:conservative}, whose statement we recall.

\begin{THMconservative}
Let $\Gamma$ be a finite constraint language with all unary relations. If $\Pol(\Gamma)$ has PGP, then $\QCSP(\Gamma)$ is in NP. If $\Gamma$ further admits a WNU polymorphism, then $\QCSP(\Gamma)$ is in P, else it is NP-complete. Otherwise, $\Pol(\Gamma)$ has EGP and $\QCSP(\Gamma)$ is PSpace-complete.
\end{THMconservative}

\begin{proof}
Assume $\Gamma$ is a finite constraint language with all unary relations. Suppose Pol$(\Gamma)$ has PGP. Then we know from Theorem \ref{thm:PGP-in-NP} that QCSP$(\Gamma)$ reduces to a polynomial number of instances of CSP$(\Gamma)$. It follows from Theorem \ref{thm:csp} that if $\Gamma$ admits a WNU then QCSP$(\Gamma)$ is in P, otherwise QCSP$(\Gamma)$ is NP-complete.

Suppose now Pol$(\Gamma)$ has EGP. By Lemma~\ref{tauRelationsExistence} there exist $\alpha,\beta$ as in  Definition \ref{def:tau} such that 
$\tau_{3}$ is pp-definable from $\Gamma$.
Combining this with Lemma~\ref{lem:cons} we conclude that 
there are polynomially (in $k$) computable pp-definitions of $\tau_k$ in $\Gamma$. 
We will reduce from the complement of \emph{Quantified Not-All-Equal 3-Satisfiability} (QNAE3SAT) which is known to be PSpace-complete (see, e.g., \cite{papa}). From an instance $\phi=\neg \forall x_1 \exists y_1 \ldots \forall x_n \exists y_n  \ \Phi$ of co-QNAE3SAT, where $\Phi=\mathrm{NAE}_3(z^1_1,z^2_1,z^3_1) \wedge \ldots \wedge \mathrm{NAE}_3(z^1_k,z^2_k,z^3_k)$ and $z^1_1,z^2_1,z^3_1,\ldots,z^1_k,z^2_k,z^3_k \in \{x_1,y_1,\ldots,x_n,y_n\}$, we build an instance $\phi'$ of QCSP$(\Gamma)$ as follows. Consider $\phi$ to be $\exists x_1 \forall y_1 \ldots$ $ \exists x_n \forall y_n  \ \neg \Phi$ and set $\phi'=$
\[ \exists x_1 \forall y_1 \ldots \exists x_n \forall y_n  \ x_1,\ldots,x_n \in (\alpha \setminus \beta \cup \beta \setminus \alpha) \ \wedge \ \tau_k(z^1_1,z^2_1,z^3_1,\ldots,z^1_k,z^2_k,z^3_k).\]
The idea is that the set $\alpha \setminus \beta$ plays the role of $0$ and $\beta \setminus \alpha$ plays the role of $1$.

($\phi \in \mathrm{co\mbox{-}QNAE3SAT}$ implies $\phi' \in \mathrm{QCSP}(\Gamma)$.) 
Let the universal variables be evaluated in $\phi'$ and match them in $\phi$ according to $\alpha \setminus \beta$ being $0$ and $\beta \setminus \alpha$ being $1$. If a universal variable in $\phi'$ is evaluated in $\alpha \cap \beta$, then we can match it in $\phi$ \mbox{w.l.o.g.} to $0$. Now, read a valuation of the existential variables of $\phi'$ from those in $\phi$ according to $0$ becoming any fixed $d_0 \in \alpha \setminus \beta$ and $1$ becoming any fixed $d_1 \in \beta \setminus \alpha$. By construction we have $\phi' \in \mathrm{QCSP}(\Gamma)$.

($\phi' \in \mathrm{QCSP}(\Gamma)$ implies $\phi \in \mathrm{co\mbox{-}QNAE3SAT}$.) Suppose $\phi' \in \mathrm{QCSP}(\Gamma)$. We will prove $\phi \in \mathrm{co\mbox{-}QNAE3SAT}$ again using the form of $\phi$ being $\exists x_1 \forall y_1 \ldots \exists x_n \forall y_n  \ \neg \Phi$. Let the universal variables be evaluated in $\phi$ and match them in $\phi'$  according to $0$ becoming any fixed $d_0 \in \alpha \setminus \beta$ and $1$ becoming any fixed $d_1 \in \beta \setminus \alpha$. Now, read a valuation of the existential variables of $\phi$ from $\phi'$ according to  to $\alpha \setminus \beta$ being $0$ and $\beta \setminus \alpha$ being $1$. By construction we have $\phi \in \mathrm{co\mbox{-}QNAE3SAT}$.
\end{proof}

\section{QCSP Monsters}\label{MonstersProof}

This section explains the building of monsters with greater than a $3$-element domain. It has no bearing on the $3$-element classification.

\begin{lem}\label{ConstraintConjunctionOne}
Suppose $\Gamma$ is a finite constraint language on a set $A$ where $|A|>1$
containing the unary relation $x=a$. 
Then there exists a constraint language $\Gamma'$ 
on a domain of size $|A|+1$
such that $\QCSP(\Gamma')$ is polynomially equivalent to
$\QCSP(\Gamma)\wedge \mathit{NP}$, 
that is the following decision problem: 
given an instance of $\QCSP(\Gamma)$ and 
an instance of some NP-complete problem;
decide whether both of them hold.
\end{lem}
\begin{proof}
Let $a\in A$ be as in the statement of the lemma and let $a'$ be some element not in $A$.
Put $A' = A\cup\{a'\}$.

Put $\phi(x) = \begin{cases} 
x, &\text{if $x\in A$}\\
a, &\text{if $x= a'$}
\end{cases}.$ 
We assign a relation $R'$ on the set $A'$ to every $R\in \Gamma$ as follows:
$R' = \{(a_{1},\ldots,a_{h})\mid 
(\phi(a_{1}),\ldots,\phi(a_h))\in R\}$.
Let $\mathrm{NAE}_3\subseteq \{a,a'\}^{3}$ be the ternary relation 
containing all tuples on $\{a,a'\}$ except for
$(a,a,a),(a',a',a')$.
Let $\Gamma'=\{R'\mid R\in \Gamma\}\cup \{\mathrm{NAE}_3\}$.

Suppose $I$ is an instance of $\QCSP(\Gamma)$ and 
$J$ is an instance of $\CSP(\{\mathrm{NAE}_3\})$, which is an NP-complete problem.
If we replace every relation $R$ from $\Gamma$ by the corresponding relation $R'$, we get an instance $I'$ that is equivalent to $I$.
Then the instance $I'\wedge J$ can be viewed as an instance of $\QCSP(\Gamma')$ that is equivalent to $I\wedge J$.

Suppose $I'$ is an instance of $\QCSP(\Gamma')$. \mbox{W.l.o.g.} we will assume that no variable appearing in an $\mathrm{NAE}_3$ relation is universally quantified, else, since $|A|>1$, this is a no-instance of $\QCSP(\Gamma')$ and can be reduced to a fixed no-instance (e.g.) $J$ of $\CSP(\{\mathrm{NAE}_3\})$. Now, we define an instance $I$ of $\QCSP(\Gamma)$
and an instance $J$ of $\CSP(\{\mathrm{NAE}_3\})$ as follows.
$I$ is obtained from $I'$ by replacement of 
all relations $R'$ by the corresponding $R$ and 
$\mathrm{NAE}_3$ by $\{(a,a,a)\}$.
Since $\Gamma$ contains $x=a$, $I$ is an instance of 
$\QCSP(\Gamma)$.
The instance $J$ consists of the $\mathrm{NAE}_3$-part of $I'$ which is a CSP as we already assumed it contains no universal variables. Now, to see $I'\in \QCSP(\Gamma')$ iff $I\in \QCSP(\Gamma)$ and $J \in \CSP(\{\mathrm{NAE}_3\})$ it is enough to observe that $\QCSP(\Gamma)$ and $\QCSP(\Gamma'\setminus \{\mathrm{NAE}_3\})$ are equivalent on all 
instances. 
\end{proof}

\begin{lem}\label{ConstraintDisjunction}
Suppose $\Gamma_{1}$ and $\Gamma_{2}$ are finite constraint languages on sets $A_1$ and $A_{2}$ respectively. 
Then there exists a constraint language $\Gamma$ 
on a domain of size $|A_{1}|\cdot|A_{2}|+|A_{1}|+|A_{2}|+2$
such that $\QCSP(\Gamma)$ is polynomially equivalent to 
$(\QCSP(\Gamma_{1})\vee \QCSP(\Gamma_{2}))
\wedge
\dots
\wedge
(\QCSP(\Gamma_{1})\vee \QCSP(\Gamma_{2}))$, 
i.e. the following decision problem:
given $n$,
instances 
$I_{1},\ldots,I_{n}$ of $\QCSP(\Gamma_{1})$,
and
instances $J_{1},\ldots,J_{n}$ of $\QCSP(\Gamma_{2})$;
decide whether
$(I_{1}\vee J_{1})\wedge \dots\wedge (I_{n}\vee J_{n})$
holds.
\end{lem}

\begin{proof}
Assume that $A_{1}\cap A_{2} = \varnothing$, 
$a_{1},a_{2}\notin A_{1}\cup A_{2}$.
Let $$A = (A_{1}\times A_{2})\cup A_{1} \cup A_{2}\cup \{a_{1},a_{2}\},$$
$$\sigma= (A_{1} \times \{a_{1}\})\cup 
(\{a_{2}\}\times A_{2}),$$
$$\sigma_{1} = 
\{(a,(a,b))\mid a\in A_{1}, b\in A_{2}\}
\cup 
(\{a_{2}\}\times A)\cup 
(A_{1}\times (A_{1}\cup A_{2}\cup\{a_{1},a_{2}\})),$$
$$\sigma_{2} = 
\{(b,(a,b))\mid a\in A_{1}, b\in A_{2}\}
\cup 
(\{a_{1}\}\times A)\cup 
(A_{2}\times (A_{1}\cup A_{2}\cup\{a_{1},a_{2}\})),$$
$$\Gamma = \{R\cup\{(a_{2},\ldots,a_{2})\}\mid 
R\in \Gamma_{1}\}
\cup 
\{R\cup\{(a_{1},\ldots,a_{1})\}\mid 
R\in \Gamma_{2}\}
\cup\{\sigma_{1},\sigma_{2},\sigma\}.$$

We would like to assume, w.l.o.g., that $\Gamma$ contains constants $a_1$ and $a_2$. This would follow from our definitions, so long as both $\Gamma_1$ and $\Gamma_2$ contain a unary empty relation. We may assume this, through definitions of the form $\forall x_1,\ldots,x_k R(x_1,\ldots,x_k)$, so long as not every relation in either $\Gamma_1$ or $\Gamma_2$ contains all tuples. Suppose, w.l.o.g., that every relation in $\Gamma_1$ contains all tuples, then we take $\Gamma$ to be $\Gamma_1$ and we are done. Thus we may assume w.l.o.g., that $\Gamma$ contains constants $a_1$ and $a_2$.

Suppose we have an instance $I_{1}$ of $\QCSP(\Gamma_{1})$ 
and an instance $I_{2}$ of $\QCSP(\Gamma_{2})$. W.l.o.g. we will assume that neither $I_1$ nor $I_2$ is empty and that they are variable disjoint. We will explain how to build an instance $J$ of  $\QCSP(\Gamma)$.
Let $x_{1},\ldots,x_{n}$ be all universally quantified variables 
of $I_{1}$. 

We replace every atomic relation $R$ of $I_{1}$ by
$R\cup\{(a_{2},\ldots,a_{2})\}$ 
and add relational constraints 
$\sigma_{1}(x_{i},y_{i})$ for every $i\in[n]$ (where the variables $y_i$ are new).
Also we replace $\forall x_{i}$ by $\forall y_{i}\exists x_{i}$ 
for every $i\in[n]$.
The result we denote by $I_{1}'$.
Similarly, but with $a_{1}$ instead of $a_{2}$ and $\sigma_{2}$ instead of $\sigma_{1}$ we define $I_{2}'$.
We claim that the sentence $J$  defined by 
$$I_{1}'\wedge I_{2}'\wedge\bigwedge\limits_{u\in\Var(I_{1}), v\in\Var(I_2)}\sigma(u,v)$$ 
(we move all the quantifiers to the left part after joining)
holds if and only if $I_{1}$ holds or $I_{2}$ holds. W.l.o.g. we will henceforth assume the first variable in $J$ is existential (if necessary we could enforce this by a dummy existential variable at the beginning of $I_1$).

Let us guide our proof with an example. Suppose that $I_1$ is $\exists x_1 \forall x_2 \exists x_3 \forall x_4 \Phi_1(x_1,x_2,x_3,x_4)$ and $I_2$ is $\forall x'_1 \exists x'_2 \Phi_2(x'_1,x'_2)$ then $J$ has the form
\[
\begin{array}{ll}
\exists x_1 \forall y_2 \exists x_2 \exists x_3 \forall y_4 \exists x_4 \forall y'_1 \exists x'_1 \exists x'_2 & \Phi_1(x_1,x_2,x_3,x_4) \wedge  \Phi_2(x'_1,x'_2) \wedge \\
& \bigwedge_{i} \sigma_1(x_i,y_i) \wedge \bigwedge_{j} \sigma_2(x'_j,y'_j) \wedge  \bigwedge_{i,j} \sigma(x_i,x'_j). 
\end{array}
\]
($I_1 \in \QCSP(\Gamma_1) \vee I_2 \in \QCSP(\Gamma_2)$ implies $J \in \QCSP(\Gamma)$.) W.l.o.g. $I_1 \in \QCSP(\Gamma_1)$. Let us give a winning strategy for Existential in $J$ based on the winning strategy Existential enjoys on $I_1$. Existential will evaluate all (existential) variables of $J$ coming from $I_2$ as $a_1$. It follows that all of the atoms in $J$ arising purely from $I_2$, in our example $\Phi_2(x'_1,x'_2)$ and $\bigwedge_{j} \sigma_2(x'_j,y'_j)$, will be satisfied. There is no longer a need to worry about how Universal plays in $J$ on variables $y'_i$ coming from $I_2$. Now, consider Universal's play in $J$ on a variable $y_i$ coming from $I_1$. If he plays some pair $(a,b)$, with $a \in A_1$ and $b \in A_2$, then let Existential respond with $x_i$ set to $a$. If he plays some element $a \in A_1 \cup A_2 \cup \{a_1,a_2\}$, then let Existential respond with $x_i$ set to any arbitrary $a \in A_1$ (we could imagine this corresponding to Universal instead having played some $(a,b)$ for $y_i$). Finally, Existential plays the remaining (existential) variables of $J$ matching her winning strategy in the game on $I_1$, supposing Universal plays $x_i$ in $I_1$ whatever we just described Existential playing for $x_i$ (associated with $y_i$ in $J$). This is a winning strategy on $J$ by construction of $\Gamma$.

($J \in \QCSP(\Gamma)$ implies $I_1 \in \QCSP(\Gamma_1) \vee I_2 \in \QCSP(\Gamma_2)$.) Consider a winning strategy for Existential in $J \in \QCSP(\Gamma)$ where Universal only played on elements of the form $(a,b)$ where $a \in A_1$, $b \in A_2$. The first variable $x$ of $J$ is existential and indeed is associated with $I_1$. This must be played by Existential in $A_1$ or as $a_2$. If $x$ is evaluated in $A_1$ then the $\sigma$ constraints force all variables associated with $I_2$ to now be $a_1$ and thus all variables associated with $I_1$ to be in $A_1$. Existential can now witness $I_1 \in \QCSP(\Gamma_1)$ by considering that Universal in $I_1$ plays $a$, where Universal in $J$ played $(a,b)$. If $x$ is evaluated to $a_2$, then  the $\sigma$ constraints force all variables associated with $I_2$ to now be in $A_2$ and thus all variables associated with $I_1$ to be $a_2$. Existential can now witness $I_2 \in \QCSP(\Gamma_1)$ by considering Universal in $I_2$ plays $b$, where Universal in $J$ played $(a,b)$.

We can reduce a more complicated set of instances $I_1,\ldots,I_n$ of $\QCSP(\Gamma_1)$ and $J_1,\ldots,J_n$ of $\QCSP(\Gamma_2)$ to $K$ in $\QCSP(\Gamma)$, in such a way that $K \in \QCSP(\Gamma)$ iff $(I_1 \in \QCSP(\Gamma_1) \vee J_1 \in \QCSP(\Gamma_2)) \wedge \ldots \wedge (I_n \in \QCSP(\Gamma_1) \vee J_n \in \QCSP(\Gamma_2))$ by taking the conjunction of our given reduction over each pair $I_i$ and $J_i$.

Now, let us prove that any instance of 
$\QCSP(\Gamma)$ can be reduced to some conjunction of instances of $\QCSP(\Gamma_{1})\vee \QCSP(\Gamma_{2})$. Call an instance $K$ of $\QCSP(\Gamma)$ \emph{connected} if the Gaifman graph of the existential variables of $K$ is connected. 
The number of connected components of $K$ will give the number of conjuncts $\QCSP(\Gamma_{1})\vee \QCSP(\Gamma_{2})$. Let us assume now w.l.o.g. that $K$ is connected otherwise we can split $K$ into a conjunction of its connected instances where each connected instance contains all universal variables but only the atoms containing instances of (universal variables and) its existential variables.  
Notice that all variables in $K$ are typed, in that any variable in a relation from $\Gamma$ 
either takes on values ranging across:
$A_{1}\cup \{a_{2}\}$; or 
$A_{2}\cup \{a_{1}\}$; or
$A$. If a variable appears with more than one type but the types are consistent (i.e. one type is $A$ and the other is one from $A_{1}\cup \{a_{2}\}$ or  $A_{2}\cup \{a_{1}\}$) then this is because the variable appears in some $\sigma_i$ in the second position. But now we could remove this $\sigma_i$ constraint because the other existing type restriction to one of  $A_{1}\cup \{a_{2}\}$ or  $A_{2}\cup \{a_{1}\}$ means $\sigma_i$ will always be satisfied. Furthermore, if some variable has inconsistent types or a fixed element constant appears in a position where it is forbidden due to type, then we know the instance is false. This would also be the case if a universal variable appears in any type other than $A$. We will now assume none of these situations occurs and we term such an input \emph{reduced}.

We would like now to assume that $K$ has no existential variables $x$ in the second position in a $\sigma_i$. First we must argue that if $K$ is reduced then Existential can witness the truth of $K$ while never playing outside of $A_1 \cup A_2 \cup \{a_1,a_2\}$. Suppose Existential ever played outside of this set, then any element in the set could be chosen as a legitimate alternative. Indeed, Existential could only win by playing an element of the form $(a,b)$ in the second position of some $\sigma_i$ and in this circumstance the atom would be equally satisfied by any choice from $A_1 \cup A_2 \cup \{a_1,a_2\}$. Now we can make the assumption that $K$ has no existential variables $x$ in the second position in a $\sigma_i$ because any choice among $A_{1}\cup \{a_{2}\}$ or  $A_{2}\cup \{a_{1}\}$ satisfies this.

Let us remark that Universal has a winning strategy in $K$ iff he has this winning strategy only playing elements of $A_1 \times A_2$. Indeed, we already used this property in the simpler case above. Suppose we have in $K$ some $\sigma_{1}(x_1,y)\wedge \sigma_{1}(x_2,y)$, where $y$ is universally quantified before $x_{1}$ and $x_{2}$. Then adding the constraint $x_{1}=x_{2}$ doesn't change the result (since $y$ will be played in $A_1 \times A_2$). Let us do this and propagate out the innermost of $x_1$ and $x_2$.

Let us process likewise similar instances in $K$ of the form $\sigma_{2}(x_1,y)\wedge \sigma_{2}(x_2,y)$.


Suppose we have in $K$ some $\sigma_{1}(x_{1},y)\wedge \sigma_{2}(x_{2},y)$, where $y$ is universally quantified before $x_{1}$ and $x_{2}$. Note that, since we applied already previous rewriting rules, these are the only occurrences of $y$ in the whole sentence. Owing to this, and since $\forall y \exists x_1 \exists x_2 \sigma_{1}(x_{1},y)\wedge \sigma_{2}(x_{2},y)$ is logically equivalent to $\forall y_1 \forall y_2 \exists x_1 \exists x_2 \sigma_{1}(x_{1},y_1)\wedge \sigma_{2}(x_{2},y_2)$, we may substitute the quantification and atoms in our sentence of the first form for those of the latter form.

Finally, if $\sigma_{1}(x,y)$ appears in the instance $K$ with
$x$ is quantified before $y$ then it is equivalent to 
the substitution $x = a_2$. Similarly, for $\sigma_{2}(x,y)$ with
$x$ is quantified before $y$ then it is equivalent to
the substitution $x = a_1$. In the former case, $K_2$ becomes redundant, and in the latter case, $K_1$ becomes redundant.

We are now in a position to build an instance $K_1 \vee K_2$ of $\QCSP(\Gamma_{1})\vee \QCSP(\Gamma_{2})$. We can now split $K$ into $K_1$ and $K_2$ based on the types of the existential variables using the following additional rule. If $y$ is quantified before $x$ (recall it must be universally quantified) then we may consider this enforces in $K_1$ universal quantification of $x$ but restricted to $A_{1}$. Similarly, with $\sigma_2(x,y)$, and $K_2$ and $A_2$. 

We claim $K \in \QCSP(\Gamma)$ iff $K_1 \in \QCSP(\Gamma_{1})$ or $K_2 \in \QCSP(\Gamma_{2})$.

(Forward.) 
Assume the converse, then there exist
winning strategies for Universal players 
for $K_{1}$ and $K_{2}$.
We need to build a winning strategy for $K$. If one of $K_1$ or $K_2$ is empty, then we may play all existential variables to $a_2$ or $a_1$, respectively.
To do this we apply both strategies (choose different strategies for 
different variables)
until the moment when the first existential variable (let it be $x$) is evaluated.
Recall we assume  existential variable $x$ is either of type  $A_1 \cup \{a_2\}$ or of type  $A_2 \cup \{a_1\}$. W.l.o.g. let it be the former. If $x$ is evaluated in $A_1$ then the Universal player of $K$ uses
the strategy of $K_{1}$, if it is evaluated as $a_2$ then we use the strategy for $K_{2}$. Since $K$ is connected, if $x$ is evaluated in $A_1$ then all variables of type $A_1 \cup \{a_2\}$ must be evaluated in $A_1$, while if $x$ is evaluated as $a_2$ then all variables of type $A_2 \cup \{a_1\}$ must be evaluated from $A_2$ (because $a_1$ can not appear).
Thus the strategy we built is a winning strategy for the Universal player in $K$.

(Backwards.) W.l.o.g. assume $K_1 \in \QCSP(\Gamma_{1})$. Evaluate all variables in $K$ of type $A_2 \cup \{a_1\}$ to $a_1$.  Evaluate all variables in $K$ of type $A_1 \cup \{a_2\}$ in $A_1$ according to the winning strategy for $K_1 \in \QCSP(\Gamma_{1})$.

\end{proof}

Similarly we can prove the following two lemmas.

\begin{lem}\label{ConstraintDisjunctionFirst}
Suppose $\Gamma_{1}$ and $\Gamma_{2}$ are finite constraint languages on sets $A_1$ and $A_{2}$ respectively. 
Then there exists a constraint language $\Gamma$ 
on a domain of size $2\cdot|A_{1}|+|A_{2}|+2$
such that $\QCSP(\Gamma)$ is polynomially equivalent to 
$(\QCSP(\Gamma_{1})\vee \CSP(\Gamma_{2}))
\wedge
\dots
\wedge
(\QCSP(\Gamma_{1})\vee \CSP(\Gamma_{2}))$, 
i.e. the following decision problem:
given $n$,
instances $I_{1},\ldots,I_{n}$ of $\QCSP(\Gamma_{1})$,
and
instances $J_{1},\ldots,J_{n}$
of $\CSP(\Gamma_{2})$;
decide whether
$(I_{1}\vee J_{1})\wedge \dots\wedge (I_{n}\vee J_{n})$
holds.
\end{lem}

\begin{proof}
It is sufficient to define a new language as follows.
Let $A_{1}'$ be a copy of $A_{1}$. For any $a\in A_{1}$ by $a'$ we denote the corresponding element of $A_{1}'$.
Let $$A = A_{1}'\cup A_{1} \cup A_{2}\cup \{a_{1},a_{2}\},$$
$$\sigma= (A_{1} \times \{a_{1}\})\cup 
(\{a_{2}\}\times A_{2}),$$
$$\sigma_{1} = 
\{(a,a')\mid a\in A_{1}\}
\cup 
(\{a_{2}\}\times A)\cup 
(A_{1}\times (A_{1}\cup A_{2}\cup\{a_{1},a_{2}\})),$$
$$\Gamma = \{R\cup\{(a_{2},\ldots,a_{2})\}\mid 
R\in \Gamma_{1}\}
\cup 
\{R\cup\{(a_{1},\ldots,a_{1})\}\mid 
R\in \Gamma_{2}\}
\cup\{\sigma_{1},\sigma\}.$$
\end{proof}

\begin{lem}\label{ConstraintConjunctionTwo}
Suppose $\Gamma_{1}$ and $\Gamma_{2}$ are finite constraint languages on sets $A_1$ and $A_{2}$ respectively. 
Then there exists a constraint language $\Gamma$ 
on a domain of size 
$|A_{1}|\cdot|A_{2}|+|A_{1}|+|A_{2}|$
such that $\QCSP(\Gamma)$ is polynomially equivalent to 
$(\QCSP(\Gamma_{1})\wedge \QCSP(\Gamma_{2}))$,
i.e. the following decision problem:
given an instance $I$ of $\QCSP(\Gamma_{1})$ and 
an instance $J$ of $\QCSP(\Gamma_{2}))$;
decide whether
$I\wedge J$ holds.
\end{lem}

\begin{proof}
It is sufficient to define a new language as follows.
Let $$A = (A_{1}\times A_{2})\cup A_{1} \cup A_{2},$$
$$\sigma_{1} = 
\{(a,(a,b))\mid a\in A_{1}, b\in A_{2}\}
\cup 
(A_{1}\times (A_{1}\cup A_{2})),$$
$$\sigma_{2} = 
\{(b,(a,b))\mid a\in A_{1}, b\in A_{2}\}
\cup 
(A_{2}\times (A_{1}\cup A_{2})),$$
$$\Gamma = \Gamma_1\cup\Gamma_2
\cup\{\sigma_{1},\sigma_{2}\},$$
where we consider any relation from $\Gamma_1\cup\Gamma_2$ as a relation on $A$.
\end{proof}

\section{Reductions and hardness results for the 3-element domain}\label{HardnessSection}

In this section we consider the domain $A= \{0,1,2\}$
and prove all the hardness results we need for the 3-element domain.

\begin{lem}
Suppose $b\in\{0,1\}$,  
$$\begin{array}{c}
\sigma_0 = \{ (a_1,a_2,a_3) : a_1 \in A, a_2,a_3 \in \{b,2\}, (a_1 \in \{0,2\} \vee a_2=a_3) \},\\
\sigma_1 = \{ (a_1,a_2,a_3) : a_1 \in A, a_2,a_3 \in \{b,2\}, (a_1 \in \{1,2\} \vee a_2=a_3) \}. 
\end{array}$$
Then $\QCSP(\{\sigma_0,\sigma_1,\{b\},\{2\}\})$ 
is PSpace-hard.
\label{PSpaceHardness}
\end{lem}

\begin{proof}
 The reduction will be from the complement of (monotone) \emph{Quantified Not-All-Equal 3-Satisfiability} (co-QNAE3SAT) which is co-PSpace-hard (see \cite{papa}) and consequently also PSpace-hard (as PSpace is closed under complement). Consider an instance $I$ of co-QNAE3SAT of the form
 $$Q_{1} x_{1} Q_{1} x_2\dots Q_{n}x_{n}\;
 \mathrm{AE}_{3}(x_{a_1},x_{b_1},x_{c_1})\vee\dots\vee\mathrm{AE}_{3}(x_{a_s},x_{b_s},x_{c_s}),$$
 where 
 $Q_{1},Q_{2},\dots,Q_{n}\in\{\forall,\exists\}$ and
$\mathrm{AE}_3 = \{(0,0,0),(1,1,1)\}$.

By $p_{k}(x_1,\dots,x_{k})$ we denote the predicate 
(relation) on $\{0,1\}$ defined by 
$$Q_{k+1} x_{k+1} Q_{k+2} x_{k+2}\dots Q_{n}x_{n}\;
 \mathrm{AE}_{3}(x_{a_1},x_{b_1},x_{c_1})\vee\dots\vee\mathrm{AE}_{3}(x_{a_s},x_{b_s},x_{c_s}).$$

We inductively define a predicate 
$\delta_{k}(x_1,\dots,x_{k}, y,y')$
for $k=n,n-1,\dots,1,0$
satisfying the following conditions
\begin{enumerate}
    \item 
    if $y,y'\in\{b,2\}$, $y\neq y'$,
    $x_1,\dots,x_{k}\in\{0,1\}$ then 
    $\delta_{k}(x_1,\dots,x_{k}, y,y') = p_{k}(x_1,\dots,x_{k})$;
    \item If $y=y'\in \{b,2\}$ then $\delta_{k}(x_1,\dots,x_{k}, y,y')$ holds;
    
    \item If $\delta_{k}(x_1,\dots,x_{k}, y,y')$ holds, then it holds if we replace some of the values $x_{1},\dots,x_{k}$ by 2;
    \item $\delta_{k}$ is definable by a quantified formula 
    over $\{\sigma_0,\sigma_1,\{b\},\{2\}\}$ that can be efficiently computed.
\end{enumerate}
Put 
\begin{align*}
    \delta_{n}(x_1,\dots,x_{n}&,y_0,y_{2s}) = \exists y_{1}\dots\exists y_{2s-1} \\ &\bigwedge_{1 \le i \le s}
    \begin{array}{c}
    (\sigma_0(x_{a_i},y_{i-1},y_i)\wedge 
    \sigma_0(x_{b_i},y_{i-1},y_i) \wedge \sigma_0(x_{c_i},y_{i-1},y_i)) \\
    (\sigma_1(x_{s+a_i},y_{s+i-1},y_{s+i})\wedge
    \sigma_1(x_{s+b_i},y_{s+i-1},y_{s+i}) \wedge \sigma_1(x_{s+c_i},y_{s+i-1},y_{s+i}))
    \end{array}
\end{align*}

If $Q_{k}=\forall$ then we put 
$\delta_{k-1}(x_1,\dots,x_{k-1},y,y')=
\forall x_{k}\; \delta_{k}(x_1,\dots,x_{k},y,y')$.

If $Q_{k}=\exists$ then we put 
$$\delta_{k-1}(x_1,\dots,x_{k-1},y,y')=
\exists u 
\forall x_{k} \exists u' \; \delta_{k}(x_1,\dots,x_{k},u,u')
\wedge\sigma_0(x_{k},y,u') \wedge \sigma_1(x_{k},y',u').$$

Let us check that $\delta_{n}$ satisfies the above properties (1)-(4). 
If
$p_{n}(x_1,\dots,x_n)$ holds, then there exists 
$i$ such that 
$x_{a_i} = x_{b_i} =x_{c_i}=d$.
If $d =0$ then 
to satisfy $\delta_{n}(x_1,\dots,x_{n}, y_0,y_{2s})$
we 
send $y_{1},y_{2},\dots,y_{i-1}$ to $y_{0}$, and 
$y_{i},y_{i+1},\dots,y_{2s-1}$ to $y_{2s}$.
If $d =1$ then we 
send $y_{1},y_{2},\dots,y_{s+i-1}$ to $y_{0}$, and $y_{s+i},y_{s+i+1},\dots,y_{2s-1}$ to $y_{2s}$.
Suppose $\delta_{n}(x_1,\dots,x_{n}, y_{0},y_{2s})$
holds on $y,y'\in\{b,2\}$, $y\neq y'$,
$x_1,\dots,x_{k}\in\{0,1\}$. Then there should be $i$ such that 
$y_{i-1}\neq y_{i}$.
If $i\le s$  then 
$x_{a_i} = x_{b_i} =x_{c_i}=0$, 
if $i>s$ then 
$x_{a_{i-s}} = x_{b_{i-s}} =x_{c_{i-s}}=1$.
Thus, in both cases 
$p_{n}(x_1,\dots,x_n)$ holds, which completes the proof of (1).
To prove (2)
it is sufficient to send $y_1,\dots,y_{2s-1}$
to $y=y'$.
(3) and (4) follow from the definition 
of $\delta_{n}$.

Let us prove by induction on $k$ that $\delta_{n-1},\dots,\delta_{0}$
satisfy the above properties (1)-(4).
Suppose they hold for $\delta_{k}$.
(4) for $\delta_{k-1}$ follows from the definition and
(4) for $\delta_{k}$.
(3) for $\delta_{k-1}$ follows from the definition and 
the properties of $\sigma_{0}$ and $\sigma_{1}$.
If $Q_{k}=\forall$ then  
(1) for $\delta_{k-1}$ follows 
from (1) and (3)
for $\delta_{k}$.

Suppose $Q_{k} = \exists$.
Notice that 
the Universal player in the definition of 
$\delta_{k-1}$ 
should always play 
$x_{k}=0$ if $u=y$ and 
$x_{k}=1$ if $u=y'$, since otherwise a 
winning strategy for the Existential player is to play $u' = u$.
Thus, the Existential player controls the choice of $x_{k}$, 
and we have
\begin{align*}
    u=y\wedge(\forall x_{k} \exists u' \;  \delta_{k}(x_1,\dots,x_{k},u,u')
\wedge\sigma_0(x_{k},y,u') \wedge
\sigma_1(x_{k},y',u'))
=&\\
\exists u' \delta_{k}(x_1,\dots,x_{k-1},0,y,u')
\wedge\sigma_0(0,y,u') &\wedge
\sigma_1(0,y',u')=\\
&\delta_{k}(x_1,\dots,x_{k-1},0,y,y')
\\u=y'\wedge(\forall x_{k} \exists u' \;  \delta_{k}(x_1,\dots,x_{k},u,u')
\wedge\sigma_0(x_{k},y,u') \wedge
\sigma_1(x_{k},y',u'))
=&\\
\exists u' \delta_{k}(x_1,\dots,x_{k-1},1,y',u')
\wedge\sigma_0(1,y,u') &\wedge
\sigma_1(1,y',u')=\\
&\delta_{k}(x_1,\dots,x_{k-1},1,y',y)
\end{align*}
From the above equations and (1) for $\delta_{k}$ we derive (1)
for $\delta_{k-1}$.
To prove 
(2) for $\delta_{k-1}$ we just check that if
$y=y'$ then 
the Existential player can play $u=u'=y=y'$.
Thus, we proved the properties 
(1)-(4) for $\delta_{k-1}$, 
and by the inductive assumption 
we conclude that they hold for $\delta_{0}$.
Since 
the quantified formula 
    over $\{\sigma_0,\sigma_1,\{b\},\{2\}\}$
    defining $\delta_{0}$ 
is efficiently computable, it remains to check that 
$\exists y\exists y'\;\;\delta_{0}(y,y')\wedge y=b \wedge y=2$
holds if and only if  $I$ holds.
\end{proof}

Recall that an operation $f$ is 01-stable if it is idempotent, $f(x,0,1)=x$, and $f(x,2,2) = 2$.
\begin{lem}\label{ReductionToQCSP2}
Suppose $\Gamma$ is preserved by $s_2$ and a 01-stable operation $h_{01}$.
Then an instance 
$$\forall x_1\exists y_1 \forall x_2 \exists y_2 \dots \forall x_{n} \exists y_{n} \Phi$$
of $\QCSP(\Gamma)$ is equivalent to 
$$\forall x_1\forall x_2 \dots \forall x_{n}
\exists \exists  
((\exists' \exists' \Phi_{1}^0\wedge \Phi_{1}^1)\wedge 
(\exists' \exists' \Phi_2^0 \wedge \Phi_2^1) \wedge \dots \wedge (\exists' \exists' \Phi_n^0\wedge \Phi_n^1)
),$$
where 
$$
\Phi_{i}^0 = \Phi^{x_{i+1},\ldots,x_{n}, y_{i+1},\ldots,y_{n}}_{x_{i+1}^0,\ldots,x_{n}^0,y_{i+1}^0,\ldots,y_{n}^0} \wedge x_{i+1}^0 = 0 \wedge \dots\wedge x_{n}^0 = 0,
$$
$$
\Phi_{i}^1 = \Phi^{x_{i+1},\ldots,x_{n}, y_{i+1},\ldots,y_{n}}_{x_{i+1}^1,\ldots,x_{n}^1,y_{i+1}^1,\ldots,y_{n}^1} \wedge x_{i+1}^1 = 1 \wedge \dots\wedge x_{n}^1 = 1,
$$
(note that $\Phi_{n}^0 =\Phi_{n}^1 = \Phi$) and by 
$\exists \exists$ and $\exists' \exists'$ we mean that we add all necessary existential quantifiers for $y$-variables without a superscript and with a superscript, respectively.
\end{lem}

\begin{proof}

(Forwards/ downwards.)
If we have a solution $(f_1,\ldots,f_{n})$ of the original instance then it is 
also a solution of the new instance with the additional assignments $x_{j}^0 = 0$, $x_{j}^1 = 1$, $y_{j}^0 = f_{j}(x_{1},\ldots,x_{i},0,\ldots,0)$,
and 
$y_{j}^1 = f_{j}(x_{1},\ldots,x_{i},1,\ldots,1)$
in the definition of $\Phi_{i}^0\wedge \Phi_{i}^1$ for every $j> i$.

(Backwards/ upwards.)
First, introduce several notations. 
By $\vee$ denote $s_{2}$, 
by $\le$ denote the partial order on $\{0,1,2\}$ such that 
$0\le 2$ and $1\le 2$ but $0$ and $1$ are
incomparable.
For an operation 
$f(x_1,\dots,x_s)$ 
by $f(x_1,\dots,x_t)$, where $t<s$ we denote 
$\bigvee\limits_{(a_{t+1},\dots,a_{s})}
f(x_1,\dots,x_t,a_{t+1},\dots,a_{s})$.
By $h_{01}'$ we denote the operation defined 
by $h_{01}'(x,y,z) = s_{2}(x,h_{01}(x,y,z))$.
Notice that $h_{01}'$ is also a $01$-stable operation, but it satisfies the 
property 
$h_{01}'(x,y,z)\ge x$.

Let us show how to build a solution to the original instance from 
a solution of the new instance. 
Let a solution of the new instance be defined by 
$y_{i} = f_{i}(x_{1},\ldots,x_{n})$, 
$y_{i}^{0} = g_{j,i}^{0}(x_{1},\ldots,x_{n})$
and 
$y_{i}^{1} = g_{j,i}^{1}(x_{1},\ldots,x_{n})$
for 
the variables $y_{i}^{0}$ and $y_{i}^{1}$
from $\Phi_{j}^{0}\wedge \Phi_{j}^{1}$.

First, let us show that 
$$(x_1,\dots,x_{j}, c,\dots,c,
f_{1}(x_1,\dots,x_{j}),\dots,
f_{j}(x_1,\dots,x_{j}),
g_{j,j+1}^{c}(x_1,\dots,x_{j}),
\dots,
g_{j,n}^{c}(x_1,\dots,x_{j}))$$
is a solution of $\Phi$ 
for every $j\in[n]$ and $c\in\{0,1\}$.
In fact, from $\Phi_{j}^{c}$ 
the tuple 
$$(x_1,\dots,x_{j}, c,\dots,c,
f_{1}(x_1,\dots,x_{n}),\dots,
f_{j}(x_1,\dots,x_{n}),
g_{j,j+1}^{c}(x_1,\dots,x_{n}),
\dots,
g_{j,n}^{c}(x_1,\dots,x_{n}))$$
is a solution of $\Phi$. 
Consider all the evaluations of the variables 
$x_{j+1},\ldots,x_{n}$ to obtain $3^{n-j}$ solutions of $\Phi$, 
and apply the semilattice $s_2$ to them. 
As a result we obtain the required tuple, which implies that this 
tuple satisfies $\Phi$.

We prove by induction on $N=0,1,\dots,n$ that 
there exist operations $q_{1},\dots,q_{n}$
satisfying 
\begin{enumerate}
    \item 
    $(x_1,\dots,x_{n},
q_{1}(x_1,\dots,x_n),\dots,q_{n}(x_1,\dots,x_{n}))$ is a solution of $\Phi$, 
\item 
$q_{i}(x_{1},\dots,x_{n})\ge f_{i}(x_{1},\dots,x_{n})$
for every $i\in[n]$, 
\item 
$q_{i}(x_{1},\dots,x_{n})$ depends fictitiously on 
$x_{i+1},\dots,x_{n}$ for every $i\in[N]$,
\item 
$q_{i}(x_{1},\dots,x_{n})$ is uniquely determined 
by 
$x_1,\dots,x_{i}$ and $f_{i}(x_{1},\dots,x_{n})$
for every $i\in[n]$.
\end{enumerate}

For $N=0$ we can 
put $q_{i}(x_1,\dots,x_{n}) = f_{i}(x_1,\dots,x_{n})$ for every $i\in[n]$.
Let us prove the induction step.
Assume that 
$q_1,\dots,q_n$ satisfy conditions (1)-(4) for some $N$. 
Apply the 01-stable operation $h_{01}'$ to 
the three solutions of $\Phi$
\begin{align*}
(x_1,\dots,x_{n},
q_{1}(x_1,\dots,x_n),\dots,
q_{n}(x_1,\dots,x_{n}&))\\
(x_1,\dots,x_{N+1}, 0,\dots,0,
f_{1}(x_1,\dots,x_{N+1}),\dots,
&f_{N+1}(x_1,\dots,x_{N+1}),\\
&g_{N+1,N+2}^{0}(x_1,\dots,x_{N+1}),
\dots,
g_{N+1,n}^{0}(x_1,\dots,x_{N+1}))\\
(x_1,\dots,x_{N+1}, 1,\dots,1,
f_{1}(x_1,\dots,x_{N+1}),\dots,
&f_{N+1}(x_1,\dots,x_{N+1}),\\
&g_{N+1,N+2}^{1}(x_1,\dots,x_{N+1}),
\dots,
g_{N+1,n}^{1}(x_1,\dots,x_{N+1}))    
\end{align*}

The first 
$n$ coordinates of the obtained tuple are
$(x_1,\dots,x_n)$ and we 
denote this tuple  
by 
$(x_1,\dots,x_{n},
q_{1}'(x_1,\dots,x_n),\dots,
q_{n}'(x_1,\dots,x_{n}))$.
Let us show that 
$q_1',\dots,q_{n}'$ satisfy (1)-(4) for $N:= N+1$.
Condition (1) follows from the fact that $h_{01}'$ preserves $\Gamma$.
Condition (2) follows from the 
fact that $q_{1},\dots,q_{n}$ satisfy (2) and $h_{01}'(x,y,z)\ge x$.

For $i\in [N]$ we have
$$q_{i}'(x_1,\dots,x_{n}) = h_{01}'(
q_{i}(x_1,\dots,x_{n}),
f_{i}(x_1,\dots,x_{N+1}),
f_{i}(x_1,\dots,x_{N+1})).$$
Let us show that 
$q_{i}'$
 depends essentially only on 
$x_1,\dots,x_{i}$.
Since 
$q_{i}$ depends fictitiously on $x_{i+1},\dots,x_n$
and 
$q_{i}(x_1,\dots,x_n)\ge f_{i}(x_1,\dots,x_n)$, 
we have 
$q_{i}(x_1,\dots,x_{n})=q_{i}(x_1,\dots,x_{i})\ge f_{i}(x_1,\dots,x_{i})$.
If 
$f_{i}(x_1,\dots,x_{i})=2$ then 
$q_{i}(x_1,\dots,x_{i})=2$ 
and 
$q_{i}'(x_1,\dots,x_{n})=2$.
If $f_{i}(x_1,\dots,x_{i})=d\neq 2$, then 
$f_{i}(x_1,\dots,x_{N+1})=d$ and
$q_{i}(x_1,\dots,x_{i})\in\{d,2\}$.
Therefore, 
$q_{i}'(x_1,\dots,x_{n}) = q_{i}(x_1,\dots,x_{i})$
and 
$q_{i}'$ depends fictitiously on $x_{i+1},\dots,x_{n}$
and satisfies (3) and (4).

We have  $$q_{N+1}'(x_1,\dots,x_{n}) = h_{01}'(
q_{N+1}(x_1,\dots,x_{n}),
f_{N+1}(x_1,\dots,x_{N+1}),
f_{N+1}(x_1,\dots,x_{N+1})).$$
If 
$f_{N+1}(x_1,\dots,x_{N+1})=2$ then  
$q_{N+1}'(x_1,\dots,x_{n})$ equals 2, otherwise it is uniquely determined by $x_{1},\dots,x_{N+1}$, and 
$f_{N+1}(x_1,\dots,x_{n})= f_{N+1}(x_1,\dots,x_{N+1})$. 
That is, in all the cases $q_{N+1}'$ is uniquely determined by 
$x_1,\dots,x_{N+1}$
and satisfies (3) and (4).

We have $q_{i}'(x_1,\dots,x_{n}) = h_{01}'(
q_{i}(x_1,\dots,x_{n}),
g_{N+1,i}^0(x_1,\dots,x_{N+1}),
g_{N+1,i}^0(x_1,\dots,x_{N+1}))$
for $i>{N+1}$. 
Since $q_{i}$ is uniquely determined 
by $x_1,\dots,x_{i}$ and $f_{i}(x_1,\dots,x_{n})$, 
the same is true for $q_{i}'$, that is, $q_{i}'$ satisfies (4).
Thus, we defined operations
$q_{1}',\dots,q_{n}'$ satisfying conditions (1)-(4) with $N:=N+1$. 

Hence, by the induction 
there exist operations 
$q_1,\dots,q_n$ satisfying conditions (1)-(4) for $N =n$.
We may check that 
$(x_1,\dots,x_{n},q_{1}(x_1,\dots,x_n),\dots,
q_{n}(x_1,\dots,x_{n}))$ is a solution of the original instance, 
which completes the proof.
\end{proof}

The next lemma follows from Lemma~\ref{ReductionToQCSP2} and the fact 
that if $\Gamma$ is preserved by a semilattice then $\CSP(\Gamma,0,1)$ can be solved in polynomial time.
Nevertheless, to explain how a $01$-stable operation can be used 
in an algorithm we give an alternative proof.

\begin{lem}\label{InCONPForStable}
Suppose $\Gamma$ is preserved by $s_2$ and a 01-stable operation $g$.
Then 
$\QCSP(\Gamma)$ is in co-NP.
\end{lem}

\begin{proof}
Suppose we have an instance 
$\forall x_1\exists y_1 \forall x_2 \exists y_2 \dots \forall x_{n} \exists y_{n} \Phi$.
We can use an oracle to choose an appropriate value for $x_{1}$ (let this value be $a_{1}$).
Then we need to find an appropriate value for $y_{1}$, such that we can use an oracle for $x_2$ and continue.
We want to be sure that if the instance holds then it holds after fixing $y_{1}$. 

To find out how to fix $y_{1}$ we check the satisfability of the two instances 
$$\Phi\wedge x_{1} = a_{1} \wedge x_{2} = x_{3}=\dots=x_{n} = 0,$$
$$\Phi\wedge x_{1} = a_{1} \wedge x_{2} = x_{3}=\dots=x_{n} = 1.$$
These are CSP instances, which can be solved in polynomial time 
because the semilattice preserves $(\Gamma,0,1)$.
We check whether both have a solution with $y_{1}=0$, $y_{1} =1$, $y_{1}=2$ (we solve six instances).
Let $Y$ be the set of possible values for $y_{1}$ such that both instances have solutions.
If $Y$ is empty, then the QCSP instance does not hold and we are done.
If $|Y|=1$, then we fix $y_{1}$ with the only value in $Y$. Obviously, 
the fixing of $y_{1}$ cannot transform  
the QCSP instance that holds into an instance that does not hold.
If $|Y|>1$ then $2\in Y$ due to the semilattice polymorphism. 
Let the solutions of the CSP instances with $y_{1} = 2$ be 
$(x_{1},y_{1},x_{2},y_{2},\dots,x_{n},y_{n}) = 
(a_{1},2, 0, c_{2},\dots,0, c_{n})$
and $(x_{1},y_{1},x_{2},y_{2},\dots,x_{n},y_{n}) = 
(a_{1},2, 1, d_{2},\dots,1, d_{n})$.
Assume that the QCSP instance has a solution 
$(a_{1},f_{1} (a_{1}), x_{2},f_{2}(x_{1},x_{2}), \ldots,x_{n},f_{n}(
x_{1},\ldots,x_n))$.
Then by applying the operation $g$ to this solution,
$(a_{1},2, 0, c_{2},\dots,0, c_{n})$, and 
$(a_{1},2, 1, d_{2},\dots,1, d_{n})$,
we get a (partial) solution of the $\QCSP(\Gamma)$
with $y_{1}=2$.

We proceed this way through the quantifier prefix, using an oracle to choose values for  $x_{2}, \ldots,x_{n}$,
while we solve CSP instances to 
choose appropriate values for 
$y_{2},y_{3},\ldots,y_{n}$.
\end{proof}

We say that a ternary relation $R$ is \emph{an AND-type} relation if 
$R\cap (\{0,1\}\times\{0,1\}\times\{0,1,2\})
=\begin{pmatrix}
0&0&1&1\\
0&1&0&1\\
0&0&0&1
\end{pmatrix}$, that is, it is row-wise the truth table of AND.
If $s_2$ preserves $R$, 
then it also contains $(2,2,0)$ and $(2,2,2)$.
We say that a ternary relation $R$ is \emph{an OR-type} relation if 
$R\cap (\{0,1\}\times\{0,1\}\times\{0,1,2\})
=\begin{pmatrix}
0&0&1&1\\
0&1&0&1\\
0&1&1&1
\end{pmatrix}$, that is, it is row-wise the truth table of OR.
Again, if $s_2$ preserves $R$, 
then it also contains $(2,2,1)$ and $(2,2,2)$. One can readily imagine AND/OR-type relations of arity $k+1$ higher than three built from AND/OR operations of arity $k$. 

\begin{lem}\label{ANDORHardness}
Suppose $\Gamma$ contains an AND-type relation and an OR-type relation,
and $\Gamma$ is preserved by $s_2$.
Then $\QCSP(\Gamma)$ is co-NP-hard.
\end{lem}
\begin{proof}

By $R_{and,2}$ and 
$R_{or,2}$ we 
denote an AND-type and an OR-type relation, respectively.
For $n=2,3,4,\dots$ put 
$$R_{and,n+1}(x_{1},\ldots,x_{n},x_{n+1},y) = 
\exists z \;R_{and,n}(x_{1},\ldots,x_{n},z)
\wedge R_{and,2}(x_{n+1},z,y),$$
$$R_{or,n+1}(x_{1},\ldots,x_{n},x_{n+1},y) = 
\exists z \; R_{or,n}(x_{1},\ldots,x_{n},z)
\wedge R_{or,2}(x_{n+1},z,y).$$

Let us define a relation $\xi_{n}$ for every $n$ by
\begin{align*}
\xi_{n}(x_{1},y_{1},z_{1},\dots,
x_{n},y_{n},z_{n})&=
\exists u\exists u_{1}\dots\exists u_{n}
\exists v\exists v_{1}\dots\exists v_{n}\;
R_{and,2}(u,v,v) \wedge \\
R_{and,3}(x_{1},y_{1},z_{1},u_{1})&\wedge
\dots\wedge
R_{and,3}(x_{n},y_{n},z_{n},u_{n})\wedge
R_{or,n}(u_1,\dots,u_{n},u)\wedge\\
R_{or,3}(x_{1},y_{1},z_{1},v_{1})&\wedge
\dots\wedge
R_{or,3}(x_{n},y_{n},z_{n},v_{n})\wedge
R_{and,n}(v_1,\dots,v_{n},v).
\end{align*}

It follows from the definition 
that 
$\xi_{n}\cap \{0,1\}^{3n}$ is defined by
$\mathrm{AE}_3(x_{1},y_{1},z_{1})\vee
\mathrm{AE}_3(x_{2},y_{2},z_{2})\vee
\dots
\vee
\mathrm{AE}_3(x_{n},y_{n},z_{n})$,
where 
$\mathrm{AE}_3 = \{(0,0,0),(1,1,1)\}$.

Suppose 
$R(y_1,\dots,y_{t}) = 
\xi_{n}(y_{i_1},y_{i_{2}},\dots,y_{i_{3n}})$, 
where 
$i_{1},\dots,i_{3n}\in\{1,2,\dots,t\}$.
Let us show that 
$\forall y_{1}\dots\forall y_{t}\;
\xi_{n}(y_{i_1},y_{i_{2}},\dots,y_{i_3n})$
holds if 
$\xi_{n}(y_{i_1},y_{i_{2}},\dots,y_{i_3n})$
holds for all $y_{1},\ldots,y_{t}\in\{0,1\}$.
We need to prove that if $\{0,1\}^{t}\subseteq R$ then 
$A^{t}\subseteq R$, which follows from the fact that the semilattice $s_2$ preserves $R$.

Now we may encode the complement of \emph{Not-All-Equal 3-Satisfiability} using $\Gamma$.
This complement can be expressed by
a formula of the following form:
$$\forall y_{1}\dots\forall y_{t} \ 
\mathrm{AE}_3(y_{i_1},y_{i_{2}},y_{i_3})\vee
\dots
\vee
\mathrm{AE}_3(y_{i_{3n-2}},y_{i_{3n-1}},y_{i_{3n}}),$$
which is equivalent to 
$$\forall y_{1}\dots\forall y_{t}\;
\xi_{n}(y_{i_1},y_{i_{2}},\dots,y_{i_{3n}}).$$
Thus, we reduced a co-NP-complete problem to $\QCSP(\Gamma)$, which completes the proof.
\end{proof}

\begin{lem}\label{NoStrange2Hardness}
Suppose $\Gamma\subseteq\Inv(s_2)$, 
$\Gamma$
contains the constant $0$, 
$\delta(x,y,z) = (x\neq 0)\vee (y=z)$ 
and an AND-type relation $R_{2}$. 
Then $\QCSP(\Gamma)$ is co-NP-hard.
\end{lem}
\begin{proof}
Here we will define a reduction from the complement of 
$\CSP(\{1\mathrm{IN}3\})$, which is known  to be NP-complete \cite{Schaefer},
where 
$1\mathrm{IN}3 = \begin{pmatrix}
1&0&0\\
0&1&0\\
0&0&1
\end{pmatrix}$.
First, 
put 
$$\delta_1(x_1,x_2,x_3,x_4) = \exists t\;\delta(x_1,x_{3},t) \wedge \delta(x_{2},t,x_4).$$
It is not hard to see that 
$\delta_1(x_1,x_2,x_3,x_{4}) = 
(x_{1}\neq 0)\vee 
(x_{2}\neq 0)\vee
(x_{3}=x_{4}).$
Put 
$$\delta_{2}(x_{1},x_{2},x_{3},x_{4}) = 
\delta_{1}(x_{1},x_{2},x_{3},x_{4})
\wedge 
\delta_{1}(x_{1},x_{3},x_{2},x_{4})
\wedge 
\delta_{1}(x_{2},x_{3},x_{1},x_{4}).
$$
By the definition, the first three variables of
$\delta_{2}$ are symmetric.
Suppose $(a,b,c,d)\in\delta_{2}$. 
It is not hard to see that
$a=b=c=0$ implies $d=0$.
Also 
$a=1, b=c=0$ implies $d=1$.
If there are at least two 1s in $(a,b,c)$ then 
$d$ can be arbitrary.
Put $R'(x,y) = R_{2}(x,x,y)$.
By the definition of an AND-type relation, 
$R'$ contains $(0,0),(1,1),(2,2),(2,0)$
but doesn't contain $(0,1),(0,2),(1,0),(1,2)$.
Put
$$\delta_{3}(x_{1},x_{2},x_{3},x_{4}) = 
\exists x_{1}'\exists x_{2}'\exists x_{3}'\;
\delta_{2}(x_{1}',x_{2}',x_{3}',x_{4})
\wedge 
R'(x_{1},x_1')
\wedge 
R'(x_{2},x_2')
\wedge 
R'(x_{3},x_3').$$
Then 
$(a,b,c,0)\in\delta_{3}$
for all tuples $(a,b,c)$ but the tuples from $1\mathrm{IN}3$.

Again, we apply the recursive 
formula to build an $n$-ary ``and''
$$R_{n+1}(x_{1},\ldots,x_{n+1},z) = 
\exists t\; R_{n}(x_{1},\ldots,x_{n},t)\wedge
R_{2}(t,x_{n+1},z).$$

Define 
$\zeta_{n}(x_{1},\ldots,x_{3n})$ by
$$\exists z_{1}\dots\exists z_{n}\exists t\;
R_{n}(z_{1},\ldots,z_{n},t)
\wedge 
\delta_{3}(x_{1},x_{2},x_{3},z_{1})\wedge
\dots\wedge \delta_{3}(x_{3n-2},x_{3n-1},x_{3n},z_{n})
\wedge (t=0).
$$

Now we may encode the complement of $\CSP(\{1\mathrm{IN}3\})$ using $\Gamma$. If $\overline{1\mathrm{IN}3}$ is the complement of $1\mathrm{IN}3$ with respect to $\{0,1\}^3$, then we use
a formula of the following form:
$$\forall y_{1}\dots\forall y_{t} \ 
\overline
{1\mathrm{IN}3}(y_{i_1},y_{i_{2}},y_{i_3})\vee
\dots
\vee
\overline{1\mathrm{IN}3}(y_{i_{3n-2}},y_{i_{3n-1}},y_{i_{3n}}),$$
which is equivalent to 
$$\forall y_{1}\dots\forall y_{t}\;
\xi_{n}(y_{i_1},y_{i_{2}},\dots,y_{i_{3n}}).$$
Thus, we reduced a co-NP-complete problem to $\QCSP(\Gamma)$, which completes the proof.
\end{proof}

\begin{lem}
Suppose $\sigma = \{0,2\}^2 \cup \{1,2\}^{2}$,  
$$\begin{array}{c}
\sigma_0 = \{ (a_1,a_2,a_3) : a_1 \in A, a_2,a_3 \in \{1,2\}, (a_1 \in \{0,2\} \vee a_2=a_3) \},\\
\sigma_1 = \{ (a_1,a_2,a_3) : a_1 \in A, a_2,a_3 \in \{0,2\}, (a_1 \in \{1,2\} \vee a_2=a_3) \}. 
\end{array}$$
Then $\QCSP(\{\sigma, \sigma_0,\sigma_1,\{0\},\{1\}\})$ 
is co-NP-hard.
\label{lem:coNP-hardness-oldPSpace}
\end{lem}

\begin{proof}
Let us encode the complement of \emph{Not-All-Equal 3-Satisfiability} using $\{\sigma, \sigma_0,\sigma_1,\{0\},\{1\}\}$.
An instance $\mathcal I$ of this complement can be expressed by
a formula of the following form:
$$\forall x_{1}\dots\forall x_{t} \ 
\mathrm{AE}_3(x_{a_1},x_{b_{1}},x_{c_1})\vee
\dots
\vee
\mathrm{AE}_3(x_{a_s},x_{b_s},x_{c_s}),$$
where $\mathrm{AE}_3 = \{(0,0,0),(1,1,1)\}$.
Let
\begin{align*}
    \Phi_0= & \bigwedge_{1 \le i \le s} (\sigma_0(x_{a_i},y_{i-1},y_i)\wedge
    \sigma_0(x_{b_i},y_{i-1},y_i) \wedge \sigma_0(x_{c_i},y_{i-1},y_i))\\
    \Phi_1= & \bigwedge_{1 \le i \le s} (\sigma_1(x_{a_i},z_{i-1},z_i)\wedge
    \sigma_1(x_{b_i},z_{i-1},z_i) \wedge \sigma_1(x_{c_i},z_{i-1},z_i))\\
\Psi =& \forall x_{1}\dots\forall x_{t}\;
\exists y_{0}\exists y_1\dots\exists y_{s}\;
\exists z_{0}\exists y_1\dots\exists z_{s}\;
(y_{0} = 1\wedge \Phi_{0}\wedge z_{0} = 0\wedge \Phi_1 \wedge \sigma(y_{s},z_{s})).
\end{align*}
Let us show that $\mathcal I$ is equivalent to $\Psi$.

($I\in \mathrm{co\mbox{-}NAE3SAT}$ implies $\Psi \in \QCSP(\{\sigma,\sigma_0,\sigma_1\})$).
Suppose in $\Psi$ that Universal plays some sequence of $0$s, $1$s and $2$s for the $x_i$s. Let this be duplicated by Universal in $I$ where we additionally consider $2$ in $\Psi$ as $1$ in $I$. We may assume that the resulting instance of  $\mathrm{co\mbox{-}NAE3SAT}$ has a clause $j$ that is either $(0,0,0)$ or $(1,1,1)$. Suppose first it is $(1,1,1)$, which arises from a setting in $\Psi$ of $(x_{a_j},x_{b_j},x_{c_j})\in \{1,2\}^3$.  Existential plays in $\Psi$ as follows. $y_0,\ldots,y_{j-1},y_j,\ldots,y_s$ are set to $1$, $z_0,\ldots,z_{j-1}$ are set to $0$ and $z_j,\ldots,z_s$ are set to $2$. Let us look at the interesting conjunction in $\Phi_1$, which is $\sigma_1(x_{a_j},z_{j-1},z_j)\wedge
    \sigma_1(x_{b_j},z_{j-1},z_j) \wedge \sigma_1(x_{c_j},z_{j-1},z_j)$, and note that this is satisfied trivially as $(x_{a_j},x_{b_j},x_{c_j})\in \{1,2\}^3$. The fact that the other parts are satisfied is straightforward since $y_{i}=y_{i-1}$ and $z_{i}=z_{i-1}$. Suppose now that it is $(0,0,0)$, which  arises only from a setting in $\Psi$ of $(x_{a_j},x_{b_j},x_{c_j})=(0,0,0)$.  Existential plays in $\Psi$ as follows. $z_0,\ldots,z_{j-1},z_j,\ldots,z_s$ are set to $0$, $y_0,\ldots,y_{j-1}$ are set to $1$ and $y_j,\ldots,y_s$ are set to $2$. This case concludes as the previous.

($\Psi \in \QCSP(\{\sigma,\sigma_0,\sigma_1\})$ implies $I\in \mathrm{co\mbox{-}NAE3SAT}$). Let us consider Universal playing in $I$ and we duplicate this in $\Psi$ (in other words, we only consider Universal plays in $\Psi$ on $\{0,1\}$). We cannot have $y_0,\ldots,y_s=1$ and $z_0,\ldots,z_s=0$ because we violate $\sigma(y_s,z_s)$. Thus, at least at some point in one of these sequences, we make a transition to $2$. First, let us assume that it is in the $y_i$ at the point $y_{j-1}=1$ and $y_j=2$. Then the conjunction $\sigma_0(x_{a_j},y_{j-1},y_j)\wedge
    \sigma_0(x_{b_j},y_{j-1},y_j) \wedge \sigma_0(x_{c_j},y_{j-1},y_j)$ enforces that $(x_{a_j},x_{b_j},x_{c_j})=(0,0,0)$. Now let us assume that it is in the $z_i$ at the point  $z_{j-1}=0$ and $z_j=2$. Then the conjunction $\sigma_1(x_{a_j},z_{j-1},z_j)\wedge
    \sigma_1(x_{b_j},z_{j-1},z_j) \wedge \sigma_1(x_{c_j},z_{j-1},z_j)$ enforces that $(x_{a_j},x_{b_j},x_{c_j})=(1,1,1)$. Thus, in both cases, we deduce that $I$ is a no-instance of NAE3SAT.

Thus, we reduced a co-NP-complete problem to $\QCSP(\{\sigma, \sigma_0,\sigma_2,\{0\},\{1\}\})$, which completes the proof.
\end{proof}
\section{Strange structure 1}\label{StrangeOneSection}

In this section we will define a constraint language $\Gamma$
consisting of just 2 relations and constants such that $\Pol(\Gamma)$ has the EGP 
property but every pp-definition of $\tau_{n}$ (see Definition~\ref{def:tau}) has at least 
$2^{n}$ existential quantifiers. Moreover, we will show that 
$\QCSP(\Gamma)$ can be solved in polynomial time.

Let $A=\{0,1,2\}$. 
Recall that 
$\tau_{n}$ for $\alpha =\{0,2\}$ and $\beta = \{1,2\}$
is
the $3n$-ary relation 
defined by 
$$\{(x_{1},y_{1},z_{1},x_{2},y_{2},z_{2},\dots,x_{n},y_{n},z_{n})\mid 
\exists i \colon \{0,1\}\not\subseteq \{x_{i},y_{i},z_{i}\}\}.$$
By $\sigma_{n}$  we denote the $2n$-ary relation 
defined by 
$$\{(x_{1},y_{1},x_{2},y_{2},\dots,x_{n},y_{n})\mid 
\exists i \colon \{x_{i},y_{i}\} \neq\{0,1\}\}.$$
Note $\tau_{n}$ can be pp-defined from $\sigma_{n}$ but the obvious definition 
is of size exponential in $n$ (see \cite{MFCS2017}).
At the same time, $\sigma_{n}$ can be pp-defined from $\tau_{n}$ 
by identification of variables.
Let 
$$R_{and,2} = 
\{(0,0,0),(1,0,0), (0,1,0),(1,1,1)\}
\cup \{(2,a,b),(a,2,b)\mid a,b\in A\},$$
$$\delta = \{(0,0),(1,0),(2,0),(1,2),(2,2)\},$$
$$\Gamma = \{R_{and,2},\delta,\{0\},\{1\},\{2\}\}.$$

The relation $\rho$ of arity $2n$ omitting just one tuple
$1^{n}0^{n}$ can be pp-defined over $\Gamma$ as follows.
First, as usual, we define an $n$-ary ``and'' by the following recursive formula
$$R_{and,n+1}(x_{1},\ldots,x_{n},x_{n+1},y) = 
\exists z R_{and,n}(x_{1},\ldots,x_{n},z)
\wedge R_{and,2}(x_{n+1},z,y),$$
Then $\rho$ can be defined by 
\begin{multline*}
\rho(x_1,\ldots,x_n,y_1,\ldots,y_n) = 
\exists y_{1}'\dots\exists y_{n}'\exists z\exists t\;
R_{and,n}(x_{1},\dots,x_{n},z)
\wedge \\
\delta(y_1,y'_1)
\wedge 
\dots
\wedge 
\delta(y_n,y'_n)
\wedge
R_{and,n}(y_{1}',\dots,y_{n}',t)
\wedge R_{and,2}(z,z,t).
\end{multline*}

As a conjunction of $2^{n}$ 
copies of $\rho$
with permuted variables we can define the relation $\sigma_{n}$ but this definition will be of exponential size. 
Then, we know from \cite{ZhukGap2015} that 
$\Pol(\Gamma)$ has the EGP property, 
and from \cite{MFCS2017}
that $\tau_{n}$ can be pp-defined from $\Gamma$.
Below we will prove 
that any pp-definition of $\sigma_{n}$ and $\tau_{n}$ is of exponential size,
as well as the fact that 
$\QCSP(\Gamma)$ can be solved in polynomial time. In the following $<$ is the coordinatewise partial  order on $\{0,1\}^n$ built from $0<1$.
Recall that for a conjunctive formula $\Phi$
by 
$\Phi(x_{1},\dots,x_{n})$ we denote the set of all tuples 
$(a_{1},\dots,a_{n})$ such that 
$\Phi$ has a solution with 
$(x_{1},\dots,x_{n}) = (a_{1},\dots,a_{n})$.

\begin{lem}\label{StrangeStructureMainLemma}
Suppose 
$R=\Phi(x_{1},\ldots,x_{n})$, where $\Phi$ is a conjunctive formula over $\Gamma$,
$\alpha\in \{0,1\}^{n}\setminus R$,
there exists $\beta\in \{0,1\}^{n}\cap R$ such that 
$\beta<\alpha$ and 
there exists $\beta\in \{0,1\}^{n}\cap R$ such that 
$\beta>\alpha$.
Then there exists a variable $y$ in $\Phi$, such that 
for $R'= \Phi(x_{1},\ldots,x_{n},y)$ 
we have the following property
$$\beta\in \{0,1\}^{n}\wedge (\beta<\alpha)\wedge 
\beta d\in R' \Rightarrow  d = 0,$$
$$\beta\in \{0,1\}^{n}\wedge (\beta>\alpha)\wedge 
\beta d\in R' \Rightarrow  d = 1.$$
\end{lem}
\noindent 
Informally speaking, this lemma says that 
whenever we have a tuple outside of a relation 
there should 
be a variable in its pp-definition distinguishing between smaller and greater tuples of the relation.

\begin{proof}
For every variable $y$ of $\Phi$ 
let $C_{y}$ be the set of all elements $d$ such that 
there exists $\beta\in \{0,1\}^{n}\cap R$, $\beta<\alpha$
and $\Phi$ has a solution 
with $y=d$ and $(x_{1},\ldots,x_{n}) = \beta$.
Similarly, let $D_{y}$ be the set of all elements $d$ such that 
there exists $\beta\in \{0,1\}^{n}\cap R$, $\beta>\alpha$
and $\Phi$ has a solution 
with $y=d$ and $(x_{1},\ldots,x_{n}) = \beta$.
Thus, we need to prove that there exists a variable $y$ such that $C_{y} = \{0\}$ and $D_{y} = \{1\}$.

Then we assign a value $v(y)$ to every variable $y$ in the following way:
if $C_{y} = \{0\}$ then put $v(y):=0$;
otherwise, if $C_{y} \subseteq \{0,1\}$ 
and $D_{y} = \{1\}$ then put $v(y):=1$;
otherwise put $v(y):=2$.

If $\alpha(i) = 0$ then $C_{x_{i}} =\{0\}$ and 
$v(x_{i}) = 0$.
If $\alpha(i) = 1$ and then $C_{x_{i}} \subseteq\{0,1\}$  
$D_{x_{i}} = \{1\}$. 
If $C_{x_{i}} = \{0\}$ then we found the required variable 
and we are done;
otherwise we have
$v(x_{i}) = 1$.
Since $\alpha\notin R$, 
$v$ cannot be a solution of $\Phi$, therefore
$v$ breaks at least one of the relations in $\Phi$.
We consider several cases:
\begin{enumerate}
    \item The corresponding relation is $y=a$ for some $a$.
    If $a =0$ then $C_{y}=\{0\}$ and $v(y) = 0$,
    if $a =1$ then $C_{y}=D_{y} = \{1\}$ and $v(y) = 1$, 
    if $a =2$ then $C_{y}=\{2\}$ and $v(y) = 2$.
    Thus, the evaluation $v$ cannot break the relation 
    $y=a$.
    \item The corresponding relation is 
    $R_{and,2}(y_1,y_2,y_3)$.
    Assume that $v(y_{1}) = 0$ and $v(y_{2})\in \{0,1\}$.
    Then $C_{y_{1}}=\{0\}$ and $C_{y_{2}}\subseteq\{0,1\}$,
    which means that on all tuples $\beta<\alpha$ the 
    value of $y_{3}$ should be equal to 0. Hence 
    $C_{y_{3}} =\{0\}$ and $v(y_3) = 0$.
    If $v(y_{1}) = 2$ or $v(y_{2})=2$,
    then we cannot break the relation $R_{and,2}$.
    The only remaining case is when 
    $v(y_1) = v(y_2) = 1$, which means that 
    $C_{y_{1}},C_{y_{2}}\subseteq\{0,1\}$
    and 
    $D_{y_{1}}= D_{y_{2}} = \{1\}$.
    This implies that $C_{y_{3}}\subseteq\{0,1\}$
    and 
    $D_{y_{3}} = \{1\}$. 
    If $C_{y_{3}} = \{0\}$, 
    then $y_{3}$ is the variable we were looking for.
    Otherwise, the evaluation of $y_{3}$ is 1, which agrees with the definition of $R_{and,2}$.
    \item The corresponding relation 
    is $\delta(y_1,y_{2})$.
    If $v(y_{1})=0$ then $C_{y_{1}} = \{0\}$,
    and by the definition of $\delta$ we have $C_{y_{2}} = \{0\}$,
    which means that $v(y_{2})=0$.
    If $v(y_{1})\neq 0$,
    it follows from the fact that 
    $v(y_2)$ cannot be outside of $\{0,2\}$.
\end{enumerate}
\end{proof}

We may check that 
operations $s_{0,2}$ and $g_{0,2}$ are polymorphisms 
of $\Gamma$. Moreover, we will show later that 
$\Pol(\Gamma) = \Clo(\{s_{0,2},g_{0,2}\})$.
Put 
$$h_{0,2}(x,y,z) = \begin{cases}
x, & \text{if $x=z = 0$}\\
x, & \text{if $x=1, y =z\in\{0,1\}$}\\
2, & \text{otherwise.}
\end{cases}.$$

Since, 
$s_2(x,y) = s_{0,2}(y,s_{0,2}(x,y))$ and
$h_{0,2}(x,y,z) = 
g_{0,2}(s_{0,2}(x,z),s_2(y,z))$,
the operations $s_{2}$ and $h_{0,2}$ are also polymorphisms 
of $\Gamma$.

The following lemma and corollary do not play a role in our main result but we include them for their intrinsic intriguingness as well as by way of a sanity check.
\begin{lem}
Any pp-definition of $\sigma_{n}$ over $\Gamma$, where $n\ge 3$,  has at least $2^n$ variables.
\end{lem}
\begin{proof}
Let the pp-definition be given by a conjunctive formula $\Phi$ such that $\sigma_{n} = 
\Phi(x_{1},\dots,x_{2n})$.
By Lemma~\ref{StrangeStructureMainLemma}
for any $\alpha\in\{0,1\}^{2n}\setminus \sigma_{n}$ there should be a variable $y$ 
such that if we define the relation
$R' = \Phi(x_{1},\ldots,x_{2n},y)$, 
then 
for every $\beta<\alpha$ (we consider only tuples from 
$\{0,1\}^{2n}$) we have 
$\beta d\in R'\Rightarrow d =0$ 
and 
for every $\beta>\alpha$
we have
$\beta d\in R'\Rightarrow d =1$.

Assume that one variable $y$ can be used for two different tuples
$\alpha_{1},\alpha_{2}\in\{0,1\}^{2n}\setminus \sigma_{n}$.
We consider two cases.

Case 1. Assume that there is $i$ such that 
$\alpha_{1}(i) = \alpha_{2}(i)$. 
Without loss of generality we assume that 
$\alpha_{1}(1) = \alpha_{2}(1)$
and $\alpha_{1}(2n) \neq \alpha_{2}(2n)$.
Let us define tuples 
$\beta_{1},\beta_{2},\beta_{3}\in \sigma_{n}$.

Put 
$\beta_{1}(i) = 
\begin{cases}
1, & \text{if $i\in\{1,2\}$}\\
\alpha_{1}(i), & \text{otherwise}
\end{cases},$
$\beta_{2}(i) = 
\begin{cases}
1, & \text{if $i\in\{3,4\}$}\\
\alpha_{2}(i), & \text{otherwise}
\end{cases},$

$\beta_{3}(i) = 
\begin{cases}
\alpha_{1}(i), & \text{if $\alpha_{1}(i) = \alpha_{2}(i)$ or $i\le 4$.}\\
0, & \text{otherwise}
\end{cases}.$

Let us show that 
$h_{0,2}(\beta_{1},\beta_{2},\beta_{3}) = \beta_{1}$.
In fact, for the first two rows, reading down through the $2n$ rows of $h_{0,2}(\beta_{1},\beta_{2},\beta_{3})$, we have  
$h_{0,2}(1,0,0) = h_{0,2}(1,1,1)=1$.
For the next two rows 
we have 
$h_{0,2}(0,1,0) = 0$ and 
$h_{0,2}(1,1,1)=1$.
For the remaining rows 
we either use
$h_{0,2}(0,0,0) = 0$ and 
$h_{0,2}(1,1,1) = 1$, 
or 
$h_{0,2}(0,1,0) = 0$ and 
$h_{0,2}(1,0,0) = 1$.

Since $\beta_{1}>\alpha_{1}$, 
$\beta_{2}>\alpha_{2}$, 
$\beta_{3}<\alpha_{1}$,
by Lemma~\ref{StrangeStructureMainLemma}, 
$y$ should be equal to $1$ 
in any solution of $\Phi$ such that 
$(x_{1},\ldots,x_{2n}) \in\{\beta_{1},\beta_{2}\}$, 
and it should be equal to $0$ 
in any solution of $\Phi$ such that 
$(x_{1},\ldots,x_{2n})=\beta_{3}$.
Since
$h_{0,2}(\beta_{1},\beta_{2},\beta_{3}) = \beta_{1}$, $h_{0,2}(1,1,0) = 2$ and $h_{0,2}$ is a polymorphism of $\Gamma$, we get a contradiction to the uniqueness of $d$ in Lemma~\ref{StrangeStructureMainLemma}.
Case 2.
Assume that $\alpha_{1}(i)\neq \alpha_{2}(i)$ for every $i$.
Put 
\begin{align*}
\beta_{1}(i) &= 
\begin{cases}
1, & \text{if $i\in\{1,2\}$}\\
\alpha_{1}(i), & \text{otherwise}
\end{cases},
\beta_{2}(i) = 
\begin{cases}
1, & \text{if $i\in\{1,2\}$}\\
\alpha_{2}(i), & \text{otherwise}
\end{cases},
\beta_{3}(i) = 
\begin{cases}
1, & \text{if $i\in\{1,2\}$}\\
0, & \text{otherwise}
\end{cases},\\
\beta_{4}(i) &= 
\begin{cases}
1, & \text{if $i\in\{3,4\}$}\\
\alpha_{1}(i), & \text{otherwise}
\end{cases},
\beta_{5}(i) = 
\begin{cases}
\alpha_{1}(i), & \text{if $i\in\{1,2\}$}\\
0, & \text{otherwise}
\end{cases}.
\end{align*}

Since 
$h_{0,2}(1,1,1) =h_{0,2}(1,0,0) = 1$, 
$h_{0,2}(0,1,0) =0$, we have
$h_{0,2}(\beta_{1},\beta_{2},\beta_{3}) =\beta_{1}$.
Since $\beta_{1}>\alpha_{1}$
and
$\beta_{2}>\alpha_{2}$, 
by Lemma~\ref{StrangeStructureMainLemma}, 
$y$ should be equal to $1$ 
in any solution of $\Phi$ such that 
$(x_{1},\ldots,x_{2n}) \in\{\beta_{1},\beta_{2}\}$.
Since $h_{0,2}(1,1,a) = 1$ only if $a=1$,
$y$ should be equal to 1 on $\beta_{3}$
(in any solution of $\Phi$ 
such that $(x_{1},\ldots,x_{2n})=\beta_{3}$).
Since 
$h_{0,2}(\beta_{3},\beta_{4},\beta_{5}) = 
\beta_{3}$, and $y$ should be equal to $1$ on $\beta_{4}$ and equal to $0$ on $\beta_{5}$, 
we would obtain 
$h_{0,2}(1,1,0) =1$. However, in fact $h_{0,2}(1,1,0)=2$, which, since $h_{0,2}$ is a polymorphism of $\Gamma$, contradicts the uniqueness of $d$ in Lemma~\ref{StrangeStructureMainLemma}.

Thus, for every tuple $\alpha\in\{0,1\}^{2n}\setminus \sigma_{n}$ 
there exists a unique variable $y$, which completes the proof.
\end{proof}

Since $\sigma_{n}$ can be obtained from $\tau_n$ by identification of variables, 
we have the following corollary. 

\begin{cor}
Any pp-definition of $\tau_{n}$ over $\Gamma$
has at least $2^{n}$ variables.
\end{cor}

Let us characterize the clone $\Pol(\Gamma)$.
First, let us define relations on the set $A$.
Put
$\epsilon_{2}(x_1,x_2) = (x_1=0)\vee (x_1=1\wedge x_2\neq 2)$,
$\zeta_{2}(x_1,x_2) = (x_1=0)\vee (x_1=1\wedge x_2=1)$,
$\epsilon_{n+1}(x_{1},\ldots,x_{n}, x_{n+1}) = 
(x_1=0) \vee (x_{1}= 1\wedge \epsilon_{n}(x_{2},\ldots,x_{n+1}))$,
$\zeta_{n+1}(x_{1},\ldots,x_{n}, x_{n+1}) = 
(x_1=0) \vee (x_{1}= 1\wedge \zeta_{n}(x_{2},\ldots,x_{n+1}))$.
Put $\Delta=\{\{1\}, \{0,1\}, 
\epsilon_{2},\epsilon_{3},\epsilon_{4},\ldots, 
\zeta_{2},\zeta_{3},\zeta_{4},\ldots\}$.

Suppose $B\subseteq \{1,2,\ldots,n\}$.
Suppose $\rho\subseteq A^{n}$ can be defined as a conjunction 
of relations from $\Delta$
and
 $\proj_{B} \rho\subseteq \{0,1\}^{|B|}$.
Then we define an operation 
$f_{B,\rho}(x_{0},x_{1},\ldots,x_{n})$ as follows.
If
$x_{i}=0$ for every $i\in\{0\}\cup B$ then it returns 0.
If 
$x_{0} = 1$ and $(x_{1},\ldots,x_{n})\in\rho$ then it returns 1.
In all other cases it returns 2.

By $\mathcal C$ we denote the set of all operations
that can be obtained from $f_{B,\rho}$ for all $B$ and $\rho$ by a permutation of variables (it is sufficient to move just the first variable).
We now begin our journey towards proving Theorem~\ref{StrangeCloneDefinition}, which says that
$\Pol(\Gamma) =\Clo(\{g_{0,2},s_{0,2}\})$.
\begin{lem}\label{ConjunctionDelta}
Suppose $\rho\subseteq A^{n}$, 
$(1,1,\ldots,1)\in\rho$, and
in every $(m+1)\times n$-matrix
, whose first $m$ rows are from $\rho$ and 
all columns are from $R_{and,m}$,
the last row is also from $\rho$.
Then $\rho$ can be represented as a conjunction of relations 
from $\Delta$.
\end{lem}
\begin{proof}
We prove by induction on the arity of $\rho$.
If $n =1$, then 
$\rho\in\{\{1\},\{0,1\},\{0,1,2\}\}$.
Thus, either $\rho\in \Delta$, or $\rho$ is full.
Assume that $n\ge 2$.

We want to build a conjunction of relations, which we denote by $\Phi$. 
We start with $\Phi=\varnothing$.
For every tuple $\alpha = (a_{1},\ldots,a_{n})\notin\rho$
we add the corresponding constraints to $\Phi$ to exclude this tuple. If we can exclude every tuple then $\Phi$ defines $\rho$.

Let $I\subseteq\{1,2,\ldots,n\}$ be a minimal subset such that 
$\proj_{I}(\alpha)\notin \proj_{I}(\rho)$.
It is not hard to see that 
$\proj_{I}(\rho)$ satisfies all the assumptions of this lemma. Therefore, if 
$I\neq \{1,2,\dots,n\}$, then 
by the inductive assumption
$\proj_{I}(\rho)$ can be represented as a conjunction 
of relations from $\Delta$.
Then, we may add this representation to $\Phi$ to 
exclude $\alpha$.
Thus, we may assume that 
$I=\{1,2,\dots,n\}$.

Since $I$ is minimal, 
there exist
$b_{1},\ldots,b_{n}$ such that 
$\beta_{j} = (a_{1},\ldots,a_{j-1},b_{j},a_{j+1},\ldots,a_{n})\in \rho$ for every $j$.

There should be at most one $i$ 
such that 
$(a_{i},b_{i},a_{i})\in R_{and,2}$.
Otherwise (if we have $i_{1}$ and $i_{2}$), we build a matrix 
whose rows are $\beta_{i_{1}},\beta_{i_{2}},\alpha$ for different $i_{1},i_{2}$.
It is not hard to see that every column of the matrix is from $R_{and,2}$, which means that 
$\alpha\in\rho$. Contradiction.
Hence, 
there exists at most one $i$ such that $(a_{i},b_{i},a_{i})\in R_{and,2}$.

Notice that $(a,b,a)\in R_{and,2}$ if 
$a\in \{0,2\}$ or $b=2$.
Then there should be exactly one $i$ such that 
$(a_{i},b_{i},a_{i})\in R_{and,2}$ (otherwise $a_{i} =1$ for every $i$, which contradicts the fact that 
$(1,1,\ldots,1)\in\rho$).
Without loss of generality, let this $i$ be equal to $n$, 
then $\alpha = (1,1,\ldots,1, a_{n})$, where $a_{n}\in\{0,2\}$.

Assume that 
$\proj_{i}(\rho)\not\subseteq\{0,1\}$ for every $i\in\{1,2,\dots,n-1\}$,
then we consider tuples
$\mathbf{b}^{1},\dots,\mathbf{b}^{n-1}\in \rho$ such that 
the $i$-th element of $\mathbf{b}^{i}$ is $2$,
and build a matrix whose rows are
$\mathbf{b}^{1},\dots,\mathbf{b}^{n-1}, \beta_{1}, \alpha$.
We can check that every column of this matrix is from $R_{and,n}$, 
therefore 
$\alpha\in\rho$, which contradicts our assumptions.

Thus, 
there exists 
$i\in\{1,2,\dots,n-1\}$ such that 
$\proj_{i}(\rho)\subseteq\{0,1\}$.
Without loss of generality we assume that 
$i=1$.
Put 
$\delta(x_{2},\ldots,x_{n}) = 
\rho(1,x_2,\ldots,x_{n})$.
Since any tuple from $R_{and,m}$ starting with $m$ 1s
should end with 1, 
$\delta$ satisfies all the assumptions of 
this lemma.
By the inductive assumption, 
$\delta$ can be represented as a conjunction of 
relations from $\Delta$.
In this representation we replace 
$\epsilon_{s}(x_{i_1},\dots,x_{i_s})$
by 
$\epsilon_{s+1}(x_{1},x_{i_1},\dots,x_{i_s})$,
$\zeta_{s}(x_{i_1},\dots,x_{i_s})$
by 
$\zeta_{s+1}(x_{1},x_{i_1},\dots,x_{i_s})$,
$x_{j}\in\{0,1\}$
by 
$\epsilon_{2}(x_{1},x_{j})$,
$x_{j}\in\{1\}$
by 
$\zeta_{2}(x_{1},x_{j})$
and add the obtained constraints to $\Phi$.
Since $a_{1} = 1$, 
we excluded the tuple $\alpha$ from the solution set by 
adding these constraints.
Note that we maintain the property that 
every tuple $\beta\in\rho$ satisfies $\Phi$.
Thus, we can exclude every tuple $\alpha\in A^{n}\setminus \rho$, which means that 
$\Phi$ defines $\rho$.
\end{proof}

\begin{lem}\label{StrangeCharLemma}
$\Pol(\Gamma) \subseteq  \mathcal C$.
\end{lem}

\begin{proof}
Suppose $f\in\Pol(\Gamma)$ is an operation of arity 
$n+1$.
Since $\sigma_{m}$ can be pp-defined from $\Gamma$
for every $m$, 
$f$ is $\{0,2\}\{1,2\}$-projective \cite{ZhukGap2015}. 
Without loss of generality we assume that it is an $\{0,2\}\{1,2\}$-projection to the first variable.

Let $B$ be the set of all $i\in\{1,2,\ldots,n\}$ such that 
$f(a_{0},\ldots,a_{n})=0$ implies $a_{i}=0$.
Let $\mathbf a^{1},\ldots,\mathbf a^{s}$
be all the tuples such that $f(\mathbf a) = 0$.
Let $\omega(x_{1},\ldots,x_{s},z)$ be defined by 
$$\exists x_{1}'\dots\exists x_{s}'\; 
R_{and,s+1}(x_{1}',\ldots,x_{s}',z)\wedge \delta(x_{1},x_{1}')\wedge \dots\wedge \delta(x_{s},x_{s}').$$
Let 
$\mathbf b = (0,b_{1},\ldots,b_{n})$ be a tuple such that
$b_{i} = 0$ if $i\in B$.
Let us build a matrix whose rows are tuples 
$\mathbf a^{1},\ldots,\mathbf a^{s},\mathbf b$.
By the definition of $B$, every column of this matrix is from $\omega$.
Since $f$ preserves $\omega$, the result of 
applying $f$ to the matrix (that is $(f(\mathbf a^{1}),\dots,f(\mathbf a^{s}),f(\mathbf b))$) should be from $\omega$.
But the first $s$ elements of the result equal 0, therefore $f(\mathbf b)=0$.
Thus, we proved that if
$x_{i}=0$ for every $i\in\{0\}\cup B$ then 
$f(x_{0},\ldots,x_{n})=0$.

Let $\rho$ be the set of all tuples $(a_{1},\ldots,a_{n})$ such that 
$f(1,a_{1},\ldots,a_{n}) = 1$.
Since $f$ preserves $\{1\}$, 
we have $(1,1,\ldots,1)\in \rho$.

Assume that for some tuple 
$\mathbf a = (a_{1},\ldots,a_{n})\in \rho$ and $i\in B$ we have $a_{i} = 2$.
Let $\mathbf c=(c_{1},\ldots,c_{n})$ be the tuple 
such that 
$c_{j} = 0$ if $j\in B\setminus\{i\}$ and $c_{j} = 2$ otherwise.
Let $\mathbf d=(d_{1},\ldots,d_{n})$ be the tuple 
such that 
$d_{j} = 0$ if $j\in B$ and $d_{j} = 2$ otherwise.
It is not hard to see that 
the rows $\mathbf a, \mathbf d, \mathbf c$ form a matrix whose columns are from $R_{and,2}$.
Therefore
$(f(1 \mathbf a), 
f(0 \mathbf d), 
f(0 \mathbf c))\in R_{and,2}$.
Since
$f(1 \mathbf a) =1$ and $f(0 \mathbf d) =0$,
we have
$f(0 \mathbf c) = 0$, which contradicts the definition of $B$. Thus, $\proj_{B}(\rho)\subseteq \{0,1\}^{|B|}$.

It remains to prove that $\rho$ can be represented as a conjunction of relations 
from $\Delta$.
Let us show that $\rho$ satisfies the assumptions 
of Lemma~\ref{ConjunctionDelta}.
Consider a matrix whose columns are from 
$R_{and,m}$ and whose first $m$ rows are from $\rho$.
Add a column with 1s in the beginning of the matrix
and apply $f$.
Since $f$ preserves $R_{and,m}$, 
the result should be from $R_{and,m}$.
Since the first $m$ elements are equal to 1, the last element 
should be equal to 1. Therefore, the last row of the matrix 
is from $\rho$.
It remains to apply Lemma~\ref{ConjunctionDelta} to complete the proof.
\end{proof}


\begin{lem}\label{StrangeBasis}
$\mathcal C\subseteq \Clo(\{g_{0,2},s_{0,2}\})$.
\end{lem}
\begin{proof}
We need to prove that 
every operation from $\mathcal C$ can be built 
from $g_{0,2}$ and $s_{0,2}$.
For $B\subseteq \{1,2,\ldots,n\}$
by $\psi(B)$ we denote the set of all tuples $\alpha\in A^{n}$
such that $\proj_{B}(\alpha) \in\{0,1\}^{|B|}$.
The following formulas show how to 
generate 
$f_{\varnothing,\epsilon_{n}}$, 
$f_{\varnothing,\zeta_{n}}$,
and 
$f_{B,\psi(B)}$
for all $n\ge 2$ and $B\neq\varnothing$:
\begin{align*}
f_{\varnothing, \{0,1\}}(x_0,x_1) =&g_{0,2}(x_0,x_1),\\
f_{\varnothing, \{1\}}(x_0,x_1) =& g_{0,2}(x_0,s_2(x_0,x_1)), \\
f_{\varnothing,\epsilon_{2}}(x_{0},x_{1},x_{2}) =& g_{0,2}(x_{0},g_{0,2}(x_{1},x_{2})),\\
f_{\varnothing,\epsilon_{n+1}}(x_{0},\ldots,x_{n+1}) =& g_{0,2}(x_{0},f_{\varnothing,\epsilon_{n}}(x_{1},\ldots,x_{n+1})),\\
f_{\varnothing,\zeta_{2}}(x_{0},x_{1},x_{2}) =& g_{0,2}(x_{0},g_{0,2}(x_{1},s_{0,2}(x_{2},x_{1}))),\\
f_{\varnothing,\zeta_{n+1}}(x_{0},\ldots,x_{n+1}) =& g_{0,2}(x_{0},f_{\varnothing,\zeta_{n}}(x_{1},\ldots,x_{n+1})),\\
f_{\varnothing,\psi(\varnothing)}(x_{0},x_{1},\ldots,x_{n})=&x_0,\\
f_{B,\psi(B)}(x_{0},x_{1},\ldots,x_{n})=&\bigvee_{i\in B} s_{0,2}(x_0,x_{i}),
\end{align*}
where by $\bigvee$ we mean the semilattice operation $s_2$.

It remains to show how to combine such operations.
For two relations 
$\rho_{1},\rho_{2}\subseteq A^{n}$, 
the following equality holds
$$f_{\varnothing,\rho_{1}\cap\rho_{2}}(x_{0},x_{1},\ldots,x_{n}) = s_2(f_{\varnothing,\rho_{1}}(x_{0},x_1,\ldots,x_{n}),f_{\varnothing,\rho_{2}}(x_{0},x_{1},\ldots,x_{n})).$$
Note that adding dummy variables to a relation $\rho$ is equivalent to adding dummy variables to the operation $f_{\varnothing,\rho}$.

To finish the proof it is sufficient to show that $f_{B,\rho}\in\Clo(\{g_{0,2},s_{0,2}\})$
for any $B\subseteq \{1,2,\ldots,n\}$
and $\rho\subseteq A^{n}$ such that 
$\proj_{B}(\rho)\subseteq\{0,1\}^{|B|}$
and $\rho$ is a conjunctions of relations 
from $\Delta$.
Since for every $\delta\in\Delta$ we showed how to generate 
$f_{\varnothing,\delta}$ and how to define conjunction (intersection), 
$f_{\varnothing,\rho}$ can be generated from $g_{0,2}$ and $s_{0,2}$.
Then 
$$f_{B,\rho}(x_{0},x_{1},\ldots,x_{n})) = s_2(f_{\varnothing,\rho}(x_0,x_{1},\ldots,x_{n}), 
f_{B,\psi(B)}(x_{0},x_{1},\ldots,x_{n})).$$%
\end{proof}

\begin{thm}\label{StrangeCloneDefinition}
$\Pol(\Gamma) = \mathcal C=\Clo(\{g_{0,2},s_{0,2}\})$.
\end{thm}

\begin{proof}
The claim follows from the following inclusions
$$\Pol(\Gamma) \subseteq \mathcal C\subseteq \Clo(\{g_{0,2},s_{0,2}\})\subseteq \Pol(\Gamma).$$
The first inclusion is by Lemma~\ref{StrangeCharLemma}, 
the next inclusion is by Lemma~\ref{StrangeBasis},
and the last one follows from the fact that 
$g_{0,2}$ and $s_{0,2}$ preserve $R_{and,2}$ and $\delta$.
\end{proof}

\newcommand{\Solve}{\mbox{\textsc{Solve}}}
\newcommand{\SolveCSP}{\mbox{\textsc{SolveCSP}}}
\newcommand{\Break}{\State \textbf{break} }
\newcommand{\Output}{\mbox{Output}}

Below we present an algorithm 
that solves
$\QCSP^{2}(\Gamma)$ in polynomial time (see the pseudocode).
For an input 
$\forall x_1 \dots\forall x_{n} \exists y_1\dots\exists y_s \Phi$, 
the function $\Solve_1$
first checks whether $\Phi$ holds 
on $\mathbf x= (0,\dots,0)$, 
$\mathbf x=(1,\dots,1)$,
and on each
tuple containing exactly one 1. 
Then for every variable $y_{j}$ 
and every variable $x_{i}$ 
it calculates the set $D_{i,j}$ of 
possible values for $y_{j}$ when 
$\mathbf x = 1^{i-1}01^{n-i}$.
Finally, it checks another $s$ tuples determined by 
$D_{i,j}$.
Note that it would be an exponential algorithm if it just checks all possible $\mathbf x$.
Moreover, since the relation $\rho$ omitting exactly one tuple is pp-definable over $\Gamma$, the tuples we need to check could not be 
found without looking into the formula $\Phi$. 
By $h$ we denote the operation defined on subsets of $A$ by 
$h(B) = \begin{cases}
0, & \text{if $B = \{1\}$}\\
1, & \text{otherwise}
\end{cases}$.
By $\SolveCSP$ we denote a polynomial algorithm, solving 
constraint satisfaction problem for a constraint language preserved by the semilattice operation $s_2$: it returns true if it has a solution, it returns false otherwise.

\begin{algorithm}
\begin{algorithmic}[1]
\Function{Solve$_1$}{$\Theta$}
  \State{\textbf{Input:} $\QCSP^{2}(\Gamma)$ instance $\Theta = \forall x_1 \dots\forall x_{n} \exists y_1\dots\exists y_s \Phi.$}
    \If{$\neg\SolveCSP(\mathbf x = (0,\ldots,0)\wedge \Phi)$}
        \Return{false} \Comment{$\mathbf x = (x_{1},\dots,x_{n})$}
    \EndIf
    \If{$\neg\SolveCSP(\mathbf x = (1,\ldots,1)\wedge \Phi)$}
        \Return{false} 
    \EndIf
    \For{$i:=1,\ldots,n$}
        \If{$\neg\SolveCSP(\mathbf x = (\underbrace{0,\ldots,0}_{i-1},1,0,\dots,0)\wedge \Phi)$}
        \Return{false}         
        \EndIf
    \EndFor
    \For{$j:=1,\ldots,s$}
        \For{$i:=1,\ldots,n$}
            \State{$D_{i,j}:=\varnothing$}                    
            \State{$\mathbf{c}:=(\underbrace{1,\ldots,1}_{i-1},0,1,\dots,1)$}
            \For{$a\in A$}
                \If{$\SolveCSP(\mathbf x = \mathbf c\wedge y_{j} = a\wedge\Phi)$}
                    \State{$D_{i,j}:=D_{i,j}\cup\{a\}$}                    
                \EndIf    
            \EndFor
            \If{$D_{i,j}=\varnothing$}
                \Return{false}        
            \EndIf
        \EndFor
        \If{$\neg\SolveCSP(\mathbf x = (h(D_{1,j}),\ldots,h(D_{n,j}))\wedge \Phi)$}
            \Return{false}
        \EndIf
    \EndFor
    \Return{true}
\EndFunction
\end{algorithmic}
\end{algorithm}

\begin{lem}\label{ComplexityStrange1}
Function $\Solve_1$ solves $\QCSP^{2}(\Gamma)$ in polynomial time.
\end{lem}

\begin{proof}
First, let us show that the algorithm actually solves the problem.
If the answer is false, then we found an evaluation 
of $(x_1,\ldots,x_{n})$ such that the corresponding CSP has no solutions, which means that the answer is correct.

Assume that the answer is true.
Let $R(x_{1},\ldots,x_{n})$ be defined by the formula 
$\exists y_1\dots\exists y_s \Phi$. 
We need to prove that $R$ is a full relation.
Assume the converse. 
Using the semilattice operation $s_2$ we can generate
$A^{n}$ from $\{0,1\}^{n}$, hence 
$\{0,1\}^{n}\not\subseteq R$.
Then let $\alpha$ be a minimal tuple from $\{0,1\}^{n}\setminus R$.
Without loss of generality we assume that 
$\alpha = 1^{k}0^{n-k}$. 
For every $i$ we put $\alpha_{i} = 1^{i-1}0 1^{n-i}$.
Since $\alpha$ is minimal, 
all the tuples smaller than $\alpha$ should be in $R$.
We checked that 
$R$ contains $(0,0,\dots,0)$ and all tuples with 
exactly one 1, 
hence $\alpha$ contains at least two 1s.
Then by Lemma~\ref{StrangeStructureMainLemma}
there should be a variable $y$ 
such that for any $\beta<\alpha$ we have
$\beta d \in R' \Rightarrow d=0$, 
for any $\beta>\alpha$ we have
$\beta d \in R' \Rightarrow d=1$,
where
$R' = \Phi(x_{1},\ldots,x_{n},y)$.
Since $D_{i,j}\neq \varnothing$ for all $i,j$, $\alpha_{i}\in R$ for every $i$,
and for every $i>k$ we have $\alpha_{i}d\in R'\Rightarrow d=1$.
Assume that $y=x_{i}$ for some $i$. If $\alpha(i) = 0$ then 
$\alpha_{i}>\alpha$ and $\alpha_{i}d\in R'\Rightarrow d=1$
contradicts $y=x_{i}$.
If $\alpha(i) = 1$, then there is a tuple $\alpha'<\alpha$ such that
$\alpha'(i) = 1$, which contradicts
$\alpha'd\in R'\Rightarrow d=0$.
Hence $y\neq x_{i}$ and $y = y_{j}$ for some $j$.

Let $\beta = 01^{k-1}0^{n-k}$.
Since $\beta<\alpha$, we have $\beta0\in R'$.
Assume that
$D_{1,j}$ is equal to $\{1\}$,
then 
$\alpha_{1}d\in R'\Rightarrow d=1$.
Put $\gamma_{0} = s_{0,2}(\beta,\alpha_{1})=01^{k-1}2^{n-k}$.
Since $s_{0,2}$ preserves $R'$ and $s_{0,2}(0,1) = 2$,
we have $\gamma_{0}2\in R'$.
Put $\gamma_{i} = g_{0,2}(\alpha_{k+i},\gamma_{i-1})$
for $i=1,2,\dots,n-k$.
Since $g_{0,2}(1,2)=2$, we have $\gamma_{i}2\in R'$
for every $i$.
Note that $\gamma_{n-k}\in\{0,1\}^{n}$.
We can check that
$g_{0,2}(\alpha_{1}1,\gamma_{n-k}2)
= \alpha_{1}2$, which contradicts the fact that $D_{1,j}=\{1\}$.
In this way we can show that 
$D_{1,j},\dots,D_{k,j}$ are not equal to $\{1\}$.
We also know that 
$D_{k+1,j},\dots,D_{n,j}$ are equal to $\{1\}$.
Hence, the tuple 
$(h(D_{1,j}),\dots,h(D_{n,j})) = \alpha$ 
was checked in the algorithm, 
which contradicts the fact that $\alpha\notin R$.

It remains to show that the algorithm works in polynomial time.
It follows from the fact that in the algorithm we just solve 
$3\cdot s\cdot n + s +n +2$ CSP instances over a language preserved by the semilattice operation $s_2$.
\end{proof}

\begin{cor}\label{CorComplexityStrange1Pr}
$\QCSP(\Gamma)$  is in P.
\end{cor}
\begin{proof}
Since $f(x,y,z) = s_{0,2}(x,y)$ is a $01$-stable operation, by Lemma~\ref{ReductionToQCSP2} 
$\QCSP(\Gamma)$ can be polynomially reduced to $\QCSP^{2}(\Gamma)$, 
and 
$\QCSP^{2}(\Gamma)$ can be solved by the function $\Solve_{1}$.
\end{proof}

\begin{cor}\label{CorComplexityStrange1}
$\QCSP(\Gamma_0)$ is in P for every finite $\Gamma_0\subseteq \Inv(\{g_{0,2},s_{0,2}\})$.
\end{cor}
\begin{proof}
By Theorem~\ref{StrangeCloneDefinition}, 
$\Pol(\Gamma) = \Clo(\{g_{0,2},s_{0,2}\})$.
Then $\Pol(\Gamma_0)\supseteq \Pol(\Gamma)$, 
which implies that 
each relation from $\Gamma_0$ can be pp-defined from $\Gamma$.
Therefore, $\QCSP(\Gamma_0)$ can be
polynomially reduced to $\QCSP(\Gamma)$ 
(we just replace every relation by its pp-definition).
Hence, 
by Corollary~\ref{CorComplexityStrange1Pr},
$\QCSP(\Gamma_0)$ is in P.
\end{proof}

\section{Strange structure 2}\label{StrangeTwoSection}
In this section we 
show that 
$\QCSP(\Gamma)$ is solvable in polynomial time 
if $\Gamma\subseteq \Inv(f_{0,2})$.
Note that 
$s_{0,2}(x,y) = f_{0,2}(x,y,y)$
and $s_2(x,y) = s_{0,2}(x,s_{0,2}(y,x))$.

Suppose $R = \Phi(x_{1},\ldots,x_{n})$,
where $\Phi$ is a conjunctive formula over a constraint language $\Gamma\subseteq\Inv(f_{0,2})$. 
For a variable $y$ of $\Phi$ we define a partial operation $F_{y}(x_{1},\ldots,x_{n})$ on $A$
as follows.
If $\alpha\in R$ and every solution of $\Phi$ with 
$(x_{1},\ldots,x_{n}) = \alpha$ has $y = c$,
where $c\in\{0,1\}$,
then $F_{y}(\alpha) = c$.
Otherwise we say that $F_{y}(\alpha)$ is not defined.
We say 
that $\alpha\in R\cap \{0,1\}^{n}$ is \emph{a minimal 1-set for a variable $y$} 
if $F_{y}(\alpha) = 1$, and $F_{y}(\beta) = 0$ for every $\beta<\alpha$ (every time we use $<$ or $\le$ we mean that both tuples are on $\{0,1\}$).

The following lemma proves that $F_{y}$ is monotonic.
\begin{lem}
Suppose $\alpha\le \beta$, $F_{y}(\alpha)$ and $F_{y}(\beta)$ are defined.
Then 
$F_{y}(\alpha)\le F_{y}(\beta)$.
\end{lem}
\begin{proof}
Assume the contrary, 
then $F_{y}(\alpha) = 1$ and $F_{y}(\beta) = 0$.
We have
$s_{0,2}(\beta 0,\alpha 1) = \beta 2$, which means that 
there exists a solution of $\Phi$ with 
$(x_{1},\ldots,x_{n}) = \beta$ and $y = 2$,
hence $F_{y}(\beta)$ is not defined. Contradiction.
\end{proof}

\begin{lem}
There is at most one minimal 1-set for every variable $y$.
\end{lem}
\begin{proof}
Assume the contrary. Let $\alpha_{1}$ and $\alpha_{2}$ be two minimal 1-sets for $y$.
It follows from the definition that $\alpha_{1}$ and $\alpha_2$ should be incomparable. 
Let $\alpha = \alpha_1\wedge \alpha_2$
(by $\wedge$ we denote the conjunction on $\{0,1\}$).
Then $f_{0,2}(\alpha_{1}1,\alpha 0,\alpha_2 1) = \alpha_1 2$, which contradicts the fact that 
$F_{y}$ is defined on $\alpha_{1}$.
\end{proof}

\begin{lem}\label{MinimalSetExistence}
Suppose $\alpha\in\{0,1\}^{n}\setminus R$, $\alpha$ contains at least two $1$s, and 
$\beta\in R$ for every $\beta<\alpha$.
Then there exists 
a constraint $\rho(z_{1},\ldots,z_{l})$ in $\Phi$  
and $B\subseteq\{1,\ldots,l\}$ 
such that $\alpha = \bigvee_{i\in B} \alpha_{i}$, 
where $\alpha_{i}$ is the minimal 1-set for the variable $z_{i}$
(by $\vee$ we denote the disjunction on $\{0,1\}$).
\end{lem}
\begin{proof}
First, to every variable $y$ of $\Phi$ we 
assign a value $v(y)$ in the following way.
If $F_{y}(\beta) = 0$ for every $\beta<\alpha$ then we put $v(y):=0$.
Otherwise, if $F_{y}(\beta) \in \{0,1\}$ for every $\beta<\alpha$ then we put $v(y):=1$.
Otherwise, put $v(y):=2$.

If $\alpha(i) = 0$ then $F_{x_{i}}(\beta) = 0$ for every 
$\beta<\alpha$, which means that $v(x_{i}) = 0$.
If $\alpha(i) = 1$ then $F_{x_{i}}(\beta) \in\{0,1\}$ for every 
$\beta<\alpha$.
Since $\alpha$ has at least two 1, 
for some $\beta<\alpha$ we have 
$F_{x_{i}}(\beta) =1$,  
which means that $v(x_{i}) = 1$.
Thus we assigned the tuple $\alpha$ to 
$(x_{1},\dots,x_{n})$. 

Since $\alpha\notin R$ the evaluation $v$ cannot be a solution of 
$\Phi$, therefore it breaks at least one constraint from 
$\Phi$. Let us add to $\Phi$ all projections of 
all constraints we have in $\Phi$. 
Thus, for every constraint 
$C = \rho(z_{1},\ldots,z_{l})$ we add 
$\proj_{S}C$, where $S\subseteq \{z_{1},\ldots,z_{l}\}$.
Obviously, when we do this, we do not change 
the solution set of $\Phi$ and stay in $\Inv(f_{0,2})$.

Choose a constraint of the minimal arity $\rho(z_{1},\ldots,z_{l})$ that does not 
hold in the evaluation $v$, 
that is, 
$(v(z_{1}),\dots,v(z_{l}))\notin\rho$.
Let 
$(a_{1},\ldots,a_{l})=(v(z_{1}),\dots,v(z_{l}))$.
Since $\rho$ is a constraint of the minimal arity, 
the evaluation $v$ holds for every proper projection
of $\rho(z_{1},\ldots,z_{l})$, 
which means that 
for every $i$ there exists 
$b_{i}$ such that 
$(a_{1},\dots,a_{i-1},b_{i},a_{i+1},\dots,a_{l})\in\rho$.

Assume that $(a_1,\dots,a_l)$ has two 2, that is 
$a_{i} = a_{j} = 2$ for $i\neq j$.
Then the semilattice operation $s_2$ 
applied to 
$(a_{1},\ldots,a_{i-1},b_{i},a_{i+1},\dots,a_{l})$
and 
$(a_{1},\ldots,a_{j-1},b_{j},a_{j+1},\dots,a_{l})$
gives $(a_{1},\ldots,a_{l})$, which contradicts the fact that $s_2$ preserves $\rho$.

Assume that $a_{i}=2$ for some $i$. 
W.l.o.g. we assume that $a_{l}=2$.
By the definition, there should be a tuple $\beta<\alpha$ such that $F_{z_{l}}(\beta)$ is not defined.
Put $c_{i} = F_{z_{i}}(\beta)$ for every $i<l$, and $c_{l}=2$.
By the definition of $F_{z_{l}}(\beta)$, there should be a solution of $\Phi$ with 
$(x_{1},\ldots,x_{n}) = \beta$ and $z_{l} =2$, or 
two solutions of $\Phi$ with 
$(x_{1},\ldots,x_{n}) = \beta$ and $z_{l} =0,1$. 
Since $s_2$ preserves $\Gamma$, 
in both cases we have a solution of $\Phi$ with 
$(x_{1},\ldots,x_{n}) = \beta$ and $z_{l} =2$.
Note that 
$(z_{1},\ldots,z_{l}) = (c_{1},\ldots,c_{l})$ in this solution,
therefore $(c_{1},\ldots,c_{l})\in\rho$.
By the definition, $c_{i}\le a_{i}$ for every $i<l$.
We apply $s_{0,2}$ to the tuples 
$(a_{1},\ldots,a_{l-1},b_{l})$ and $(c_{1},\ldots,c_{l})$ to obtain the tuple $(a_{1},\ldots,a_{l})$, which is not from $\rho$. This contradicts the fact that $s_{0,2}$ preserves $\rho$.

Assume that $a_{i}\neq 2$ for every $i$. 
W.l.o.g. we assume that 
$a_{1} = \dots = a_{k} =1$ and 
$a_{k+1} = \dots = a_{l} =0$.
If $k=0$ and $(a_{1},\ldots,a_{l}) = (0,\dots,0)$,
then we consider a solution of $\Phi$ corresponding to 
$(x_{1},\ldots,x_{n}) = (0,\dots,0)$.
By the definition of $F_{z_{i}}$ we have 
$(z_{1},\ldots,z_{l})=(0,\dots,0)$ in this solution. 
Hence, $(0,\dots,0)\in\rho$, which contradicts our assumption.
Assume that $k\ge 1$.
For each $i\in\{1,2,\dots,k\}$
we define a tuple $\alpha_{i}$ as follows.
Since 
$F_{z_{i}}$ is defined on any tuple $\beta<\alpha$
and 
$F_{z_{i}}(\beta) = 1$ for some $\beta<\alpha$, 
there exists a minimal 1-set $\alpha_{i}\le\beta$ for $z_{i}$.
Assume that $\alpha' := \alpha_{1}\vee\dots\vee\alpha_{k}<\alpha$.
Consider a solution of $\Phi$ with 
$(x_{1},\ldots,x_{n}) = \alpha'$.
Since $F_{z_{i}}(\alpha')$ is defined,
$F_{z_{i}}(\alpha_{i})=1$ and 
$F_{z_{i}}$ is monotonic, 
we have 
$F_{z_{i}}(\alpha') = 1$ for every $i\in\{1,2,\dots,k\}$.
Therefore, $(z_{1},\ldots,z_{l}) = (a_{1},\ldots,a_{l})$
in this solution, which means that $(a_{1},\ldots,a_{l})\in\rho$ and contradicts the assumption.

Thus, $\alpha'\not<\alpha$. Since $\alpha_{i}\le\alpha$ for every $i$, we obtain 
$\alpha'\le\alpha$, and therefore 
$\alpha' = \alpha$, which completes the proof.
\end{proof}

\begin{lem}\label{OneCalculate}
Suppose $\alpha$ is a minimal 1-set for $y$,
$i\in\{1,2,\dots,n\}$, and 
$2^{i-1}02^{n-i}\in R$. 
Then 
$\alpha(i) = 1$ if and only if  $F_{y}(2^{i-1}02^{n-i})=0$.
\end{lem}

\begin{proof}
Assume that $\alpha(i) = 0$ and  $F_{y}(2^{i-1}02^{n-i})=0$.
We have $s_2(2^{i-1}02^{n-i}0,\alpha 1) = 2^{i-1}02^{n-i}2$,
which means that 
$\Phi$ has a solution 
with $(x_{1},\dots,x_{n}) = 2^{i-1}02^{n-i}$ and $y=2$.
This contradicts the fact that $F_{y}(2^{i-1}02^{n-i})=0$.

Assume that $\alpha(i) =1$ and 
$F_{y}(2^{i-1}02^{n-i})$ is not defined or equal to 1.
Then 
$\Phi$ has a solution 
with $(x_{1},\dots,x_{n}) = 2^{i-1}02^{n-i}$ and $y=c$, 
where $c\neq 0$.
Let $\beta<\alpha$ be the tuple that differs from $\alpha$ only in the $i$-th coordinate.
Since $f_{0,2}(\alpha, \beta,2^{i-1}02^{n-i}) = \alpha$
and $f_{0,2}(1,0,c)=2$, 
$\Phi$ should have a solution 
with $(x_{1},\dots,x_{n}) = \alpha$
and $y = 2$, which contradicts the definition of a minimal 1-set.
\end{proof}

\textbf{Example.}
Let 
$R_{and,2}' = 
R_{and,2}\setminus 
\{(0,2,1),(0,2,2),(2,0,1),(2,0,2)\}$
(see Section~\ref{StrangeOneSection}),
$\delta' = \{(0,1)\}\cup(\{1,2\}\times\{0,1,2\})$,
$\Gamma' = \{R_{and,2}',\delta',\{0\},\{1\},\{2\}\}$.
We can check that $\Gamma'$ is preserved by $f_{0,2}$.
It follows from the following 
lemma that $\Pol(\Gamma')$ has the EGP property.

As in the previous section here we use the notations 
$\tau_{n}$ and $\sigma_{n}$ for $\alpha =\{0,2\}$ and $\beta = \{1,2\}$.

\begin{lem}
$\sigma_{n}$ can be pp-defined from $\{R_{and,2}',\delta',\{0\}\}$.
\end{lem}
\begin{proof}

Recursively we define 
\begin{align*}
R_{and,n+1}'(x_{1},\ldots,x_{n},x_{n+1},y) &= 
\exists z \; R_{and,n}'(x_{1},\ldots,x_{n},z)
\wedge R_{and,2}'(x_{n+1},z,y),\\
\omega_{n}(x_{1},y_{1},x_{2},y_{2},\ldots,x_{n},y_{n})
&=\exists u_1\dots\exists u_{n}\exists z\;
R_{and,2n}'(x_{1},\ldots,x_{n},u_1,\ldots,u_{n},z)
\wedge \\
&\;\;\;\;\;\;\;\;\;\;\;\;\;\;\;\;\;\;\;\;\;\;\;\;\;\;\;\;\;\;\;\;\;\;\;\;\;\;\;\;\;\;\;\;\;\;\delta'(y_{1},u_{1})\wedge \dots \wedge \delta'(y_{n},u_{n})
\wedge z=0.
\end{align*}
It is not hard to see that $\omega_{n}$ contains all the tuples 
but $(1,0,1,0,\ldots,1,0)$. Then the relation 
$\sigma_{n}$ can be represented as a conjunction of $2^{n}$ relations such that each of them 
is obtained from $\omega_{n}$ by a permutation of variables.
\end{proof}

The following lemma and corollary do not play a role in our main result, but we present them for further curiosity and another sanity check.
\begin{lem}
Suppose $\Gamma\subseteq\Inv(f_{0,2})$, 
all relations in $\Gamma$ are of arity at most $k$.
Then any pp-definition of $\sigma_{n}$ over $\Gamma$, where $n\ge 2$,  has at least $2^n/2^k$ constraints
and at least $2^n/2^k$ variables.
\end{lem}

\begin{proof}
Suppose $\sigma_n = \Phi(x_{1},\ldots,x_{n})$,
where $\Phi$ is a conjunctive formula over $\Gamma$.
There exist $2^n$ tuples 
from $A^{2n}\setminus \sigma_{n}$ and each of them has at least two 1s.
By Lemma~\ref{MinimalSetExistence},
for each $\alpha\in A^{2n}\setminus \sigma_{n}$ there should be a constraint 
$\rho(z_{1},\dots,z_{l})$ such that 
$\alpha = \bigvee_{i\in B} \alpha_{i}$ for some $B\subseteq \{1,2,\dots,l\}$.
Since every constraint of $\Phi$ is of arity at most k, 
there are  at most $2^{k}$ options to choose $B$.
Therefore, one constraint of $\Phi$ can cover at most $2^k$ tuples from $A^{2n}\setminus \sigma_{n}$.
Thus, $\Phi$ has at least $2^n/2^k$ constraints.

Similarly, if $V$ is the
set of all variables in $\Phi$, then 
the above argument gives an injection from the set
$A^{2n}\setminus \sigma_{n}$ of size $2^{n}$ to 
the set of all subsets of $V$ of size at most $k$.
This implies that $|V|^k\ge 2^n$ and $|V|\ge 2^{n-k}$.
\end{proof}

Thus, for a fixed $\Gamma$ we need exponentially many 
constraints to define $\sigma_n$.

\begin{cor}
Suppose $\Gamma\subseteq\Inv(f_{0,2})$, 
all relations in $\Gamma$ are of arity at most $k$.
Then any pp-definition of $\tau_{n}$ over $\Gamma$, where $n\ge 2$,  has at least $2^n/2^k$ constraints
and at least $2^n/2^k$ variables.
\end{cor}

Below we present an algorithm that solves
$\QCSP^{2}(\Gamma)$ in polynomial time  for $\Gamma\subseteq \Inv(f_{0,2})$ (see the pseudocode of the function $\Solve_2$).
For an input 
$\forall x_1 \dots\forall x_{n} \exists y_1\dots\exists y_s \Phi$, 
the function $\Solve_2$
first checks whether $\Phi$ holds 
on $\mathbf x= (0,\dots,0)$ 
and on each
tuple containing exactly one 1. 
Then for every variable $y_{j}$ 
it calculates the minimal 1-set using
Lemma \ref{OneCalculate}. 
Precisely, for every variable $y_{j}$ 
it calculates the set $D_{i,j}$ of 
possible values for $y_{j}$
when $\mathbf x = 2^{i-1}02^{n-i}$, 
and sets the i-th element of the 
1-set $\alpha_{j}$ to 1 whenever $D_{i,j}=\{0\}$.
Note that the minimal 1-set for a variable $x_{i}$ 
is $0^{i-1}10^{n-i}$.
Then, for each constraint 
$\rho(z_1,\dots,z_{l})$ and each subset $V$ of its variables the algorithm calculates the disjunction of 
the minimal 1-sets for the variables from $V$, and checks that $\Phi$ holds on the result.
Again, it would be an exponential algorithm if it just checks all possible $\mathbf x$. 
By $\SolveCSP$ we denote a polynomial algorithm, solving 
constraint satisfaction problem for a constraint language preserved by a semilattice operation: it returns true if it has a solution, it returns false otherwise.

\begin{algorithm}
\begin{algorithmic}[1]
\Function{Solve$_2$}{$\Theta$}
  \State{\textbf{Input:} $\QCSP^{2}(\Gamma)$ instance $\Theta = \forall x_1 \dots\forall x_{n} \exists y_1\dots\exists y_s \Phi.$}
    \If{$\neg\SolveCSP(\mathbf x = (0,\ldots,0)\wedge \Phi)$}
        \Return{false} \Comment{$\mathbf x = (x_{1},\dots,x_{n})$}
    \EndIf
    \For{$i:=1,\ldots,n$} \Comment{Check all tuples with just one 1}
        \State{$\mathbf{c}:=(\underbrace{0,\ldots,0}_{i-1},1,0,\dots,0)$}
        \If{$\neg\SolveCSP(\mathbf x = \mathbf c\wedge\Phi)$}
            \Return{false}        
        \EndIf    
    \EndFor

    \For{$j:=1,\ldots,s$} \Comment{Calculate the minimal 1-set for every $y_{j}$}
        \State{$\alpha_{j} := (0,\ldots,0)$}
        \For{$i:=1,\ldots,n$}
            \State{$D_{i,j}:=\varnothing$}                    
            \State{$\mathbf{c}:=(\underbrace{2,\ldots,2}_{i-1},0,2,\dots,2)$}
            \For{$a\in A$}
                \If{$\SolveCSP(\mathbf x = \mathbf c\wedge y_{j} = a\wedge\Phi)$}
                    \State{$D_{i,j}:=D_{i,j}\cup\{a\}$}                    
                \EndIf    
            \EndFor
            \If{$D_{i,j}=\varnothing$}
                \Return{false}        
            \EndIf
            \If{$D_{i,j}=\{0\}$}
                \State{$\alpha_{j}:= \alpha_{j}\vee (\underbrace{0,\ldots,0}_{i-1},1,0,\dots,0)$}
            \EndIf            
        \EndFor
    \EndFor
    \For{a constraint $\rho(z_{1},\dots,z_{l})$ of $\Phi$}\Comment{Check all constraints}
        \For{$V\subseteq \{1,2,\dots,l\}$} \Comment{Check all subsets of variables}
            \State{$\beta:=(0,\dots,0)$}
            \For{$j\in V$}
                \If{$z_{j} = x_{i}$ for some $i$}   
                    \State{$\beta:=\beta\vee (\underbrace{0,\ldots,0}_{i-1},1,0,\dots,0)$} \Comment{Add the minimal 1-set for $x_{i}$}
                \EndIf
                \If{$z_{j} = y_{i}$ for some $i$} 
                    \State{$\beta:=\beta\vee \alpha_{i}$} 
                    \Comment{Add the minimal 1-set for $y_{i}$}
                \EndIf                
            \EndFor
            \If{$\neg\SolveCSP(\mathbf x = \beta\wedge \Phi)$}
                \Return{false}
            \EndIf
        \EndFor
    \EndFor    
    \Return{true}
\EndFunction
\end{algorithmic}
\end{algorithm}

\begin{lem}\label{ComplexityStrange2}
Function $\Solve_2$ solves $\QCSP^{2}(\Gamma)$ in polynomial time
for a finite constraint language $\Gamma\subseteq \Inv(f_{0,2})$.
\end{lem}
\begin{proof}
First, let us show that the algorithm actually solves the problem.
If the answer is false, then we found an evaluation 
of $(x_1,\ldots,x_{n})$ such that the corresponding CSP has no solutions, which means that the answer is correct.

Assume that the answer is true.
Let $R(x_{1},\ldots,x_{n})$ be defined by the formula 
$\exists y_1\dots\exists y_s \Phi$. 
We need to prove that $R$ is a full relation.
Assume the converse. 
Using the semilattice operation $s_2$ we can generate
$A^{n}$ from $\{0,1\}^{n}$, hence 
$\{0,1\}^{n}\not\subseteq R$.
Then let $\alpha$ be a minimal tuple from $\{0,1\}^{n}\setminus R$.
Since we checked that $(0,0,\dots,0)$ and all tuples having just one 1
are from $R$, $\alpha$ contains at least two 1.
Then, by Lemma~\ref{MinimalSetExistence}, there should be a constraint 
$\rho(z_{1},\dots,z_{l})$ of $\Phi$ and 
a subset $V\subseteq\{1,2,\dots,l\}$ such that 
$\alpha$
is a disjunction of the minimal 1-sets of 
$z_{i}$ for $i\in V$. Thus, it is sufficient to find the minimal 1-set
corresponding to each variable and check all the disjunctions.

By Lemma~\ref{OneCalculate}, if $\alpha_{j}$ is minimal 1-set for a variable 
$y_{j}$ then it was correctly found in lines 7-17 of the algorithm. 
Note that if $y_{j}$ does not have a minimal 1-set then 
we do not care what we found.
Then, in lines 18-25 we check all constraints of $\Phi$, 
check all subsets of variables $V$, and calculate the corresponding disjunction.
Here we use the fact that the minimal 1-set for $x_{i}$ is 
$0^{i-1}10^{n-i}$.
In line 26 we check whether $\Phi$ has a solution with
$(x_{1},\dots,x_{n})=\alpha$.
Thus, Lemma~\ref{MinimalSetExistence} guarantees that $\{0,1\}^{n}\subseteq R$,
and therefore $A^{n}\subseteq R$.

It remains to show that the algorithm works in polynomial time.
In the algorithm we just solve at most 
$1+n+s\cdot n\cdot 3+m\cdot 2^{r}$ CSP instances over a language preserved by the semilattice operation $s_2$,
where $m$ is the number of constraints in $\Phi$ and 
$r$ is the maximal arity of constraints in $\Phi$.
Since $\Gamma$ is finite, $r$ is a constant, hence the algorithm is polynomial.
\end{proof}

\begin{cor}\label{CorComplexityStrange2}
$\QCSP(\Gamma)$  is in P for every finite $\Gamma\subseteq \Inv(f_{0,2})$.
\end{cor}
\begin{proof}
Since $f(x,y,z) = s_{0,2}(x,y)$ is a $01$-stable operation, by Lemma~\ref{ReductionToQCSP2} 
$\QCSP(\Gamma)$ can be polynomially reduced to $\QCSP^{2}(\Gamma)$, 
and 
$\QCSP^{2}(\Gamma)$ can be solved by the function $\Solve_2$.
\end{proof}

\section{EGP and WNU on 3-element domain}

In this section we consider constraint languages on $\{0,1,2\}$
with constants
for which
$\Pol(\Gamma)$ has the EGP property and contains a WNU operation.
It can be shown (see the proof of Theorem~\ref{ThreeElementClassification})
that in this case $\Pol(\Gamma)$ contains a semilattice operation.
Since adding pp-definable relations to $\Gamma$ does not change the complexity of 
$\QCSP(\Gamma)$, we may assume that $k$ is the maximal arity of the relations in $\Gamma$ 
and $\Gamma$ contains all pp-definable relations of arity at most $k$.
Also, EGP implies that there exists $\alpha$ and $\beta$, neither equal to $A$ but so that $\alpha \cup \beta = A$, so that all operations of $\Pol(\Gamma)$ are $\alpha\beta$-projective \cite{ZhukGap2015}. If $\alpha \cap \beta = \emptyset$, then there can be no WNU. Hence, we may assume that 
$\alpha = \{0,2\}$ and $\beta=\{1,2\}$, which implies that 
$\{0,2\}$, $\{1,2\}$, and $\{0,2\}^{2}\cup \{1,2\}^{2}$ are in $\Gamma$.
 Thus, in this section we have the following assumptions:
\begin{enumerate}
    \item $\Gamma$ is a finite constraint language on $A = \{0,1,2\}$
    \item $\Gamma$ contains the relations $\{0\},\{1\},\{2\},\{0,2\},\{1,2\}$, 
    $\{0,2\}^{2}\cup \{1,2\}^{2}$,
    and the equality relation
    \item $\Gamma$ is preserved by the semilattice operation $s_2$
    \item $\Pol(\Gamma)$ is $\{0,2\}\{1,2\}$-projective
    
    \item all relations in $\Gamma$ are of arity at most $k\ge 4$
    \item $\Gamma$ contains all relations of arity at most $k$
    that can be pp-defined from $\Gamma$
\end{enumerate}
\noindent  These assumptions will hold up to, but not including, Theorem~\ref{SpecialClassification}. 
Similar to a $01$-stable operation, we define $0$-stable and $1$-stable operations, 
where $f$ is called \emph{$c$-stable}
if it is idempotent, $f(x,c) = x$, and $f(x,2) = 2$.
Note that by adding a dummy variable we can make it a $01$-stable operation.

\subsection{Finiteness of the language}
In this subsection we will derive some properties of 
$\Gamma$ based on the facts that 
$\Gamma$ is finite and 
$\Pol(\Gamma)$ is $\{0,2\}\{1,2\}$-projective.

For $B,C\subseteq A$ and $n\ge 3$ by $nu_{B,C}^{n}$ we denote 
the operation of arity $n$ defined on a tuple 
$(a_1,\ldots,a_n)$ as follows. 
If $a_{1} = \dots = a_n$ then it returns $a_1$.
If $a_{1} = \dots = a_{i-1} = a_{i+1} =\dots = a_n\in B$
and 
$a_{i}\in C$ for some $i$, then it returns $a_{j}$ for $j\neq i$.
Otherwise, it returns $2$.

\begin{lem}\label{almostNU}
Suppose $\rho$ is a relation of the minimal arity $N$ in $\Gamma$ 
that is not preserved by $nu_{B,C}^{2k}$.
Then, after some permutation of variables, $\rho$ satisfies one of the following conditions:
\begin{enumerate}
\item 
    There exist $\alpha_{1},\alpha_{2}\in\rho$ 
    such that 
    $nu_{B,C}^{2k}(\alpha_{1},\alpha_{2},\alpha_{2},\ldots,\alpha_{2}) = \beta\notin\rho$, 
    $\alpha_{2}(i)\in B$ and
    $\beta(i) = \alpha_{2}(i)\neq 2$ for $i\in\{1,\ldots,N-1\}$, 
    $\beta(N) = 2$, 
    and
    $\alpha_{1}(i)\neq \alpha_{2}(i)$ for every $i$.
    Moreover, if $C = A$ then $N = 2$, $\alpha_{2}(N)\notin B$, and 
    $\alpha_{1} = (2, 2)$.
    \item There exist $\alpha_{1},\alpha_{2},\alpha_{3}\in\rho$ 
    such that 
    $nu_{B,C}^{2k}(\alpha_{1},\alpha_{2},\alpha_{3},\alpha_{3},\ldots,\alpha_{3}) = \beta\notin\rho$, 
    $\beta(i) = \alpha_{3}(i)\neq 2$ for $i\in\{1,\ldots,N-1\}$, 
    $\beta(N) = 2$, $\alpha_{1}(N),\alpha_{2}(N)\in C$, 
    $\alpha_{3}(i) \in B$ for every $i$, and 
    $\alpha_{1}(i)\neq \alpha_{3}(i)$ or $\alpha_{2}(i)\neq \alpha_{3}(i)$ for every $i$. 
    Moreover, if $2\in C$, then $\alpha_{1}(i),\alpha_{2}(i)\in\{\alpha_{3}(i),2\}$ for every $i$, 
    $\alpha_{1}(N) = \alpha_{2}(N) = 2$, and $N=3$.
        
\end{enumerate}
\end{lem}

\begin{proof}
Note that the hypotheses imply that $N\geq 2$.
Consider tuples 
$\alpha_{1},\ldots,\alpha_{2k}\in\rho$ such that 
$nu_{B,C}^{2k}(\alpha_{1},\ldots,\alpha_{2k})=\beta\notin \rho$.
Since $\rho$ is of the minimal arity, 
$nu_{B,C}^{2k}$ preserves any proper projection of $\rho$.
For example, $nu_{B,C}^{2k}$ preserves $\proj_{1,\ldots,N-1}(\rho)$, 
therefore 
$$nu_{B,C}^{2k}(\proj_{1,\ldots,N-1}(\alpha_1),\ldots,\proj_{1,\ldots,N-1}(\alpha_{2k}))\in \proj_{1,\ldots,N-1}(\rho).$$
Hence, there exists a tuple $\beta_{N}\in\rho$ that differs from $\beta$ only in the $N$-th element. 
Similarly, for every $i$
there exists a tuple $\beta_{i}\in\rho$ that differs from $\beta$ only in the $i$-th element. 

Let us consider the matrix $M$ 
whose columns are $\alpha_{1},\ldots,\alpha_{2k}$. 
Assume that 
for some $j$ we have 
$\alpha_{1}(j) = \dots = \alpha_{2k}(j)$.
Then if we substitute constant 
$\alpha_{1}(j)$ for the $j$-th variable of $\rho$ we get a relation of smaller arity that is not preserved by $nu_{B,C}^{2k}$, which contradicts our assumptions.
Therefore in any row of the matrix there should be elements that are not equal.
Also, all rows of the matrix should be different, since otherwise we
can identify the corresponding variables of $\rho$ to obtain 
a relation of smaller arity that is not preserved by 
$nu_{B,C}^{2k}$.

Assume that $\beta$ has at least two elements equal to $2$ and they appear at the $i$-th and $j$-th positions.
Then $s_2(\beta_{i},\beta_{j}) = \beta$, which contradicts the fact that 
$s_2$ preserves $\rho$.
Assume that $\beta$ has no $2$. 
Every row of the matrix must have exactly one element different from all other elements,
and the result of applying $nu_{B,C}^{2k}$ to every row is the most popular element of the row. Since we have $N\le k$ rows and $2k$ columns, there should be a column among $\alpha_{1},\ldots,\alpha_{2k}$ equal to $\beta$, which contradicts our assumptions.
W.l.o.g., we assume that $\beta(N) = 2$ 
(otherwise we just permute variables of $\rho$).

Put $\alpha_{i}' = \proj_{1,\ldots,N-1}(\alpha_{i})$ 
for every $i$, $\beta' = \proj_{1,\ldots,N-1}(\beta)$.
We know that for every $i\in\{1,\dots,N-1\}$ the $i$-th row contains exactly one element different from the others.
Since $N\le k$, at least $k+1$ of the tuples 
$\alpha_{1}',\ldots,\alpha_{2k}'$ are equal to $\beta'$.
W.l.o.g., we assume that 
$\alpha_{s+1}' = \ldots = \alpha_{2k}' = \beta'$, $\alpha_{i}'\neq\beta'$ 
for every $i\in\{1,\ldots,s\}$, where $s< N$.  
Since not all elements of the first row of the matrix are equal, 
we have $s\ge 1$.

Replace the first column of the matrix $M$ by $\alpha_{2k}$
and consider 
$nu_{B,C}^{2k}(\alpha_{2k},\alpha_{2},\alpha_{3},\ldots,\alpha_{2k}) 
= \gamma_{1}$.
Since $\alpha_{1}'\neq \beta'$, there should be a row
in the new matrix whose elements are equal.
Therefore, $\gamma_{1}\in\rho$, otherwise 
we consider $\alpha_{2k},\alpha_{2},\alpha_{3},\ldots,\alpha_{2k}$
instead of the original sequence, and get a contradiction 
with the fact that every row should contain at least two different elements.
By the definition $\proj_{1,\ldots,N-1}(\gamma_{1}) = \beta'$, 
therefore $\gamma_{1}(N)\neq 2$.
This means that 
at least $2k-2$ among the elements
$\alpha_{2}(N),\ldots,\alpha_{2k}(N)$ are equal 
to $\alpha_{2k}(N)$.
In the same way we can substitute $\alpha_{j}$ instead of $\alpha_{i}$ 
for any $i\in\{1,\dots,s\}$ and $j\in\{s+1,\dots,2k\}$ to conclude that 
$2k-2$ elements among the elements 
$\alpha_{1}(N),\dots,\alpha_{i-1}(N),\alpha_{i+1}(N),\dots,\alpha_{2k}(N)$ 
are equal to $\alpha_{j}(N)$.
Since $\beta(N) = 2$, we also have
$\alpha_{i}(N)\neq\alpha_{j}(N)$.
This implies that 
$\alpha_{s+1}(N) = \dots = \alpha_{2k}(N)$, 
and one of the following cases holds: 
\begin{enumerate}
    \item $s = 1$, $\alpha_{2} = \dots = \alpha_{2k}$;
    \item $s = 2$, $\alpha_{1}(N),\alpha_{2}(N)\in C$,
    $\alpha_{3} = \dots = \alpha_{2k}$, 
    $\alpha_{3}(N)\in B$.
\end{enumerate}

If $s=1$ then $\alpha_{1}(i)\neq \alpha_{2}(i)$ for every $i$ since there should be different elements in every row. 
Suppose $A=C$. Since $s_{2}(\alpha_{1},\alpha_{2})=(2,\dots,2)\in \rho$, we can take $\alpha_{1}=(2,\dots,2)$. 
Since $\alpha_{2}(i)\neq \alpha_{1}(i)$ and all rows are different, 
we have $N=2$ and $\alpha(2)(N)\notin B$.

If $s=2$, then $\alpha_3(i) \in B$ for all $i<N$ since $\beta(i) \neq 2$.
It remains to get an extra condition for the case when $2\in C$
and $s=2$.
In this case we replace 
$\alpha_{1}$ by $\delta_{1} = s_2(\beta_{1},\beta_{N})$ (see the definition of $\beta_{1}$ and $\beta_{N}$ above),
and $\alpha_{2}$ by $\delta_{2} = s_2(\beta_{2},\beta_{N})$. 
Then 
$\delta_{1}(i),\delta_{2}(i)\in\{\alpha_{3}(i),2\}$ for every $i$, 
and 
$nu_{B,C}^{2k}(\delta_1,\delta_2,\alpha_{3},\ldots,\alpha_{3})= \beta$.
Note that $\delta_{1}$ can differ from $\alpha_{3}$ in at most 2 coordinates and one of them is the $N$-th coordinate.
The same is true for $\delta_{2}$.
Therefore, since $\delta_{1},\delta_{2},\alpha_{3}$ (viewed as a matrix with 3 columns)
cannot have equal rows,
$N$ should be equal to 3.
\end{proof}

Our plan is to apply the above lemma to define relations that will be used to prove hardness results.

\begin{lem}\label{ArityThreeRelations}
One of the following cases holds.
\begin{enumerate}
\item $(\{0,2\}\times\{0,2\}\times\{0,2\})\setminus\{(0,0,2)\}\in\Gamma$,
\item $(\{1,2\}\times\{1,2\}\times\{1,2\})\setminus\{(1,1,2)\}\in\Gamma$,
\item $(\{0,2\}\times\{0,2\}\times\{1,2\})\setminus\{(0,0,2)\}\in\Gamma$,
\item $(\{1,2\}\times\{1,2\}\times\{0,2\})\setminus\{(0,0,2)\}\in\Gamma$.
\end{enumerate}
\end{lem}

\begin{proof}
Since $nu_{A,A}^{2k}$ is not $\{0,2\}\{1,2\}$-projective,
$\Gamma$ is not preserved by $nu_{A,A}^{2k}$.
By Lemma \ref{almostNU} there exists a relation $\rho\in\Gamma$
satisfying one of the two cases. The first case is not possible because $nu_{A,A}^{2k}(\alpha_{1},\alpha_{2},\alpha_{2},\ldots,\alpha_{2}) = \alpha_{2}$ for every $\alpha_{1},\alpha_{2}$.
Thus, $\rho$ is of arity 3 and there exist $\alpha_{1},\alpha_{2},\alpha_{3}\in\rho$ 
    such that 
    $nu_{A,A}^{2k}(\alpha_{1},\alpha_{2},\alpha_{3},\alpha_{3},\ldots,\alpha_{3}) = \beta\notin\rho$, 
    $\beta(i) = \alpha_{3}(i)\neq 2$ for $i\in\{1,2\}$, 
    $\beta(3) = 2$,  
    $\alpha_{1}(i)\neq \alpha_{3}(i)$ or $\alpha_{2}(i)\neq \alpha_{3}(i)$,
    $\alpha_{1}(i),\alpha_{2}(i)\in\{\alpha_{3}(i),2\}$ for every $i$, 
    $\alpha_{1}(3) = \alpha_{2}(3) = 2$.
    Suppose $\alpha_{3} = (a_1,a_2,a_3)$. 
    We know that $(a_{1},a_{2},a_{3}), (2,a_{2},2),(a_1,2,2)\in \rho$ (as $\alpha_{3},\alpha_{1},\alpha_{2}$) 
    and $(a_{1},a_{2},2)\notin\rho$ (as $\beta$).
    Since the semilattice $s_2$ applied to $(a_{1},a_{2},0)$ and 
    $(a_{1},a_{2},1)$ gives $(a_{1},a_{2},2)\notin\rho$, 
    $(a_{1},a_{2},c)\in\rho$ implies $c=a_{3}$.
    
    Since $a_{1},a_{2},a_{3}\in\{0,1\}$, 
    at least two of them are equal.
    We consider three cases:
    
    Case 1. $a_{2} = a_{3}$.
    Put $\delta(x,y,z) = \exists t\;\rho(x,a_{2},t)\wedge \rho (y, t,z)$. 
    It is not hard to see that 
    $(a_{1},a_{1}, 2)\notin \delta$ (we need to put $t = a_{2} = a_{3}$ and, therefore, $z = a_{3}$), 
    $(a_{1},2,2),(a_{1},a_{1},a_{3})\in\delta$ (put $t = a_{2} = a_{3} $),
    $(2,2,2),(2,a_{1},2)\in\delta$ (put $t = 2$).

    Case 2. $a_{1} = a_{3}$.
    Put $\delta(x,y,z) = \exists t\;\rho(a_{1},x,t)\wedge \rho (t, y,z)$. 
    It is not hard to see that 
    $(a_{2},a_{2}, 2)\notin \delta$ (we need to put $t = a_{1} = a_{3}$ and, therefore, $z = a_{3}$), 
    $(a_{2},2,2),(a_{2},a_{2},a_{3})\in\delta$ (put $t = a_{1} = a_{3} $),
    $(2,2,2),(2,a_{2},2)\in\delta$ (put $t = 2$).    
    
    Case 3. $a_{1} = a_{2}$. In this case we put $\delta = \rho$.
    
    Thus, we have a relation $\delta\in\Gamma$ 
    satisfying the following properties:
    for some $a,b\in\{0,1\}$ we have 
    $(a,a,b),(2,a,2),(a,2,2),(2,2,2)\in \delta$, 
    $(a,a,2)\notin\delta$.
    Since $s_2$ preserves $\delta$, 
    $(a,a,c)\in\delta$ implies $c=b$.
    
    Put $\delta_{1} = \delta\cap (\{a,2\}\times\{a,2\}\times\{b,2\})$,
$\delta_{2}(x,y,z) = 
\exists t \;
\delta_{1}(x,y,t)\wedge \delta_{1}(x,z,t)
\wedge \delta_{1}(z,x,t).$

Assume that 
$(a,a,2)\in\delta_{2}$.
Since we need to set $t=b$, we have
$(a,2,b),(2,a,b)\in\delta_{1}$.
Since $s_2$ preserves $\delta_{1}$ we obtain 
$(2,2,b)\in\delta_{1}$, 
which means that $\delta_{1}=
(\{a,2\}\times\{a,2\}\times\{b,2\})\setminus\{(a,a,2)\}$
and completes this case.

Assume that 
$(a,a,2)\notin\delta_{2}$.
We know that $(a,a,a)\in\delta_{2}$ (put $t=b$), 
$(a,2,2)$, $(2,a,2)$, $(2,2,a)$, $(2,2,2)\in\delta_{2}$ (put $t = 2$).
Put
$$\delta_{3}(x,y,z) = 
\exists x'\exists y'\exists z' \;
\delta_{2}(x,x,x')\wedge
\delta_{2}(y,y,y')\wedge
\delta_{2}(z',z',z)\wedge
\delta_{2}(x',y',z').$$
Let us show that $\delta_{3} = (\{a,2\}\times\{a,2\}\times\{a,2\})\setminus\{(a,a,2)\}$.
If $x = 2$, then we put $x'=z'=2$ and $y'=y$.
If $y=2$ then we put $y'=z'=2$ and $x'=x$.
If $x=y=z=a$, then we put $y'=z'=y' = a$.
It remains to show that $(a,a,2)\notin\delta_{3}$, which follows from the fact that 
$x'$ should be equal to $a$, 
$y'$ should be equal to $a$, 
$z'$ should be equal to $2$, 
$(a,a,2)\notin\delta_{2}$.
Thus, $\delta_{3} = (\{a,2\}\times\{a,2\}\times\{a,2\})\setminus\{(a,a,2)\}$,
which completes the proof.
\end{proof}

\begin{lem}\label{ArityThreeOrTransZero}
One of the following conditions holds.
\begin{enumerate}
\item $(\{0,2\}\times\{0,2\}\times\{0,2\})\setminus\{(0,0,2)\}\in\Gamma$;
\item there exists a relation $\delta\in\Gamma$ such that 
$(0,1),(2,2)\in\delta$, 
$(0,2)\notin\delta$;
\end{enumerate}
\end{lem}
\begin{proof}
Since $nu_{\{0\},A}^{2k}$ is not $\{0,2\}\{1,2\}$-projective,
$\Gamma$ is not preserved by $nu_{\{0\},A}^{2k}$.
By Lemma \ref{almostNU} there exists a relation $\rho\in\Gamma$
satisfying one of the two cases of Lemma \ref{almostNU}.
Suppose we have the second case of Lemma \ref{almostNU}, where we import that lemma's notation.
Then $\alpha_{3} = (0,0,0), \alpha_{1} = (2,0,2), \alpha_{2}=(0,2,2)$,
$\beta = (0,0,2)$
(or $\alpha_{1}$ and $\alpha_{2}$ can be switched).
Since $s_2$ preserves $\rho$, we have $(2,2,2)\in\rho$.

Let 
$\delta_{0} = \rho\cap\{0,2\}^{3}$.
Consider two cases:

Case 1. $(2,2,0)\in\delta_{0}$. 
For $\epsilon(x,y)=\delta_{0}(y,y,x)$ 
we have $\epsilon=
\begin{pmatrix}
0 & 0& 2\\
0 & 2 & 2
\end{pmatrix}$.

Case 2. $(2,2,0)\notin\delta_{0}$. 
Put $\epsilon(x,y)=\exists t \; \delta_{0}(t,x,y)
\wedge \delta_{0}(x,t,y)$.
If $(2,0)\in\epsilon$, then $(0,2,0),(2,0,0)\in\delta_{0}$,
which, using the semilattice $s_2$, implies 
$(2,2,0)\in\delta_{0}$. This contradiction
shows that $(2,0)\notin\epsilon$ and $\epsilon=
\begin{pmatrix}
0 & 0& 2\\
0 & 2 & 2
\end{pmatrix}$.

Then the required relation $(\{0,2\}\times\{0,2\}\times\{0,2\})\setminus\{(0,0,2)\}$
can be pp-defined by
$$\delta_{1}(x,y,z) = \exists x'\exists y' \exists z'\;
\epsilon(x',x)\wedge \epsilon(y',y)\wedge 
\epsilon(z,z')\wedge \delta_{0}(x',y',z').$$

Suppose we have the first case in Lemma \ref{almostNU}.
Then $\beta = (0,2), \alpha_{1} = (2,2), \alpha_{2} = (0,1)$, 
$s_2(\alpha_{1},\alpha_{2}) = (2,2)\in \rho$.
Thus, $\rho$ satisfies the conditions of item 2.
\end{proof}

Similarly, if we switch 0 and 1 we get the following lemma.
\begin{lem}\label{ArityThreeOrTransOne}
One of the following cases holds.
\begin{enumerate}
\item $(\{1,2\}\times\{1,2\}\times\{1,2\})\setminus\{(1,1,2)\}\in\Gamma$
\item there exists $\delta\in\Gamma$ such that 
$(1,0),(2,2)\in\delta$, 
$(1,2)\notin\delta$
\end{enumerate}
\end{lem}

\begin{lem}\label{ArityThreeRelationsStrong}
One of the following cases holds:
\begin{enumerate}
\item $(\{0,2\}\times\{0,2\}\times\{0,2\})\setminus\{(0,0,2)\},(\{1,2\}\times\{1,2\}\times\{1,2\})\setminus\{(1,1,2)\}\in\Gamma$
\item $(\{1,2\}\times\{0,2\}\times\{0,2\})\setminus\{(1,0,2)\}, (\{0,2\}\times\{0,2\}\times\{0,2\})\setminus\{(0,0,2)\}\in\Gamma$
\item $(\{0,2\}\times\{1,2\}\times\{1,2\})\setminus\{(0,1,2)\}, (\{1,2\}\times\{1,2\}\times\{1,2\})\setminus\{(1,1,2)\}\in\Gamma$
\end{enumerate}
\end{lem}
\begin{proof}
We apply Lemmas \ref{ArityThreeOrTransZero} and
\ref{ArityThreeOrTransOne}.
If we have the first case in both of them, then the first condition holds and we are done.

Assume that we have the second case in both of them and the corresponding relations are $\delta_{0}$ and $\delta_{1}$, respectively.
Put $$\delta(x,y) = \delta_{0}(x,y)\wedge \delta_{1}(y,x)\wedge (x\in\{0,2\})
\wedge (y\in\{1,2\}).$$
It is not hard to see that 
$\delta = \begin{pmatrix} 0 & 2\\ 1& 2\end{pmatrix}$.
Using this relation we can easily switch $0$ and $1$ in relations.
By Lemma~\ref{ArityThreeRelations}
we have $\rho = (\{a,2\}\times\{a,2\}\times\{b,2\})\setminus\{(a,a,2)\}\in\Gamma$
for some $a,b\in\{0,1\}$.
If $a=b=1$, then we use the relation $\rho$ and the relation defined by the following 
pp-definition to satisfy condition 1:
$$\rho_{1}(x,y,z) = \exists x'\exists y'\exists z'\;
\rho(x',y',z') \wedge \delta(x,x')\wedge \delta(y,y')\wedge \delta(z,z').$$
If $a=b=0$ then the same pp-definition, but with the arguments of $\delta$ switched, applies.

If $a\neq b$, and $(a,b)=(0,1)$, then the relations we need to satisfy condition 1 can be defined by:
$$\rho_{2}(x,y,z) = \exists z'\;
\rho(x,y,z') \wedge \delta(z,z'),$$
$$\rho_{3}(x,y,z) = \exists x'\exists y'\;
\rho(x',y',z) \wedge \delta(x',x)\wedge \delta(y',y).$$
If $(a,b)=(1,0)$ then the same pp-definitions, but with the arguments of $\delta$ switched, applies.

Assume that we have 
the first case in Lemma \ref{ArityThreeOrTransZero}
and 
the second case in Lemma \ref{ArityThreeOrTransOne}.
Thus, $\rho_{4} = (\{0,2\}\times\{0,2\}\times\{0,2\})\setminus\{(0,0,2)\}\in\Gamma$
and 
there exists $\delta\in\Gamma$ such that 
$(1,0),(2,2)\in\delta$, 
$(1,2)\notin\delta$.
Put
$$\rho_{5}(x,y,z)= \exists x' \;\rho_{4}(x',y,z) \wedge \delta(x,x')
\wedge x\in\{1,2\}\wedge x'\in\{0,2\}.$$
Then $\rho_{4}$ and $\rho_{5}$ satisfy the second condition of this lemma.

Assume that we have 
the second case in Lemma \ref{ArityThreeOrTransZero}
and 
the first case in Lemma \ref{ArityThreeOrTransOne}.
Thus, $\rho_{6} = (\{1,2\}\times\{1,2\}\times\{1,2\})\setminus\{(1,1,2)\}\in\Gamma$
and 
there exists $\delta\in\Gamma$ such that 
$(0,1),(2,2)\in\delta$, 
$(0,2)\notin\delta$.
Put
$$\rho_{7}(x,y,z)= \exists x' \; \rho_{6}(x',y,z) \wedge \delta(x,x')
\wedge x\in\{0,2\}\wedge x'\in\{1,2\}.$$
Then $\rho_{6}$ and $\rho_{7}$ satisfy the third condition of this lemma.
\end{proof}

\subsection{Connection with Boolean functions}

A mapping $f:\{0,1\}^{n}\to\{0,1\}$ is called \emph{a Boolean function}.
We say that 
a Boolean function $f$ of arity $n$ is represented by $\Gamma$
if there exists a relation $R$ of arity $n+1$ pp-definable using $\Gamma$
such that for 
all $a_{1},\ldots,a_{n}\in\{0,1\}$, $b\in\{0,1,2\}$
$$(a_{1},\ldots,a_{n},b)\in R\Leftrightarrow  f(a_{1},\ldots,a_{n}) = b.$$
For example, an AND-type relation is a representation of conjunction, 
an OR-type relation is a representation of disjunction.
A $k$-ary Boolean function $f$ is \emph{monotonic} if $f(x_1,\ldots,x_k)\leq f(y_1,\ldots,y_k)$ whenever 
$x_{i}\le y_{i}$ for every 
$i\in\{1,2,\dots,k\}$.
A $k$-ary Boolean function is \emph{linear} if it is of the form $c_{1}x_1+\ldots+c_{k}x_k+c_{0}$ for 
$c_{0},c_{1},\dots,c_{k}\in\{0,1\}$.

\begin{lem}\label{IsClone}
The set of all Boolean functions represented by $\Gamma$ is a clone generated by one of the following sets:
$\{0,1\}$, $\{+,0,1\}$, $\{\wedge,0,1\}$, $\{\vee,0,1\}$, $\{\wedge,\vee,0,1\}$ or $\{\wedge,\neg,0,1\}$.
\end{lem}
\noindent Note that this is true for 
arbitrary constraint languages containing 
$\{0\}$, $\{1\}$ and $=$.

\begin{proof}
Notice that 
constants 0 and 1 are represented by 
the singleton relations $\{0\}$ and $\{1\}$, which are in $\Gamma$, and 
a projection can be represented 
by the equality relation, which is assumed to be in $\Gamma$.

Suppose we have
Boolean functions 
$g_{1},\dots,g_{n}$ 
of arity $m$ represented by 
relations 
$R_{1},\ldots,R_{n}$, respectively.
Suppose 
$f$ is Boolean function of arity $n$
represented by a relation $R$.
Then the composition 
$h= f(g_1,\dots,g_n)$ can be represented by
the relation defined by 
\begin{align*}
R'(x_{1},\ldots,x_{m},z)
=
\exists y_1\dots\exists y_n
R_{1}(x_{1},\dots,x_{m},y_{1})\wedge
\dots
\wedge
R_{n}(x_{1},\dots,x_{m},y_{n})
\wedge 
R(y_1,\dots,y_n,z).
\end{align*}
Since we are closed under composition and above $\{0,1\}$ we can read from Post's Lattice that we are one of the clones in the statement of the lemma.
\end{proof}

By $\BoolClo(\Gamma)$ we denote the set of all 
Boolean functions represented by $\Gamma$, which is a clone 
by Lemma~\ref{IsClone}.

Let 
$\alpha_{1},\ldots,\alpha_{N}$, where $N = 2^{k+1}-1$,
be the set of all tuples from $\{0,1\}^{k+1}$
except for the tuple $(1,1,\dots,1)$.
Thus, $\alpha_{i}$ can be viewed as a binary representation of 
the number $(i-1)$.

Let us define an operation $f_{\Gamma}$ on $A$ of arity $N$.
If there exists a unique Boolean function $h\in\BoolClo(\Gamma)$ 
of arity $k+1$ such that 
$h(\alpha_{i}) = a_{i}$ for every $i$, then 
put $f(a_{1},\ldots,a_{N}) = h(1,1,\ldots,1)$.
Otherwise, put $f(a_{1},\ldots,a_{N}) = 2$.

\begin{lem}
Suppose Boolean functions $h_{1}$ and $h_{2}$ of arity $n\ge 2$
differ only on the tuple $(1,1,\dots,1)$.
Then conjunction can be generated from $h_{1},h_{2}$ and constants $0,1$.
\end{lem}
\begin{proof}
Since the functions $h_1$ and $h_2$ differ just on one tuple,
they cannot be both linear.
Assume that they are monotonic.
Since $h_1(1,1,\dots,1) \neq h_2(1,1,\dots,1)$, 
we have $h_1(x_1,\dots,x_n) = 0$ 
and $h_2(x_1,\dots,x_n) = x_1\wedge\dots\wedge x_n$
(or we should switch $h_1$ and $h_2$).
Thus, we generated conjunction in this case.

Assume that one of the functions $h_1$ and $h_2$ is not monotonic.
Then 
$h_{1},h_{2},0,1$ are not contained in any maximal clone 
of Boolean functions.
Hence, by \cite{Post}, any function, including conjunction, 
can be generated from $h_1$, $h_2$, $0$, and $1$.
\end{proof}

Therefore, 
if $\BoolClo(\Gamma)$ does not contain conjunction, 
then the word ``unique'' can be removed from the definition of $f_{\Gamma}$.
Also, since constants 0 and 1 are in $\BoolClo(\Gamma)$, 
the operation 
$f_{\Gamma}$ is idempotent in this case.

\begin{lem}\label{fGammaIsNotProjective}
Suppose $\BoolClo(\Gamma)$ does not contain conjunction, 
then $f_{\Gamma}$ is not 
$\{0,2\}\{1,2\}$-projective.
\end{lem}

\begin{proof}
Assume that 
$f_{\Gamma}$ is a $\{0,2\}\{1,2\}$-projection to the $i$-th coordinate.
Let $\alpha_{i}$ have $0$ at the $j$-th coordinate, 
and $h$ is the $(k+1)$-ary Boolean projection to the $j$-th coordinate.
Then 
$f_{\Gamma}(h(\alpha_{1}),h(\alpha_{2}),\dots,h(\alpha_{N}))=h(1,1,\dots,1) = 1$
and $h(\alpha_{i})=0$. Contradiction.
\end{proof}

\begin{lem}\label{0StableNegationImplies1Stable}
Suppose the negation Boolean function is represented by $\Gamma$, 
and $\Gamma$ is preserved by a 0-stable operation $v_{0}$.
Then $v_{0}$ is 1-stable.
\end{lem}

\begin{proof}
Since $v_{0}$ preserves the relation representing the negation and
$v_{0}(1,0) = 1$, we have 
$v_{0}(0,1)=0$. 
Combining this with the fact that 
$v_{0}$ is idempotent and
$v_{0}$ is a $\{0,2\}\{1,2\}$-projection to the first coordinate, we obtain that $v_{0}$ is 1-stable.
\end{proof}

\begin{thm}\label{ANDTypeForStable}
Suppose $\Gamma$ is preserved by a 0-stable operation $v_{0}$, 
then $\Gamma$ contains an AND-type relation 
(equivalently $\BoolClo(\Gamma)$ contains conjunction).
\end{thm}

\begin{proof}
Assume that $\BoolClo(\Gamma)$ does not contain conjunction.

By Lemma~\ref{fGammaIsNotProjective}, 
$f_{\Gamma}$ is not $\{0,2\}\{1,2\}$-projective, 
hence there should be a relation in $\Gamma$ that is not preserved by 
$f_{\Gamma}$.
Let $\rho$ be a relation of the minimal arity in $\Gamma$ that is not preserved by $f_{\Gamma}$.
Thus, 
there exists a matrix $M$ whose columns are tuples from $\rho$ 
such that 
$f_{\Gamma}$ applied to the rows of $M$ gives $\beta\notin\rho$.
Let $L$ be the arity of $\rho$.

Note that each row of the matrix should contain different elements, 
since otherwise we could substitute a constant into $\rho$ and reduce the arity of $\rho$ (recall that $f_{\Gamma}$ is idempotent).
Since $\rho$ is a relation of the minimal arity, for every $i$ there should be a tuple $\beta_{i}\in\rho$ such that 
$\beta_{i}$ differs from $\beta$ only in the $i$-th coordinate (otherwise, we consider the projection of $\rho$ to all coordinates but $i$-th, which has the same properties as $\rho$ but smaller arity).

Assume that $\beta$ has at least two elements equal to $2$ and they appear at the $i$-th and $j$-th coordinates.
Then $s_2(\beta_{i},\beta_{j}) = \beta$, which contradicts the fact that 
$s_2$ preserves $\rho$.

Assume that $\beta$ contains exactly one $2$.
W.l.o.g., we assume that $\beta(L) = 2$.
By the definition of 
$f_{\Gamma}$ 
for every $i\in[L-1]$
there exists a Boolean function 
$h_{i}\in\BoolClo(\Gamma)$ (not a constant function)
such that 
the $i$-th row of $M$ 
is 
$(h_{i}(\alpha_{1}),\dots,h_{i}(\alpha_{N}))$.

Assume that there exists a tuple $\gamma\in\rho$ such that $\gamma(L) = 2$ and $\gamma(i) \in\{0,1\}$
for $i\in[L-1]$.
By Lemma~\ref{0StableNegationImplies1Stable}, 
we have two cases:

Case 1. $\BoolClo(\Gamma)$ does not contain negation.
It follows from Post's lattice (see \cite{Post}) that all functions in $\BoolClo(\Gamma)$ are monotonic.
Thus, $h_{i}$ is a monotonic function that is not a constant, 
therefore 
$\beta(i) = h_{i}(1,1,\ldots,1) = 1$ for every $i\in[L-1]$.
Applying the 0-stable operation $v_{0}$ to $\beta_{L}$ and $\gamma$
we obtain $(1,1,\dots,1,2) = \beta$, which contradicts the fact that 
$v_{0}$ preserves $\rho$.

Case 2. $v_{0}$ is 1-stable.
Applying $v_{0}$ to $\beta_{L}$ and $\gamma$
we obtain $\beta$, which contradicts the fact that
$v_{0}$ preserves $\rho$.

Thus, we proved that 
$\rho$ does not contain a tuple whose last element is 2 and the remaining elements are from $\{0,1\}$. For instance, this means that the matrix $M$ contains only 0 and 1.
Let $R_{i}$ be a relation pp-definable from $\Gamma$ 
representing $h_{i}$.
Define a new relation by 
\begin{align*}
R(x_{1},\ldots,x_{k+1},z)
=&\\
\exists y_1\dots\exists y_{L-1}\;&
R_{1}(x_{1},\dots,x_{k+1},y_{1})\wedge
\dots
\wedge
R_{L-1}(x_{1},\dots,x_{k+1},y_{L-1})
\wedge 
\rho(y_1,\dots,y_{L-1},z).
\end{align*}
We can check that for each 
$\alpha_{i}$ 
there exists a unique $b_{i}\in\{0,1\}$ 
such that 
$\alpha_{i}b_{i}\in R$.
Note that $b_{i}$ is the $i$-th element of the last 
row of the matrix.
Since
$(1,1,\ldots,1,h_i(1,1,\dots,1)) \in R_{i}$ 
and $\beta(i) = \beta_{L}(i) = h_i(1,1,\dots,1)$
for every $i\in[L-1]$, we have
$(1,1,\dots,1,\beta_{L}(L))\in R$,
which means that $R$ represents the 
corresponding Boolean function 
and $f_{\Gamma}(b_{1},\dots,b_{N})=\beta_{L}(L)$.
This contradicts the fact that 
$\beta(L) = 2$.

It remains to consider the case when $\beta$ has no $2$. 
Therefore,
for every $i\in[L]$
there exists a Boolean function 
$h_{i}\in\BoolClo(\Gamma)$ (not a constant function)
such that 
the $i$-th row of $M$ 
is 
$(h_{i}(\alpha_{1}),\dots,h_{i}(\alpha_{N}))$.
It follows from Post's Lattice (see \cite{Post}) that we have only the following two cases.

Case 1. $\BoolClo(\Gamma)$ contains only linear functions.
Consider a system of linear equations 
$$
\begin{cases}
h_{1}(x_{1},\ldots,x_{k+1}) = h_{1}(1,\dots,1)\\
h_{2}(x_{1},\ldots,x_{k+1}) = h_{2}(1,\dots,1)\\
\dots\\
h_{L}(x_{1},\ldots,x_{k+1}) = h_{L}(1,\dots,1)
\end{cases}
$$
Since $L\le k$, there should be a solution of this system of equations
different from 
$(1,1,\dots,1)$.
Let this solution be equal to $\alpha_{i}$ for $i\in[N]$.
Therefore, the $i$-th column of the matrix $M$ 
is equal to 
$(h_{1}(1,\dots,1),h_{2}(1,\dots,1),\dots,h_{L}(1,1,\dots,1))$, 
which is equal to $\beta$. 
This contradicts the fact that $\beta\notin\rho$.

Case 2. $\BoolClo(\Gamma)$ is the clone generated by disjunction and constants.
Since every operation $h_{j}$ is monotonic and not constant, 
$\beta = (1,1,\dots,1)$.
Since $h_{j}$ is not a constant for every $j\in[L]$, 
we can choose $n_{j}$ such that 
$h_{j}(x_{1},\dots,x_{k+1})\ge x_{n_{j}}$.
Let $\alpha_{i}$ be the tuple having 1 only on 
coordinates $n_{1},\dots,n_{L}$ (such $\alpha_{i}$ exists because $L<k+1$).
Then 
the $i$-th column of the matrix $M$ 
is equal to $(1,1,\dots,1)$, which contradicts the fact that $\beta\notin\rho$.
\end{proof}

In the same way we can prove the following theorem.
\begin{thm}\label{ORTypeForStable}
Suppose $\Gamma$ is preserved by a 1-stable operation, 
then $\Gamma$ contains an OR-type relation 
(equivalently $\BoolClo(\Gamma)$ contains disjunction).
\end{thm}

\subsection{The proof of the EGP classification}

\begin{lem}\label{NoStableOperation}
Suppose $\Pol(\Gamma)$ does not have a $0$-stable operation.
Then $\Gamma$ contains one of the two relations
$\begin{pmatrix}
0 & 1 & 2 & 0 & 2\\
0 & 0 & 0 & 2 & 2 
\end{pmatrix}$,
$\begin{pmatrix}
0 & 1 & 2 & 0 & 2\\
1 & 1 & 1 & 2 & 2 
\end{pmatrix}$.
\end{lem}
\begin{proof}
Let $\delta_{1}$ and $\delta_2$ be the 
relations generated from 
$\begin{pmatrix}
1 & 0\\
0 & 2 
\end{pmatrix}$
and 
$\begin{pmatrix}
1 & 0\\
1 & 2 
\end{pmatrix}$
using $\Pol(\Gamma)$, respectively.
Note that since 
$s_2\in\Pol(\Gamma)$ we have $(2,2)\in \delta_{1}$ and $(2,2)\in \delta_2$.
Assume that 
$(1,2)\in \delta_{1}$ and $(1,2)\in\delta_{2}$.
Then there exist operations 
$f_{1}, f_{2}\in\Pol(\Gamma)$ 
such that 
$f_{1}(1,0) = f_{2}(1,0) = 1$, 
$f_{1}(0,2) = f_{2}(1,2) = 2$.
Since $f_{1}$ and $f_{2}$ are $\{0,2\}\{1,2\}$-projective, 
they should be $\{0,2\}\{1,2\}$-projections to the first coordinate and $f_{1}(2,x) =f_{2}(2,x) = 2$. Furthermore, $f_1(0,2), f_2(0,2) \in \{0,2\}$ as 
$f_{1},f_{2}$ preserves $\{0,2\}\in \Gamma$.
Put $f(x,y) = f_{1}(f_{2}(x,y),y)$,
then
$f(1,0) = 1$, $f(0,2) = f(1,2) = f(2,0) = 2$.
Therefore, $f$ is a $0$-stable operation. Contradiction.

Assume that
$(1,2)\notin \delta_{1}$.
Then by Lemma~\ref{ArityThreeRelationsStrong}, 
$\psi_{1} = (\{0,2\}\times\{a,2\}\times\{a,2\})\setminus\{(0,a,2)\}\in\Gamma$
for some $a\in\{0,1\}$.
Then one of the two necessary relations is defined 
by the formula 
$\delta_{1}'(x,y) = \exists z \;\delta_{1}(x,z) \wedge \psi_{1}(z,a,y)$.

It remains to consider the case when
$(1,2)\notin \delta_{2}$.
Then by Lemma~\ref{ArityThreeRelationsStrong}, 
$\psi_{2} = (\{1,2\}\times\{a,2\}\times\{a,2\})\setminus\{(1,a,2)\}\in\Gamma$ for some $a\in\{0,1\}$.
Then one of the two necessary relations is defined 
by the formula 
$\delta_{2}'(x,y) = \exists z \;\delta_{2}(x,z) \wedge \psi_{2}(z,a,y)$.

\end{proof}

In the same way we can prove the following lemma.
\begin{lem}\label{NoStableOperationOne}
Suppose $\Pol(\Gamma)$ does not have a $1$-stable operation.
Then $\Gamma$ contains one of the two relations
$\begin{pmatrix}
0 & 1 & 2 & 1 & 2\\
0 & 0 & 0 & 2 & 2 
\end{pmatrix}$,
$\begin{pmatrix}
0 & 1 & 2 & 1 & 2\\
1 & 1 & 1 & 2 & 2 
\end{pmatrix}$.
\end{lem}


\begin{lem}\label{lem:Pspace-definitions}
Suppose $\Pol(\Gamma)$ does not have a 01-stable operation, then
there exists $b\in \{0,1\}$ such that $\sigma_{0},\sigma_{1}\in\Gamma$, where 
$$\begin{array}{c}
\sigma_0 = \{ (a_1,a_2,a_3) : a_1 \in A, a_2,a_3 \in \{b,2\}, (a_1 \in \{0,2\} \vee a_2=a_3) \},\\
\sigma_1 = \{ (a_1,a_2,a_3) : a_1 \in A, a_2,a_3 \in \{b,2\}, (a_1 \in \{1,2\} \vee a_2=a_3) \}. 
\end{array}$$
\end{lem}

\begin{proof}
Let 
$\delta_0$ and $\delta_{1}$ be the relations generated from
$\begin{pmatrix}
0&0&1\\ 
1&0&1\\
0&2&2
\end{pmatrix}$
and 
$\begin{pmatrix}
0&0&1\\ 
1&0&1\\
1&2&2
\end{pmatrix}$
using $\Pol(\Gamma)$,
respectively.
Since the semilattice preserves $\delta_{0}$ and $\delta_1$, 
the tuples $(0,2,2),(2,1,2)$ belong to $\delta_{0}$ and $\delta_1$.

If both $\delta_{0}$ and $\delta_1$ contain
$(0,1,2)$ 
then there exist operations 
$f_0$ and $f_1$ in $\Pol(\Gamma)$ such that 
$f_{0}(0,0,1)=f_{1}(0,0,1)= 0$, 
$f_{0}(1,0,1) = f_{1}(1,0,1)=1$,
$f_{0}(0,2,2) = f_{1}(1,2,2)=2$.
Since $f_0$ and $f_1$ are $\{0,1\}\{0,2\}$-projective, 
$f_{0}(2,a,b) = f_{1}(2,a,b)=2$ for all $a,b\in A$. 
Then the operation 
$f(x,y,z)=f_1(f_0(x,y,z),y,z)$ is a 01-stable operation, which contradicts our assumptions.

Assume that 
$\delta_{0}$ does not contain 
$(0,1,2)$.
By Lemma~\ref{ArityThreeRelationsStrong}, 
for some $b\in \{0,1\}$ the relation $\psi=(\{0,2\}\times\{b,2\}\times\{b,2\})\setminus\{(0,b,2)\}$
is in $\Gamma$.
Then 
$\sigma_{0}$ and $\sigma_{1}$ can be defined by 
$$\sigma_{0}(x,y,z) = \exists t \;\delta_{0}(0,x,t)\wedge \psi(t,y,z)\wedge \psi(t,z,y),$$
$$\sigma_{1}(x,y,z) = \exists t \;\delta_{0}(x,1,t)\wedge \psi(t,y,z)\wedge \psi(t,z,y).$$

Assume that 
$\delta_{1}$ does not contain 
$(0,1,2)$.
By Lemma~\ref{ArityThreeRelationsStrong}, 
for some $b\in \{0,1\}$ the relation $\psi=(\{1,2\}\times\{b,2\}\times\{b,2\})\setminus\{(1,b,2)\}$
is in $\Gamma$.
Then 
$\sigma_{0}$ and $\sigma_{1}$ can be defined by 
$$\sigma_{0}(x,y,z) = \exists t \;\delta_{1}(0,x,t)\wedge \psi(t,y,z)\wedge \psi(t,z,y),$$
$$\sigma_{1}(x,y,z) = \exists t \;\delta_{1}(x,1,t)\wedge \psi(t,y,z)\wedge \psi(t,z,y).$$
\end{proof}

Combining the above lemma with Lemma~\ref{PSpaceHardness} 
we obtain the following corollary.

\begin{cor}\label{PSpaceCompleteness}
Suppose $\Pol(\Gamma)$ does not have a $01$-stable 
operation, then 
$\QCSP(\Gamma)$ is PSpace-complete.
\end{cor}

\begin{lem}
Suppose $\Pol(\Gamma)$ does not have a 0-stable or 1-stable operation, then
$\QCSP(\Gamma)$ is co-NP-hard.
\label{lem:coNPHard-nostable}
\end{lem}

\begin{proof}
By Lemmas~\ref{NoStableOperation}
and \ref{NoStableOperationOne},
we have the relations $\delta_{0} = \begin{pmatrix}
0 & 1 & 2 & 0 & 2\\
a_{0} & a_{0} & a_{0} & 2 & 2 
\end{pmatrix}$ and 
$\delta_{1}=\begin{pmatrix}
0 & 1 & 2 & 1 & 2\\
a_{1} & a_{1} & a_{1} & 2 & 2 
\end{pmatrix}$
in $\Gamma$
for some $a_{0},a_{1}\in\{0,1\}$.
By Lemma~\ref{ArityThreeRelationsStrong}, 
for some $b_{0},b_{1}\in\{0,1\}$ the relations $\psi_{0}=(\{a_{0},2\}\times\{b_{0},2\}\times\{b_{0},2\})\setminus\{(a_{0},b_{0},2)\}$,
$\psi_{1}= (\{a_{1},2\}\times\{b_{1},2\}\times\{b_{1},2\})\setminus\{(a_{1},b_{1},2)\}$ are in $\Gamma$.
We define 
$$\sigma_0(x,y,z) = \exists t \;\delta_{0}(x,t)\wedge \psi_{0}(t,y,z)\wedge \psi_{0}(t,z,y),$$
$$\sigma_1(x,y,z) = \exists t \;\delta_{1}(x,t)\wedge \psi_{1}(t,y,z)
\wedge \psi_{1}(t,z,y).$$
If $b_0=b_{1}$ then 
by Lemma \ref{PSpaceHardness}
the problem is PSpace-hard, and therefore co-NP-hard.

Since $\Gamma$ contains $\{0,2\}^{2}\cup \{1,2\}^{2}$,  
if $b_{0} = 1$ and $b_{1}=0$ then by Lemma \ref{lem:coNP-hardness-oldPSpace}
the problem is co-NP-hard.
Assume that  
$b_{0} = 0$ and $b_{1}=1$.
Put
$$\sigma_{1}'(x,y,z) = \exists z'\;
\sigma_{1}(x,1,z')\wedge \sigma_{0}(z',y,z),$$
$$\sigma_{0}'(x,y,z) = \exists z'\;
\sigma_{0}(x,0,z')\wedge \sigma_{1}(z',y,z),$$
and notice that
$\sigma_{0}'$ and $\sigma_1'$ are the 
relations required in Lemma \ref{lem:coNP-hardness-oldPSpace}
for the proof of co-NP-hardness.
\end{proof}


\begin{lem}\label{SecondPPDefinition}
Suppose $\Gamma$ is not preserved by $f_{0,2}$, and
$\Gamma$ is preserved by $s_{0,2}$,
$$\begin{pmatrix}
0&1&1&1&2&2&2\\
1&0&1&2&0&1&2
\end{pmatrix},
\begin{pmatrix}
0&1&1&1&2&2&2\\
0&0&1&2&0&1&2
\end{pmatrix},\begin{pmatrix}
0&1&2&2&2\\
0&1&0&1&2
\end{pmatrix}\in\Gamma,$$
then the relation 
$\delta(x,y,z) = (x\neq 0)\vee (y=z)$ can be pp-defined from $\Gamma$.
\end{lem}
\begin{proof}
First, denote 
the relations from the statement by $R_1$, $R_2$ and 
$R_3$, respectively.
Let $\rho_{1}$ be the relation generated from
$\begin{pmatrix}
0 & 0 & 2\\
1 & 0 & 0\\
1 & 0 & 1
\end{pmatrix}$ using $\Pol(\Gamma)$.
Assume that 
$(0,1,2)\in\rho_{1}$, 
then there exists 
a function $f\in\Pol(\Gamma)$ such that 
$f(0,0,2) = 0$, $f(1,0,0)=1$, and $f(1,0,1) =2$.
Since $f$ is $\{0,2\}\{1,2\}$-projective, 
$f$ always returns the first variable or 2.
Since $f$ preserves $R_1$
and $f(0,0,2) = 0$, we have $f(1,1,a) = 1$ for every $a$.
Since 
$f$ preserves $R_2$
and $f(0,0,2) = 0$, 
we have $f(0,0,a) = 0$ for every $a$.

Let us show that 
$f_{0,2}(x,y,z) = f(s_{0,2}(x,y),y,z)$.
If $x=y$ then 
by the above property we have
$f(s_{0,2}(x,y),y,z) = f(x,y,z) = x$.
Also, if $y=z=0$ then
since 
$f(1,0,0)=1$ and $f(0,0,0)=0$, 
we obtain
$f(s_{0,2}(x,y),y,z) = f(x,y,z) = x$.
If $x=2$ or $y=2$ then 
$f(s_{0,2}(x,y),y,z)$ returns 2.
If 
$x=0$ and $y=1$ then $s_{0,2}(x,y)=2$
and $f(s_{0,2}(x,y),y,z)=2$.
If $x=z=1$ and $y=0$, then 
$f(s_{0,2}(x,y),y,z)=f(x,y,z)=f(1,0,1)=2.$
Thus, we considered all cases and showed that 
$f_{0,2}(x,y,z) = f(s_{0,2}(x,y),y,z)$, 
which contradicts the fact that $f_{0,2}\notin\Pol(\Gamma)$.
Therefore 
$(0,1,2)\notin\rho_{1}$.
Put
$$\rho_{2}(x,y,z) = 
\exists x'\exists y' \exists z'\;
\rho_{1}(x',y',z')
\wedge 
R_{2}(x,x')
\wedge 
R_{3}(y,y')
\wedge 
R_{3}(z',z).$$

Note that $\rho_{1}\subseteq \rho_{2}$, 
and $(0,1,2)\notin\rho_{2}$.
Since $\rho_{1}$ is preserved by $s_2$ 
and $s_{0,2}$, we have 
$s_2
\begin{pmatrix}
0&2\\
0&0\\
0&1
\end{pmatrix}
=
\begin{pmatrix}
2\\
0\\
2
\end{pmatrix}\in\rho_{1}$,
$s_{0,2}
\begin{pmatrix}
0&2\\
1&0\\
1&2
\end{pmatrix}
=
\begin{pmatrix}
2\\
1\\
2
\end{pmatrix}\in\rho_{1}$,
$s_2
\begin{pmatrix}
2&2\\
0&1\\
2&2
\end{pmatrix}
=
\begin{pmatrix}
2\\
2\\
2
\end{pmatrix}\in\rho_{1}$.
Therefore $\rho_{2}$ contains all the tuples $(1,a,b),(2,a,b)$ for every 
$a,b\in A$ 
(we just put $x'=z'=2$ in the definition of $\rho_{2}$).

Define $\delta'(x,y,z) = \rho_{2}(x,y,z)\wedge \rho_{2}(x,z,y)$.
By the previous fact, 
$\delta'$ contains all the tuples with the first element different from 
$0$.
If $\rho_{2}$ contains $(0,1,0)$ then 
we may apply $s_2$ to 
$(0,1,0)$ and $(0,1,1)$ to obtain $(0,1,2)$, which is not from $\rho_{2}$. 
Thus $(0,1,0)\notin\rho_{2}$.
Therefore, $(0,1,0),(0,0,1),(0,1,2),(0,2,1)\notin\delta'$.

Assume that $\delta'$ contains the tuple 
$(0,0,2)$.
Then applying $s_{0,2}$ to 
$(0,1,1)$ and $(0,0,2)$ we get a tuple $(0,1,2)$ which contradicts our assumptions.
Therefore $\delta'=\delta$, which completes the proof.
\end{proof}

\begin{lem}\label{StrangeOneImplies}
Suppose $\Gamma$ is preserved by a 0-stable operation but not preserved by $g_{0,2}$, and 
$R= \begin{pmatrix}
0&1&2&2&2\\
0&1&0&1&2
\end{pmatrix}\in\Gamma$.
Then $\Gamma$ contains a relation
$\delta\in\left\{
\begin{pmatrix}
0&2\\
1&2
\end{pmatrix},
\begin{pmatrix}
0&2&2\\
1&1&2
\end{pmatrix}\right\}$.
\end{lem}

\begin{proof}
Let
$\rho$ 
be the relation generated from 
$\begin{pmatrix}
0&2\\
1&0\\
1&2
\end{pmatrix}$
using $\Pol(\Gamma)$.
Assume that 
$(0,1,2)\in\rho$, then there exists an operation $f\in\Pol(\Gamma)$ such that $f(0,2)= 0$, $f(1,0) = 0$, and
$f(1,2) =2$.
Let us check that $f=g_{0,2}$, 
which will give us a contradiction.
Since $f$ preserves $R$ and $f(0,2) = 0$,
we have $f(0,a) = 0$ for any $a$.
We also know that 
$f(1,0) = 1$, $f(1,1) =1$, $f(1,2) = 2$.
It remains to check that $f$ returns 2 if the first variable equals 2. 
This follows from the fact that $f$ is $\{0,2\}\{1,2\}$-projective.

Therefore
$(0,1,2)\notin\rho$.
Since $\rho$ is preserved by a 0-stable operation
and $(0,1,1),(2,0,2)\in\rho$, we have $(2,1,2)\in\rho$.
Put 
$\delta(x,y) = \rho(x,1,y)\wedge x\in\{0,2\}\wedge y\in\{1,2\}$. 
We know that 
$(0,1),(2,2)\in\delta$, $(0,2)\notin\delta$, which completes the proof.
\end{proof}

Now we are ready to prove the classification of the complexity for 
constraint languages we consider in this section.

\begin{thm}\label{SpecialClassification}
Suppose $\Gamma$ is a finite constraint language on $\{0,1,2\}$   
with constants,
$\Gamma$ is preserved by $s_2$, and 
$\Pol(\Gamma)$ is $\{0,2\}\{1,2\}$-projective. Then $\QCSP(\Gamma)$ is 
\begin{enumerate}
\item 
PSpace-complete, if $\Pol(\Gamma)$ does not contain a 01-stable operation.
\item in P, if 
$\Gamma$ is preserved by $s_{a,2}$ and $g_{a,2}$ for some $a\in\{0,1\}$.
\item in P, if
$\Gamma$ is preserved by $f_{a,2}$ for some $a\in\{0,1\}$.
\item co-NP-complete, otherwise.
\end{enumerate}
\end{thm}

\begin{proof}
Let $r$ be the maximal arity of relations in $\Gamma$, 
let $k=max(4,r)$. 
Note that if we add to $\Gamma$ all the relations of arity at most $k$
that can be pp-defined from $\Gamma$ then  we do not change the complexity of $\QCSP(\Gamma)$.
Also, adding the equality relation to $\Gamma$ does not change the complexity since 
the equalities can be propagated out.
Another important fact is that adding pp-definable relations and the equality relation to a constraint language does not affect the set of polymorphisms $\Pol(\Gamma)$.
Thus, it is sufficient to prove the claim only for constraint languages $\Gamma$ containing all such relations.
Additionally, 
since $\Pol(\Gamma)$ is $\{0,2\}\{1,2\}$-projective, 
$\{0,2\}$, $\{1,2\}$, 
and 
$\{0,2\}^{2}\cup\{1,2\}^{2}$ are from $\Gamma$.
Therefore, $\Gamma$ satisfies all the assumptions 
we formulated in the beginning of this section.

Case 1 follows from Corollary~\ref{PSpaceCompleteness}.
Otherwise, 
$\Gamma$ is preserved by a $01$-stable operation, and 
by Lemma~\ref{InCONPForStable} the problem 
$\QCSP(\Gamma)$ is in co-NP.

If $\Gamma$ is not preserved by a $0$-stable or $1$-stable operation, 
then  by Lemma~\ref{lem:coNPHard-nostable} $\QCSP(\Gamma)$ is co-NP-complete.

Then we may assume that 
$\Gamma$ is preserved by a 0-stable operation (the 1-stable case can be considered in the same way).
If $\Gamma$ is preserved by $f_{0,2}$ 
then 
by Corollary~\ref{CorComplexityStrange2} $\QCSP(\Gamma)$ is in P.
Similarly, 
if $\Gamma$ is preserved by $s_{0,2}$ and $g_{0,2}$ 
then by Corollary~\ref{CorComplexityStrange1} 
$\QCSP(\Gamma)$ is in P.
Thus we assume that $f_{0,2}$ does not preserve $\Gamma$, 
and $s_{0,2}$ or $g_{0,2}$ does not preserve $\Gamma$.

By Theorem~\ref{ANDTypeForStable},
$\Gamma$ contains an 
AND-type relation $\delta$.
We can check that $\delta_{1}(x,y) = \delta(x,x,y)$ is equal to $\begin{pmatrix}
0&1&2&2&2\\
0&1&0&1&2
\end{pmatrix}$
 (to derive $(2,1)$ we apply a 0-stable operation to $(1,1)$ and $(2,0)$).

If $\Gamma$ is also preserved by a 1-stable operation, 
then by Theorem~\ref{ORTypeForStable}, 
$\Gamma$ contains an OR-type relation,
and 
by Lemma~\ref{ANDORHardness}, $\QCSP(\Gamma)$ is co-NP-complete.

Thus, we assume that $\Gamma$ is not preserved by a 1-stable operation.
By Lemma~\ref{NoStableOperationOne},
$\Gamma$ contains a relation
$\delta_{2}\in\left\{\begin{pmatrix}
0 & 1 & 2 & 1 & 2\\
0 & 0 & 0 & 2 & 2 
\end{pmatrix},
\begin{pmatrix}
0 & 1 & 2 & 1 & 2\\
1 & 1 & 1 & 2 & 2 
\end{pmatrix}\right\}$.
Define 
$$\delta_{3}(x,y)= \exists x'\exists y'\;\delta_{2}(x,x')\wedge
\delta(x,x',y')\wedge \delta_{1}(y',y).$$
We can check that $(0,a)\in\delta_{3}$ implies $a=0$.
Applying $s_2$ to the tuples 
$(1,1,1)$,$(1,0,0)$, 
and to the tuples 
$(1,1,1)$,$(0,0,0)$, 
we obtain that 
$(1,2,2),(2,2,2)\in\delta$.
Then 
$(1,a),(2,a)\in\delta_{3}$ for every $a$
(put $x'=y' = 2$ in the definition of $\delta_{3}$),
which means that
$\delta_{3}=\begin{pmatrix}
0 & 1 & 1 & 1 & 2 & 2 & 2\\
0 & 0 & 1 & 2 & 0 & 1 & 2 
\end{pmatrix}$.

Assume that a 0-stable operation we have is 
$h_{0}$.
Since $h_{0}$ is $\{0,2\}\{1,2\}$-projective, 
it returns the first variable or 2.
If $h_{0}(0,1) = 0$, then $h_{0}$ is also 1-stable operation, which contradicts our assumptions.
Thus $h_{0}(0,1) = 2$.
If $h_{0}(0,2) =0$ then we get a contradiction 
with the fact that $h_{0}$ preserves $\delta_{3}$.
Thus $h_{0} = s_{0,2}$, which means that $\Gamma$ is not preserved by $g_{0,2}$.
By Lemma~\ref{StrangeOneImplies}, 
$\Gamma$ contains 
a relation 
$\delta_{4}\in \left\{\begin{pmatrix}
0&2\\
1&2
\end{pmatrix},
\begin{pmatrix}
0&2&2\\
1&1&2
\end{pmatrix}\right\}$.
Define 
$$\delta_{5}(x,y) = \exists 
z_{1}\exists z_{2}\;
\delta_{3}(x,z_{1})\wedge \delta_{4}(z_{1},z_{2})\wedge \delta_1(z_{2},y).$$
%
We can check that 
$\delta_{5}=\begin{pmatrix}
0&1&1&1&2&2&2\\
1&0&1&2&0&1&2
\end{pmatrix}\in\Gamma$.
Applying Lemma~\ref{SecondPPDefinition} to $\delta_{5}$, $\delta_{3}$, and $\delta_{1}$, 
we derive that the required relation $(x\neq 0)\vee (y=z)$ can be pp-defined from $\Gamma$.
Then by Lemma~\ref{NoStrange2Hardness},
$\QCSP(\Gamma)$ is co-NP-complete.
\end{proof}
\section{Main result for 3-element domain: proof of 
Theorem~\ref{ThreeElementClassification}}

\begin{lem}\label{NOWNUComplexity}
Suppose $\Gamma$ is a finite constraint language on $\{0,1,2\}$ with constants, such that $\Pol(\Gamma)$ has the EGP property and has no WNU.
Then $\QCSP(\Gamma)$ is PSpace-complete.
\end{lem}
\begin{proof}
Since $\Pol(\Gamma)$ has the EGP property,
every operation of $\Pol(\Gamma)$ is $\alpha\beta$-projective 
for some $\alpha$ and $\beta$, strict subsets of $\{0,1,2\}$, so that $\alpha \cup \beta = \{0,1,2\}$.
If $\alpha\cap\beta = \varnothing$,
then $\alpha\times\alpha\cup\beta\times\beta$ 
is an equivalence relation preserved by 
$\Pol(\Gamma)$ and any operation of $\Pol(\Gamma)$
modulo this congruence is a projection. 
It follows therefore from Lemma 5 in \cite{QCSPmonoids} and Theorem 5.2 of \cite{BBCJK} that
$\QCSP(\Gamma)$ is PSpace-complete.

Otherwise, w.l.o.g. we assume that $\alpha=\{0,2\}$ and $\beta=\{1,2\}$. 
Since there does not exist a WNU operation, it is known from \cite{CSPconjecture,MarotiMcKenzie}
that there should be a factor of size at least two whose operations are projections.

We consider two cases. There exists a congruence $\sigma$ such that 
$\Pol(\Gamma)/\sigma$ has only projections. This case we have already covered using Lemma 5 in \cite{QCSPmonoids} and Theorem 5.2 of \cite{BBCJK}.

Thus we may assume, there exists a subset $B\subsetneq \{0,1,2\}$ of size two
such that all operations on $B$ are projections
and $B$ is a subalgebra of $\Pol(\Gamma)$.

Assume that $B = \{0,2\}$. 
Then $\delta = \{0,2\}^{3}\setminus \{(2,2,0),(2,0,2)\}$ 
is preserved by $\Pol(\Gamma)$.
Recall the relations $\sigma = \alpha^{2}\cup\beta^{2}$ and 
$\sigma_{2}(x_{1},x_{2},x_{3},x_{4}) = \sigma(x_{1},x_{2})\vee\sigma(x_{3},x_{4})$ are pp-definable from $\Gamma$ because $\Pol(\Gamma)$ is $\alpha\beta$-projective.
We 
have $\sigma_{0}'(x,y,z) = \exists x'\;
\sigma(x,x')\wedge \delta(x',y,z)$,
$\sigma_{1}'(x,y,z) = \exists x'\;
\sigma_{2}(x,1,x',1)\wedge \delta(x',y,z)$, 
which by Lemma \ref{PSpaceHardness} guarantees PSpace-completeness of  $\QCSP(\Gamma)$.

If $B = \{1,2\}$, then we are in a symmetric case to the last and we again apply Lemma \ref{PSpaceHardness} (this time $b=1$) to prove PSpace-completeness.

If $B = \{0,1\}$, then we can complete the argument in various ways. Note that $B$ is a subalgebra, and we already know that $\{0,2\}$ and $\{1,2\}$ are subalgebras too. It follows that we are in the conservative case, and the result comes from Theorem \ref{thm:conservative}. 
\end{proof}

Note that the above lemma does not hold on a larger domain. 
As a counter example on the 4-element domain we can 
take the constraint language we build in Corollary~\ref{DPon4elements}.

Now, we are ready to prove Theorem~\ref{ThreeElementClassification}.

\begin{THMThreeElementClassification}
Suppose 
$\Gamma$ is a finite constraint language on $\{0,1,2\}$ with constants. Then $\QCSP(\Gamma)$ is 
\begin{enumerate}
\item in P, if $\Pol(\Gamma)$ has the PGP property and has a WNU operation.
\item NP-complete, if $\Pol(\Gamma)$ has the PGP property and has no WNU operation.
\item PSpace-complete, if $\Pol(\Gamma)$ has the EGP property and
has no WNU operation.
\item PSpace-complete, if $\Pol(\Gamma)$ has the EGP property and $\Pol(\Gamma)$ does not contain an $ab$-stable operation.
\item in P, if $\Pol(\Gamma)$ contains $s_{a,c}$ and $g_{a,c}$ for some $a,c\in\{0,1,2\}$, $a\neq c$.
\item in P, if $\Pol(\Gamma)$ contains $f_{a,c}$ for some $a,c\in\{0,1,2\}$, $a\neq c$.
\item co-NP-complete otherwise.
\end{enumerate}
\end{THMThreeElementClassification}

\begin{proof}
Suppose $\Pol(\Gamma)$ has the PGP property. We know from Theorem \ref{thm:PGP-in-NP} that $\QCSP(\Gamma)$ can be reduced to a polynomial collection of instances of $\CSP(\Gamma)$. If $\Gamma$ additionally has a WNU then it follows from \cite{BulatovFVConjecture,ZhukFVConjecture} that $\CSP(\Gamma)$, and therefore also $\QCSP(\Gamma)$, is in P. If $\Pol(\Gamma)$  does not have a WNU operation, then NP-hardness follows from the identity reduction from $\CSP(\Gamma)$, whose NP-hardness has long been known (\mbox{e.g.} see again \cite{BulatovFVConjecture,ZhukFVConjecture}).

Suppose $\Pol(\Gamma)$ has the EGP property.
If $\Pol(\Gamma)$ has no WNU operation 
then by Lemma~\ref{NOWNUComplexity}
$\QCSP(\Gamma)$ is PSpace-complete.

Suppose $\Pol(\Gamma)$ contains a WNU operation $w$.
Since $\Pol(\Gamma)$ has the EGP property, all operations of $\Pol(\Gamma)$ are $\alpha\beta$-projective for some 
$\alpha$ and $\beta$. 

If $\alpha\cap\beta = \varnothing$, then 
$\Pol(\Gamma)$ cannot have a WNU operation, 
which contradicts our assumption.

Suppose $\alpha\cap\beta \neq \varnothing$.
If $\Pol(\Gamma)$ contains 
an $ab$-stable operation $f$, that is 
$f(x,a,b) = x$ and $f(x,c,c)=c$ for $\{a,b,c\}=\{0,1,2\}$, 
then, since 
$f$ is $\alpha\beta$-projective, we have
$c \in\alpha\cap\beta$.
Similarly, if 
$\Pol(\Gamma)$ contains $s_{a,c},g_{a,c}$, or $f_{a,c}$, then 
$c \in\alpha\cap\beta$.
Since 
both variables of the operation 
$f(x,y)=w(x,\ldots,x,y)$ should be $\alpha\beta$-projective, 
$f$ is the semilattice operation $s_c$. W.l.o.g. we assume that $c=2$, 
then the remaining part of the classification 
follows from Theorem~\ref{SpecialClassification}.
\end{proof}


\section{Conclusion}

Our demonstration of QCSP monsters suggests that a complete complexity classification of $\QCSP(\Gamma)$ under polynomial reductions is likely to be exceedingly challenging. Indeed, suppose $\mathrm{P} \neq \NP$, how many equivalence classes of problems $\QCSP(\Gamma)$ are there up to polynomial equivalence? In this paper we showed that there are at least six of them. Are there any more? Are there infinitely many? We don't know
the answer.

Meanwhile, the most sensible approach to complexity classification for $\QCSP(\Gamma)$ might be to try to find those that are in P, in contradistinction to those that are $\NP$-hard under polynomial Turing reductions (which would thus capture also the $\coNP$-hardness). 
Similarly, someone could ask about a general 
criteria for the QCSP to be PSpace-hard
or to be a member of a concrete complexity class, 
which is also a very intriguing question.

As the next step, it seems very natural to work on a classification for 
constraint languages on a three-element domain 
without constants, where the reduction to CSP doesn't work and brand new ideas are required.

Even though we consider only finite constraint languages in this paper, we believe our classification of the complexity for a three-element domain holds for infinite constraint languages (relations are given by the list of tuples). To prove this stronger result 
we would need to add many technicalities, which would make 
our current proof even more complicated. 
For instance, for the algorithm 
in Section~\ref{StrangeTwoSection} to work in polynomial time we need to enforce that any constraint relation contains exponentially many tuples on 
its arity (see line 19 of the pseudocode). This can be achieved by 
showing that any relation preserved by the semilattice can be decomposed in polynomial time into several relations with exponentially many tuples
(see \cite{zhuk2018modification} for such a decomposition).
That is why we decided to limit ourselves to finite constraint languages.

\section*{Acknowledgements}

The first proof that there is a concrete finite constraint language $\Gamma$ such that QCSP$(\Gamma)$ is co-NP-complete is due to Miroslav Ol\u{s}\'ak. The observation that $\QCSP(\Gamma)$ as in Corollary \ref{cor:emil} is $\Theta^{\mathrm{P}}_2$-complete was communicated to us by Emil Je\v r\'abek. We are grateful for discussions with Victor Lagerkvist about polynomial-sized pp-definitions, with particular reference to our relations $\tau_i$.
We would like also to thank 
Antoine Mottet and Libor Barto for fruitful discussions, as well as discerning comments from several anonymous reviews.
  
\bibliographystyle{plain}
\bibliography{refs}

\end{document}